\documentclass[11pt]{article}
\usepackage{amsthm}
\usepackage{amsmath}
\usepackage{comment}
\usepackage{dsfont}
\usepackage{amssymb} 
\usepackage{graphicx}
\usepackage{tikz}
\usepackage{cases}
\usetikzlibrary{arrows,positioning,fit,shapes,calc}
\usepackage{hyperref}
\definecolor{darkgreen}{rgb}{0.1,0.5,0.1}
\hypersetup{colorlinks,linkcolor={blue},citecolor={darkgreen},urlcolor={blue}} 

\usepackage{authblk}

\newtheorem{Thm}{Theorem}
\newtheorem{Lem}[Thm]{Lemma}
\newtheorem{Def}[Thm]{Definition}
\newtheorem{Pro}[Thm]{Proposition}
\newtheorem{Cor}[Thm]{Corollary}
\newtheorem{Rem}[Thm]{Remark}

\newcommand\eq[1] {(\ref{#1})}

\newcommand{\bfm}[1]{\mbox{\boldmath ${#1}$}}

\newcommand{\beqa}{\begin{eqnarray}}
\newcommand{\eeqa}[1]{\label{#1}\end{eqnarray}}
\newcommand{\beq}{\begin{equation}}
\newcommand{\eeq}[1]{\label{#1}\end{equation}}
\newcommand{\Grad}{\nabla}

\newcommand{\Real}{\mathop{\rm Re}\nolimits}
\newcommand{\Imag}{\mathop{\rm Im}\nolimits}

\newcommand{\Gve}{\varepsilon}

\newcommand{\Gg}{\gamma}

\newcommand{\Go}{\omega}

\newcommand{\GD}{\Delta}

\newcommand{\BGve}{\bfm\varepsilon}

\newcommand{\BGG}{\bfm\Gamma}

\newcommand{\md}{{\mathrm{ d}}}
\newcommand{\mm}{{\mathrm{ m}}}

\newcommand{\bbR}{{\mathbb{ R}}}
\newcommand{\bbC}{{\mathbb{C}}}

\newcommand{\CM}{{\cal M}}

\newcommand{\CO}{{\cal O}}

\def\ii{{\rm i}}

\def\Ba{{\bf a}}

\def\Be{{\bf e}}

\def\Bn{{\bf n}}

\def\Bu{{\bf u}}
\def\Bv{{\bf v}}

\def\Bx{{\bf x}}

\def\BA{{\bf A}}

\def\BD{{\bf D}}
\def\BE{{\bf E}}
\def\BF{{\bf F}}

\def\BI{{\bf I}}
\def\BJ{{\bf J}}

\def\BL{{\bf L}}

\def\BM{{\bf M}}

\def\BU{{\bf U}}

\def\BW{{\bf W}}

\def \ba {\begin{array}}
\def \ea {\end{array}}
\DeclareMathOperator*{\ho}{\overset{\perp}{\mathcal{\oplus}}}

\newcommand{\bU}{{\bf U}}

\title{Broadband quasistatic passive cloaking: bounds and limitations in the near-field regime}

\author{Maxence Cassier$\footnote{Email: maxence.cassier@fresnel.fr}$, Graeme W.\ Milton$\footnote{Email: graeme.milton@utah.edu}$,  and  Aaron Welters$\footnote{Email: awelters@fit.edu}$ \\
{
 \small $^*$ Aix Marseille Univ, CNRS,  Centrale Med, Institut Fresnel, Marseille, France;} \vspace{0.05cm} \\ 
{\small  \quad \quad    $^\dagger$ Department of Mathematics, University of Utah, Salt Lake City UT 84112, USA;} \vspace{0.05cm} \\ 
{\small  \qquad \quad  \quad  $^\ddag$ Department of Mathematics and Systems Engineering,  Florida Institute of Technology, \\  \hspace{-7.5cm} Melbourne, FL 32901, USA}
}

\date{}

\begin{document}
\maketitle
\vspace{-1cm}
\begin{abstract}
 We consider here several aspects of the following challenging question: is it possible to use a passive cloak to make invisible a dielectric inclusion on a finite frequency interval in the quasistatic regime of Maxwell’s equations for an observer close to the object? In this work, by considering the Dirichlet-to-Neumann (DtN) map, we not only answer negatively this question, but we go further and provide some quantitative bounds on this map that provide fundamental limits to both cloaking as well as approximate cloaking. These bounds involve the following physical parameters: the length and center of the frequency interval, the volume of the cloaking device, the volume of the obstacle, and the relative permittivity of the object. Our approach is based on two key tools: i) variational principles 
 from the abstract theory of composites and
 ii) the analytic approach to deriving bounds from sum rules for passive systems.  
 To use i), we prove a new representation theorem for the DtN map which allows us to interpret this map as an effective operator in the abstract theory of composites. One important consequence of this representation is that it allows one to incorporate the broad and deep results from the theory of composites, such as variational principles, and to apply the bounds derived from them to the DtN map. These results could be useful in other contexts other than cloaking. Next, to use ii),
 we show that the passivity assumption allows us to 
 connect the DtN map (as function of the frequency) with two important classes of analytic functions, namely, Herglotz and Stieltjes functions.
The sum rules for these functions, combined with the variational approach, allows us to derive new inequalities on the DtN map which impose fundamental limitations on passive cloaking, both exact and approximate, over a frequency interval. We consider both cases of lossy and lossless cloaks.
\end{abstract}
\vspace{0.2cm}
{\noindent \bf Keywords:} Invisibility, passive cloaking, dispersive Maxwell’s equations, quasistatics, Dirich\-let-to-Neumann operators, abstract theory of composites, effective operators, variational principles, Herglotz and Stieltjes functions, sum rules.

\newpage

\tableofcontents

\section{Introduction}

\subsection{Motivations and state of the art}

Invisibility and cloaking have captured the imagination of people for countless years. From a scientific viewpoint invisibility means undetectable by appropriate probing fields, rather than just being invisible to our eyes as is much of the
electromagnetic spectrum. One of the earliest results concerning invisibility of bodies is that of Dolin \cite{Dolin:1961:PCT} who in 1961 discovered what is now called transformation optics and used it to construct
inclusions that would be invisible to any applied electromagnetic field oscillating in time at a given frequency. Kerker \cite{Kerker:1975:IB} in 1975 noted that coated ellipsoids would leave virtually undisturbed an electromagnetic wave in the surrounding medium if the coated ellipsoids are small compared to spatial variations in the exterior electromagnetic wave (i.e., in the quasistatic limit with a spatially uniform applied field) and
if the electrical permittivities and permeabilities of the core, shell, and surrounding medium are appropriately chosen. Such inclusions are called neutral coated inclusions, a terminology stemming from that of
Mansfield \cite{Mansfield:1953:NHP} who in 1953 found that certain reinforced holes, which he called neutral holes, could be cut out of a uniformly stressed plate without disturbing that surrounding stress: see \cite[Sec.\ 7.11]{Milton:2022:TOC} for more results concerning neutral coated inclusions. In two dimensions, and in the quasistatic limit, coated cylinders with core, coating, 
and matrix having electrical permittivities in the ratio 1:-1:1, constituting what is now called a poor man's superlens, were found \cite{Nicorovici:1994:ODP} to be invisible to non-uniform fields in the exterior medium. Going beyond the quasistatic limit, Alu and Engheta \cite{Alu:2005:ATP} in 2005 
discovered that coated spheres with appropriate moduli could be
invisible, in the sense that their total scattering cross section could be zero at one frequency. 
The interior of each of these coated inclusions can be considered to be cloaked by the coating in the limited sense that the coating makes it invisible.

In the first example of true cloaking, Greenleaf, Lassas, 
and Uhlmann \cite{Greenleaf:2003:ACC, Greenleaf:2003:NCI} in 2003
combined transformation conductivity with a singular transformation that maps a sphere minus the origin to an annulus to obtain a cloak which would be invisible to any exterior applied electric field and would not lose this invisibility when conducting objects were placed inside the annulus. In 2006, 
Milton and Nicorovici \cite{Milton:2006:CEA} discovered cloaking due to anomalous resonance where clusters of polarizable
dipoles or dipolar energy sources would be essentially invisible if placed in a specific region near a superlens \cite{Pendry:2000:NRM}, the essential mechanism for which was discovered in \cite{Nicorovici:1994:ODP}. Shortly afterwards, 
Pendry, Schurig, and Smith \cite{Pendry:2006:CEM}
used transformation optics with a singular transformation, again mapping a sphere minus the origin to an annulus, to obtain cloaking for any applied electromagnetic field oscillating with fixed frequency. Independently, and at the same time, 
Leonhardt \cite{Leonhardt:2006:OCM} used transformations in the geometric optics limit to obtain cloaking. A barrier to implementing transformation optics applied to cloaking, or even to the earlier work of Dolin, is that away from the geometric optics limit it ideally requires the relative magnetic permeability
to coincide with the relative electric permittivity, be anisotropic, and be tailored to the prescription demanded by a suitable transformation. Furthermore, in general, the prescriptions for transformation based cloaking demand singular anisotropies which makes their realization even more difficult: exceptions are so called carpet cloaking \cite{Li:2008:HUC} and
non-Euclidean cloaking \cite{Leonhardt:2009:BIN}. Fortunately, cloaking still holds to an arbitrarily high degree of approximation if one slightly perturbs the transformation to make it non-singular, thus having less extreme anisotropies \cite{Kohn:2008:CCV, Kohn:2010:CCV, Liu:2013:ENC,Bao:2014:NCF}. 
An approximation to the cloak of Pendry, Schurig, and Smith was experimentally tested \cite{Schurig:2006:MEC}
giving results in good agreement with
numerical simulations,
but still having significant scattering. 
The approximation was subsequently improved \cite{Cai:2007:NMC} by eliminating reflections at the outer surface of the cloak. Another interesting type of cloak, also related to the superlens, was suggested by Pendry and Ramakrishna \cite{Pendry:2003:FLN} 
in 2003 and further explored by Lai et al.\ \cite{Lai:2009:CMI}: basically the scattering of an object outside the superlens is canceled 
by a nearby antiobject
in the superlens, having the permittivity and permeability as the object, but with their signs flipped. Thus, this type of cloak needs to be tailored to object to be
cloaked. An extension is illusion optics \cite{Lai:2009:IOO}, 
where the scattering of an object can be made to mimic that of a different chosen object. 
There followed an explosion of interest including acoustic cloaking \cite{Cummer:2007:PAC, Norris:2008:ACT, Norris:2011:ECT,Craster:2013:AMN}, thermal cloaking \cite{Schittny:2013:ETT}, cloaking for linear elasticity \cite{Craster:2021:ONC}, and seismic cloaking \cite{Brule:2014:ESM}. Finally, in the context of acoustic  waveguides, we want to 
point out the interesting approach \cite{Chesnel:2022}  to make an object invisible   at a fixed frequency  in the far field regime by perturbing the boundary of the waveguides by  thin outer resonators. The main advantage of this method is that invisibility is achieved solely through geometric modifications of the waveguide's boundary, leaving the material properties unchanged.  However, the required geometric boundary perturbation depends strongly on the shape of the obstacle.

For the purpose of many applications, it is highly desirable to design cloaks that achieve perfect cloaking over a frequency interval of interest. 
However, for physical reasons there are fundamental issues which have been put forth to try to explain why perfect broadband passive cloaking should be impossible and also the challenges to achieve approximate cloaking. First, its was pointed out by  Pendry, Schurig, and Smith \cite{Pendry:2006:CEM} and Miller \cite{Miller:2006:PC} that causality and the finite speed of propagation of a wavepacket in passive media (see, for instance, \cite{Landau:1984:ECM, Welters:2014:SLL, Cassier:2017:MMD}) prohibits perfect cloaking of an object in the time domain. These prescriptions for transformation based cloaking only hold at isolated frequencies due to the dispersion of the constituent material properties. Despite these results, it is hard to turn these restrictions in the time domain into quantitative restrictions in the frequency domain.  

One approach to avoiding these restrictions is to use active cloaking. However, unlike passive cloaking, this requires energy provided directly by external sources or indirectly by active media (such gain media). The first type of active cloak was introduced by Miller \cite{Miller:2006:PC} and further developed by several groups, e.g., \cite{Vasquez:2009:AEC, Vasquez:2009:BEC, Norris:2014:AEC,cassier2021active,cassier2022active}, where sources are tailored to the applied field and which do not contribute to the field outside a certain distance while
creating a quiet zone where the field is zero and within
which one can place the objects to be cloaked. As the construction is independent of frequency it holds for multiple frequencies or for a continuum of frequencies. One may also tailor the cloak to the object one wishes to cloak and this is particularly appropriate for the plate equation as then the displacement field remains finite \cite{ONeill:2015:ACI}. A second class of active cloaks, called fast-light cloaks \cite{Tsakmakidis:2019:U3D}, arises with the recent interest in physics on active media, for instance, see \cite{Nistad:2008:CEA, Cham:2015:TTP, ElGanainy:2018:NHP, Bender:2019:PTS}. These fast-light cloaks use active media to overcome the speed-of-light limitations in passive electromagnetic systems, but are still subject to constraints due to causality and stability \cite{Abdelrahman:2021:PLB, Duggan:2022:SBS}, 
a topic which lies outside the scope of our paper. However, due to the energy requirements of these active cloaks, there is strong motivation to use passive media to achieve broadband cloaking.  We focus here on  passive electromagnetism cloaking due to properties specific to the dielectric properties of materials in the high-frequency limit.

Approximate broadband passive cloaking has recently been considered by methods of optimal design \cite{Jelinek:2021:FBP}, \cite{Strekha:2024:LBI}. These groups studied also the problem of bounds and limitations to perfect and approximate passive cloaking using the method of semidefinite relaxation of quadratically constrained quadratic programs (QCQPs) (whose general mathematical framework is reviewed in \cite{Luo:2010:SRQ}) which has had quite a few recent physics applications especially in electromagnetism, based on some general techniques laid out in, e.g., \cite{Kuang:2020:CBL, Gustafsson:2020:UBA, Chao:2022:PLE}, see also \cite{Angeris:2021:HMP, Gertler:2025:MPD} and references therein. However, these precise numerical bounds are not directly expressed in terms of physical parameters of the system.
Within this perspective, it is essential to establish fundamental physical limits that depend on the system’s key parameters such as central frequency of the frequency interval, the frequency bandwidth, the cloak and object volumes, and the relative  permittivity of the object in order to gain a clear understanding of what is theoretically achievable.
Indeed, a key limitation to passive cloaking is a consequence of passivity, it implies causality but also bounds on the speed-of-light of waves propagating in the medium (see, e.g., \cite{Welters:2014:SLL,Cassier:2017:MMD})
and this leads to constraints, due to analyticity, that restrict perfect cloaking to transparent materials at a single frequency. Moreover, dispersion and dissipation along with relativistic causality further constrains even approximate cloaking, especially when the object to cloak or bandwidth is not small.

In this context, Craeye and Bhattacharya \cite{Craeye:2012:RTB} established first an upper limit on the frequency bandwidth within which cloaking could be realized for two-dimensional systems, based on the group delay of electromagnetic wave packets. However, their conclusions rely on strong assumptions, particularly regarding the geometry of the cloak and the specific cloaking method used. Later, Monticone and Alù \cite{Mon:2014:PBE} demonstrated that passive cloaking cannot be extended across the entire frequency spectrum by deriving a general bound on the scattering cross-section. In a follow-up study \cite{Mon:2016:IEP}, they employed analogies with electrical circuits to establish frequency-dependent limits on the scattering cross-section, applying to both planar structures and three-dimensional objects with spherical symmetry. Additionally, Hashemi, Qiu, McCauley, Joannopoulos, and Johnson \cite{Hashemi:2012:DBP} showed that, for transformation-based cloaks, the possibility of broadband passive cloaking is fundamentally constrained by the physical size of the object being cloaked. It is also worth noting that, even when perfect cloaking over a finite bandwidth is unattainable, techniques exist to significantly reduce an object's scattering signature within a targeted frequency range, as explored in \cite{Chen:2007:EBE, Kallos:2011:CRC}.

In contrast, as in \cite{Cassier:2017:BHF}, the bounds derived here offer the significant advantage of being independent of both the object and cloak geometry, as well as the specific dispersive properties of the cloak. Notably, the cloaked object can lie entirely outside the cloak (see Remark \ref{rem.cloakotuside} for more details). These bounds explicitly account for the bandwidth size, and while they are limited to the quasistatic regime, they are applicable to a wide range of passive cloaking methods, including those based on anomalous resonance, transformation optics, and complementary media.

Most of the work cited above concern a far-field regime (including numerical bounds via the QCQP approach), there has been significantly less focus on the limitations to cloaking for the near field cloaking problem, which is very important for applications.
In particular, compared to the previous work \cite{Cassier:2017:BHF} by the first two authors, M. Cassier and G. W. Milton, which addressed the far-field cloaking problem in the quasistatic regime by deriving bounds on the polarizability tensor of the cloaking device, the present work focuses on the near-field cloaking problem in the same regime. More precisely, we derive bounds on the Dirichlet-to-Neumann (DtN) operator associated with the system over a given frequency interval. The DtN operator maps the voltage around the boundary of the body $\Omega$ under consideration (see Fig.\ \ref{fig.med}) to the outward normal component of the electric displacement field (see Sec.\ \ref{subsec:MathFormulationCloakingProbRegime} for the precise formalization). As we now have access to near-field information, the bounds we derive here are more explicit and involve more physical parameters than those obtained in \cite{Cassier:2017:BHF}. Moreover, unlike the polarizability tensor—which is a $3\times 3$ matrix-valued function—the DtN map is an operator-valued function defined on an infinite-dimensional space. This significantly complicates the analysis and necessitates the development of new tools. Our approach exploits deep connections between Herglotz functions and the abstract theory of composites \cite{Milton:2022:TOC, Milton:2016:ETC} to overcome these challenges.

\subsection{Synopsis of the main results  on fundamental limits of  passive cloaking}

For our setting, we consider any open bounded set $\Omega\subseteq \mathbb{R}^3$ that is simply connected with sufficiently smooth boundary $\partial \Omega$ (e.g., Lipschitz continuous). In addition, any admissible potential $u$ is required to be a square-integrable function on $\Omega$ with a (weak) gradient $\nabla u$ that is also; in particular, this implies that it's boundary values $u|_{\Omega}=V_0$ (i.e., an admissible surface potential) is a square-integrable function on $\partial\Omega$ (see Sec.\ \ref{sec:NotationsConventions} for the precise definitions of the associated functional spaces). 
 \begin{figure}[!ht]
\centering
 \includegraphics[width=0.5\textwidth]{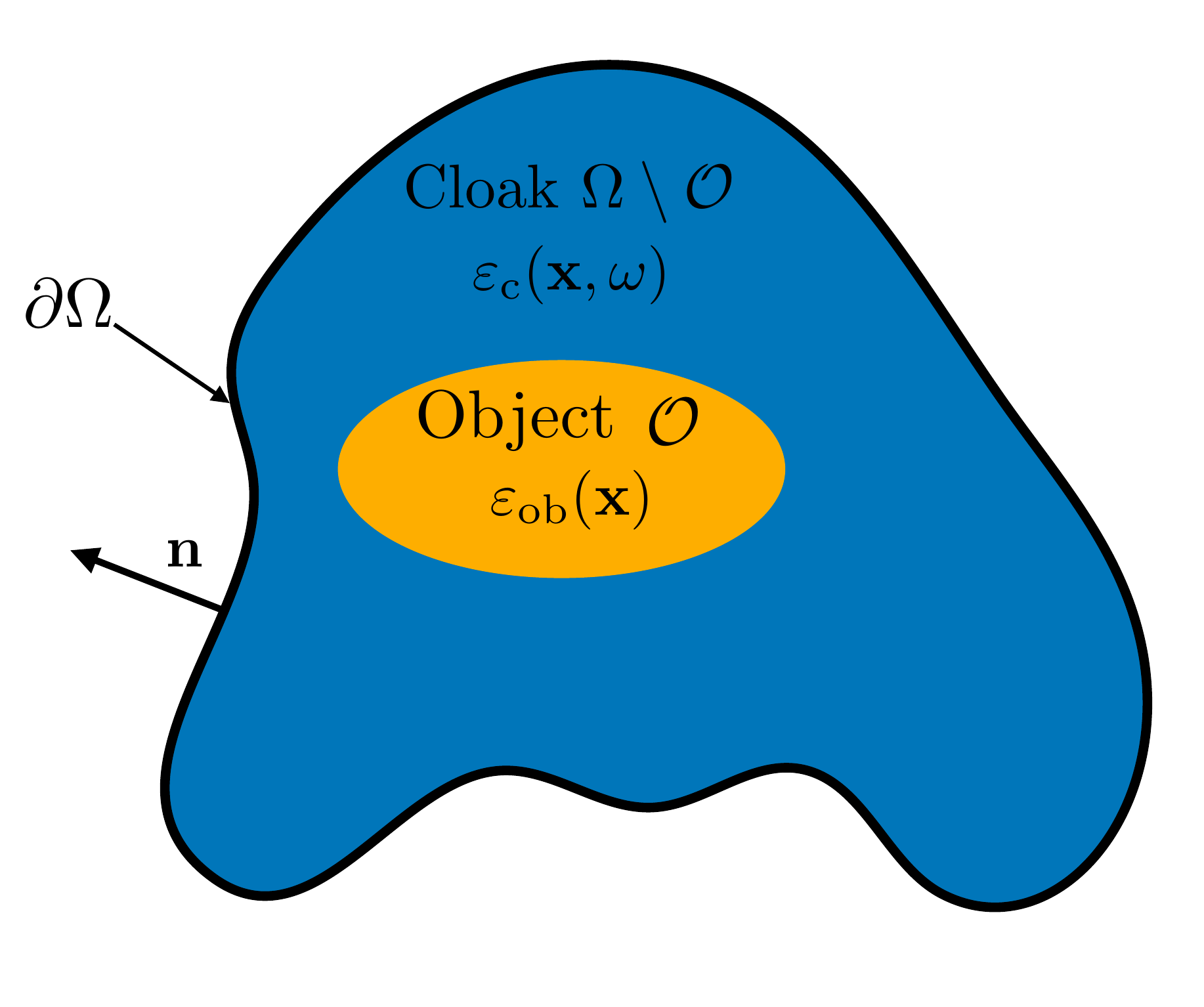}
 \caption{The cloaking device fills an open, bounded, and simply connected Lipschitz domain $\Omega$ of $\mathbb{R}^3$. It is made of both the object (in yellow), a  Lebesgue measurable set  $\mathcal{O}$ having positive volume with permittivity $\BGve(\Bx,\Go)=\BGve_{\mathrm{ob}}(\Bx)$, and the  passive cloak (in blue) contained in the set $\Omega \setminus \mathcal{O}$ having positive volume with permittivity $\BGve(\Bx,\Go)=\BGve_{\mathrm{c}}(\Bx, \Go)$. Here $\partial \Omega$ denotes the boundary of $\Omega$ and $\mathbf{n}$ is the unit outward normal vector to $\partial \Omega$. We point out that, in contrast to the figure, the object $\mathcal{O}$ is allow to have more than one connected component (e.g., $\mathcal{O}$ can be made of several inclusions in $\Omega$).}
 \label{fig.med}
\end{figure}

The object to be cloaked lies in $\mathcal{O}\subseteq\Omega$ and the cloak lies in $\Omega\setminus \mathcal{O}$. Both have positive  volumes $|\mathcal{O}|, |\Omega\setminus \mathcal{O}|$, see Figure \ref{fig.med}.  The obstacle is assumed to have a non-dispersive permittivity $\BGve_{\mathrm{ob}}(\Bx), \Bx\in\mathcal{O}$ on a frequency range $[\Go_-,\Go_+]\subseteq (0,\infty)$ which is a real symmetric positive definite matrix satisfying $\BGve_{\mathrm{ob}}(\Bx)\geq \Gve \BI$ for some constant scalar $\Gve > \Gve_0$,  where  $\Gve_0>0$ denotes the permittivity of the vacuum. 
 The cloak is made  of passive material which is dispersive and possibly anisotropic and its permittivity  tensor  $\BGve_{\mathrm{c}}(\Bx,\Go), \Bx\in \Omega \setminus \mathcal{O}$ is considered for frequencies $\omega\in [\Go_-,\Go_+]$. 
However, as the material of the cloak is passive, $\BGve_c(\Bx,\Go)$ has to satisfy  some analytical properties for complex frequencies $\omega$ in  the open  upper-half plane $\bbC^+$ of the complex plane $\bbC$  (see Secs.\ \ref{subsec:SettingCloakingProb} and \ref{subsec:MainAssumpOnPermittivity} for the explicitly stated permittivity assumptions). Thus, the permittivity  $\BGve(\Bx,\omega), \Bx\in\Omega$ of the whole device is defined for any  $\Go\in  [\Go_-, \Go_+]\cup \bbC^+$ by
$$
  \BGve(\Bx,\omega)= \BGve_{\mathrm{ob}}(\Bx) \mbox{ for  a.e. }  \Bx \in  \mathcal{O} \mbox{ and } \BGve(\Bx,\omega)=\BGve_c(\Bx,\Go) \mbox{ for  a.e. } \Bx \in \Omega \setminus \mathcal{O}.
$$ 

Given a electric (surface) potential $V_0$ on the boundary $\partial\Omega$, the quadratic forms $\langle \Lambda_{\Ba}V_0,\overline{V_0} \rangle$ of the DtN operators $\Lambda_{\Ba},$ where $\Ba\in\{\BGve(\cdot,\Go), \varepsilon_0\BI\}$ (which are defined precisely in Sec. \ref{subsec:MathFormulationCloakingProbRegime}), are related by the Green's formula [see Eq.\ \eqref{eq:GreensFormula} for a precise statement of this formula] to the electric field $\BE_{\Ba}=-\Grad V_{\Ba}$, with electric potential $V_{\Ba}$ in $\Omega$ with boundary-value $V_{\Ba}|_{\partial\Omega}=V_0$, and (divergence-free) electric displacement field $\BD_{\Ba}=\Ba \BE_{\Ba}$ in the body $\Omega$ with permittivity $\Ba$, by
\begin{gather}
    \langle \Lambda_{\Ba}V_0,\overline{V_0} \rangle=\int_{\Omega}\BD_{\Ba}(\Bx)\cdot \overline{\BE_{\Ba}(\Bx)}\mathrm{d} \Bx.\label{quadformdtnmap}
\end{gather}
This quantity is measurable by the observer since $\Lambda_{\Ba}V_0=\BD_{\Ba}\cdot \Bn$ on $\partial\Omega$, which is the normal component of the displacement field on the boundary of the body associated to surface potential $V_0$ imposed by the observer on $\partial \Omega$ (cf.\ Fig.\ \ref{fig.cloackprob}).
Moreover, multiplying \eqref{quadformdtnmap} by the frequency $\omega$ and taking the imaginary part, one gets (up to a $1/2$ factor, see \cite[Sec.\ 6.9]{Jackson:1999:CE} and \cite[Sec.\ 11.1, Eq.\ (11.16)]{Milton:2022:TOC}) that the following  represents the electric losses due to absorption by the materials:
\begin{gather}
    \operatorname{Im}\langle\Lambda_{\omega\Ba}V_0,\overline{V_0} \rangle=\operatorname{Im}(\omega\langle\Lambda_{\Ba}V_0,\overline{V_0} \rangle)=\int_{\Omega}\operatorname{Im}[\omega\Ba(\Bx)\BE_{\Ba}(\Bx)\cdot \overline{\BE_{\Ba}(\Bx)}]\mathrm{d} \Bx.\label{ElecLossDueMaterialAbsorp}
\end{gather}
 Now to compare the DtN operators of the cloaking device $\Lambda_{\BGve(\cdot,\Go)}$ to that of empty space occupying $\Omega$ whose permittivity is that of the vacuum, i.e., $\Lambda_{\varepsilon_0}$ (where by convention $\Lambda_{\varepsilon_0}:=\Lambda_{\varepsilon_0\BI}$), we consider the difference of their quadratic forms and its limit along the positive imaginary axis, i.e.,
\begin{gather}
F_{V_0}(\Go):=\langle [\Lambda_{\BGve(\cdot,\Go)}-\Lambda_{\varepsilon_0}]V_0,\overline{V_0} \rangle,\quad\omega\in [\Go_-,\Go_+] \cup \bbC^+,\label{def.FHerglotz}\\
F_{V_0,\infty}:=\lim_{y\rightarrow +\infty}F_{V_0}(\textrm {i}y)=\lim_{\omega = \textrm{i}y, y\rightarrow +\infty}\frac{\operatorname{Im}[\omega F_{V_0}(\omega)]}{\operatorname{Im}\omega}\geq 0.
\end{gather}
Most importantly, it is the modulus of \eqref{def.FHerglotz} that gives a figure-of-merit on the quality of the cloaking device at a frequency $\omega\in [\omega_-,\omega_+]$. In addition, the existence and non-negativity of that limit $F_{V_0,\infty}$ is a consequence of the physical constraints imposed by the passivity of the cloaking device (for more on this, see Theorem \ref{thm.differenceHerg}). Moreover, as we show (in Cor.\ \ref{Cor.boundleadingcoefficient}) it is strictly positive for nonzero affine boundary conditions, i.e., $V_0=-{\bf e}_0\cdot \Bx|_{\partial\Omega}$, ${\bf e}_0\in \bbC^3\setminus\{0\}$, and satisfies the inequality
\begin{gather}
    F_{V_0, \infty}\geq |\mathcal{O}| \  \Big(1- \frac{\varepsilon_0} {\varepsilon} \Big) \ \    \frac{1}{|\mathcal{O}|\,  \frac{\varepsilon_0}{\varepsilon}+ |\Omega \setminus \mathcal{O}| } \ \varepsilon_0 \ |\Omega| \, \, \|{\bf e}_0\|_{\bbC^3}^2>0,\label{AffineBCPositiveLowerBound}
\end{gather}
where $\|\cdot\|_{\bbC^3}$ denotes the standard norm for $\bbC^3$.

The discussion that follows will give a synopsis of the main results of our paper. The first theorem says that perfect cloaking on a frequency interval is impossible (see Theorem \ref{Pro.perfectcloaking} for the precise statement of this theorem).
\begin{Thm}\label{Thm.main1}
    Under appropriate assumptions on the permittivity $\varepsilon(\Bx,\omega)$, the obstacle
cannot be perfectly cloaked on the frequency interval $[\omega_-,\omega_+]$. 
\end{Thm}
The idea of the proof is that as $F_{V_0}(\Go)$ is analytic on $[\Go_-,\Go_+] \cup \bbC^+$ then perfect cloaking on $[\omega_-,\omega_+]$ implies by analytic continuation that $F_{V_0}(\Go)=0$ on $[\Go_-,\Go_+] \cup \bbC^+$ implying $F_{V_0, \infty}=0$ for nonzero affine boundary condition $V_0$, a contradiction of \eqref{AffineBCPositiveLowerBound}. 

The next result also proves the above theorem and includes quantitative bounds that provide fundamental limitations on perfect cloaking as well as approximate cloaking on a frequency interval (see Theorem \ref{Th.boundtransp} and its Cor.\ \ref{Cor.lossless} and Theorem \ref{Th.lossy} and its Cor.\ \ref{Cor.boundlossy} for the precise statements that make up this theorem).

\begin{Thm}\label{Thm.main2}
Under appropriate assumptions on the permittivity $\varepsilon(\Bx,\omega)$ and when $F_{V_0,\infty}>0$ (e.g., nonzero affine boundary condition $V_0$), the following inequalities hold:
\begin{equation*}
0<\frac{1}{4}(\Go_+^2-\Go_-^2) F_{V_0,\infty} \leq \max_{\Go\in[\Go_-,\Go_+]}|\omega^2 F_{V_0}(\Go)|.
\end{equation*}
Suppose, in addition, the cloak is lossless on $[\Go_-,\Go_+]$. If $F_{V_0}(\Go_0)=0$ for some $\omega_0\in [\Go_-,\Go_+]$ then
\begin{gather*}
F_{V_0}(\Go)\leq -\left(\frac{\omega_0^2-\omega^2}{\omega^2}\right)F_{V_0,\infty}<0,\;\text{ if } \omega_-\leq \omega< \omega_0,\\
0<\left(\frac{\omega^2-\omega_0^2}{\omega^2}\right)F_{V_0,\infty}\leq F_{V_0}(\Go),\;\text{ if } \omega_0< \omega\leq \omega_+.
\end{gather*}
More generally, if the lossless cloak achieves approximate  cloaking at some $\Go_0\in [\Go_-, \Go_+]$, i.e., there exists $\eta>0$ such that for all admissible surface potentials $\tilde{V}_0$,
\begin{equation*}
|F_{\tilde{V}_0}(\Go_0)|\leq \eta \, G^{\operatorname{vac}}_{\tilde{V}_0} \ \mbox{ where } \ G^{\operatorname{vac}}_{\tilde{V}_0}=\langle \Lambda_{\varepsilon_0} \tilde{V}_0,\overline{\tilde{V}_0} \rangle,
\end{equation*}
then
\begin{eqnarray*}
 \displaystyle F_{V_0}(\Go) & \leq &  \big(-F_{V_0,\infty}+\eta \, G^{\operatorname{vac}}_{V_0} \big) \ \frac{\Go^2_0-\Go^2}{\Go^2} + \eta \,  G^{\operatorname{vac}}_{V_0}  \ \mbox{if}  \ \Go_-\leq  \Go\leq \Go_0, \\[10pt]
 \displaystyle  F_{V_0}(\Go) & \geq & \big(F_{V_0,\infty}+\eta  \, G^{\operatorname{vac}}_{V_0} \big) \ \frac{\Go^2-\Go^2_0}{\Go^2}- \eta \,  G^{\operatorname{vac}}_{V_0} \ \  \mbox{if} \   \Go_0 \leq \Go \leq \Go_+ .
\end{eqnarray*}
\end{Thm}

Proving this theorem is not straightforward and requires significantly more effort then that of the previous theorem. Thus, most of the paper is dedicated to the proof which is based on combining the following two techniques for producing bounds:

\begin{itemize}
    \item[i)] Variational bounds in abstract theory of composites (see, e.g., \cite{Milton:2016:ETC, Milton:2022:TOC,Beard:2023:EOVP, Welters:EOM:2026}). More precisely:
    \begin{itemize}
    \item[$\circ$] A functional framework is first developed based on a Hodge decomposition (see Theorem \ref{th.tpro}) which yields a new representation of the DtN map  represented in terms of an effective operator (see Theorem \ref{thm.effopreprDtNmap}). 
                \item[$\circ$]This new representation of the DtN operator allows one to apply variational bounds on effective operators from the abstract  theory of  composites. This yields, in particular, the bounds \eqref{AffineBCPositiveLowerBound}  for $F_{V_0,\infty}$ in terms of key physical parameters of the cloaking device: the volumes of the obstacle and the cloak (i.e., $|\mathcal{O}|, |\Omega\setminus\mathcal{O}|$) and  a lower bound of the relative permittivity   of the obstacle $\BGve_{\mathrm{ob}}(\Bx)/\varepsilon _0$ given by the ratio $\Gve/\Gve_0$. 
                
            \end{itemize}
    \item[ii)] Analytic approach to get frequency-dependent bounds via sum rules for passive systems (see, e.g., \cite{Bernland:2011:SRC, Cassier:2017:BHF}).
    \begin{itemize}
    \item[$\circ$] The functions $\omega\mapsto \omega \Lambda_{\BGve(\cdot,\Go)}=\Lambda_{\omega\BGve(\cdot,\Go)}$ and $\omega\mapsto \omega F_{V_0}(\omega)$ [in \eqref{def.FHerglotz}] are first shown to be Herglotz functions (see Propositions \ref{Pro:DtNMapQuadFormIsHerglotz} and \ref{pro.DtNMapIsHerglotzFunc}, and Theorem \ref{thm.differenceHerg}) and then the function $z\mapsto F_{V_0}(\sqrt{-z})$ is proven to be a Stieltjes function (see Lemma \ref{Const.Stielt}). 
    \item[$\circ$] This yields bounds, using a general approach via sum rules (see Sec.\ \ref{sec:SumRulesApprGenBounds}), on the function $\omega\mapsto F_{V_0}(\omega)$ in term of  $F_{V_0,\infty}$, the bandwidth $\omega_+-\omega_-,$ and the center frequency $(\omega_++\omega_-)/2$ (see Theorems \ref{Th.boundtransp} and \ref{Th.lossy}).
    \end{itemize}
\end{itemize}

Finally, we are also able to use the main results discussed above to provide additional insights on cloaking by either: a) assuming more regularity of the function $[\omega_-,\omega_+]\ni\omega\mapsto\BGve(\cdot,\omega)$ (see Cor.\ \ref{Cor.boundlossy} and Remark \ref{rem.EpsilMoreRegularity}) or b) allow dispersive obstacles, i.e., a frequency-dependent $\BGve_{\mathrm{ob}}(\Bx, \Go)$ (see Cor.\ \ref{cor:ApproxCloakingWDisperObj}).

\subsection{Outline of the paper}

The rest of the paper will proceed as follows. In Sec.\ \ref{sec:Preliminaries} we introduce our preliminaries. We begin with Subsec.\ \ref{sec:NotationsConventions} by recalling the necessary functional framework on Sobolev spaces for elliptic PDEs on bounded Lipschitz domains and then the required matrix/operator notations. Next, in Subsec.\ \ref{subsec:ReviewHerglotzStieltjesFuncs}, we review the properties of two main classes of analytic functions arising in this paper, \emph{Herglotz and Stieltjes functions}. 

In Sec.\ \ref{sec:FormulationCloakingProb}, we begin with Subsec.\ \ref{subsec:SettingCloakingProb}, where we describe the cloaking problem in the near-field regime for quasistatic electromagnetism over a finite-frequency interval on a simply-connected bounded Lipschitz domain. Next, in Subsec.\ \ref{subsec:MainAssumpOnPermittivity} we describe precisely the physical and mathematical assumptions (cf.\ our hypotheses H1--H6) associated with the constituent materials (i.e., their permittivity) related to the \emph{passivity of the cloaking device}. Then, in Subsec.\ \ref{subsec:MathFormulationCloakingProbRegime}, we define precisely the cloaking problem by introducing the Dirichlet problem for the elliptic PDE corresponding to the quasistatic regime and the associated DtN map, which represents the data accessible to the observer to probe the medium  from the domain boundary. This allows us to rigorously define cloaking as well as approximate cloaking (see Def.\ \ref{Def-perfectcloaking}) in terms of this bounded linear operator. 

In Sec.\ \ref{sec:EffOpReprDtNMap}, we show how to represent the DtN map as an effective operator in the abstract theory of composites. In order to do this, we begin in Subsec.\ \ref{subsec:HodgeDecompDtNMap}, with a discussion of the Hodge decomposition associated with the Dirichlet problem for Laplace's equation (see Theorem \ref{th.tpro}). Next, in Subsec.\ \ref{subsec:AbstrThyComposites} we introduce the basic formalism in the abstract theory of composites starting with the Hilbert space framework for defining an effective operator, based on a generalized notion of a Hodge decomposition and constitutive material relation that is encapsulated in a concept known as the $Z$-problem. Next, in Subsec.\ \ref{Dirichlet-Z-problem}, we define the ``Dirichlet $Z$-problem" (see Def.\ \ref{DefKZProb}) using the Hodge decomposition from Subsec.\ \ref{subsec:HodgeDecompDtNMap}. After this, in Subsec.\ \ref{sec.linkDtnsigmastar}, we prove the main result of this section (i.e., Theorem \ref{thm.effopreprDtNmap}), a new representation of the DtN map in terms of an effective operator for the Dirichlet $Z$-problem. 

In Sec.\ \ref{Sec:VarPrincBndsEffOps}, we begin by describing, in Subsec.\ \ref{subsec:VarPrinAbstrResults}, how the abstract framework allows one to develop systematically bounds on the effective operator for $Z$-problems in a similar manner as in the theory of composites, for instance, the two variational principles known as the Dirichlet and Thomson variational principles. Then, in Subsec.\ \ref{subsec:VarPrincAppDtNMap}, we apply this framework to develop the associated minimization variational principles for the effective operator of the Dirichlet $Z$-problem to produce our first set of elementary bounds on the DtN map. 

In Sec.\ \ref{sec:DiriPrbAffineBCsDtNMap}, we narrow the set of possible boundary conditions for our elliptic PDE, which defines the DtN map, to affine boundary conditions. First, in Subsec.\ \ref{subsec:DtNMapAffBndCondsEffOpRepr} we prove two theorems (Theorems \ref{thm.EffOpForDtNMapWithAffBCs} and \ref{thm.aDisaneffoperator}) that relates the DtN map with affine boundary conditions to a new $Z$-problem and associated effective operator. Then in Subsec.\ \ref{subsec:ElemBndsAvgLocalTensors} we show how this leads to a second set of elementary bounds on the DtN map for these affine boundary conditions. 

In Sec.\ \ref{sec:AnalyticPropertiesDtNMap}, we consider the analytic properties of the DtN map for passive systems.  First, in Subsec.\ \ref{sec:DiffHerglotzFuncsDtNMapsIsHerglotz}, we relate the quadratic form of the DtN map to a Herglotz function of frequency (see Proposition \ref{Pro:DtNMapQuadFormIsHerglotz}). Then we compare the quadratic forms of the cloak to the uncloaked device and using the representation formula for the DtN map as an effective operator, to prove an unintuitive result (see Theorem \ref{thm.differenceHerg}) that their difference is also a Herglotz function. Next, in Subsec.\ \ref{sec-Herg-Stielt} we recall the connection between the two main classes of analytic functions we consider, namely, Herglotz and Stieltjes functions. Then we exploit this connection in Subsec.\ \ref{sec-Herg-phys} to show how our physical assumptions for the cloaking problem lead to the difference of the DtN maps being, after a change of variables, a Stieltjes function with some very useful properties (see Theorem \ref{Thm.Hergphys}). The properties of this function will be the key to deriving sum rules and bounds associated to the cloaking problem in the next section.

Sec.\ \ref{sec:QuantiativeBndsPassiveCloakingFreqInterval} is dedicated to deriving the quantitative bounds and limitations to the cloaking problem. First, in Subsec.\ \ref{sec:PerfectCloakingCantOccur} we prove that perfect cloaking cannot occur on the frequency interval $[\omega_-,\omega_+]$ (see Theorem \ref{Pro.perfectcloaking}). Next, in Subsec.\ \ref{sec:SumRulesApprGenBounds} we recall from \cite{Bernland:2011:SRC, Cassier:2017:BHF} the sum rules approach to bounds on Herglotz and Stieltjes function. Then in Subsections \ref{sec-losslesscase}, \ref{sec-lossycase}, and \ref{sec:GenOurCloakingBnds} we apply this to our problem to derive fundamental bounds and limitations to cloaking under different assumptions on the cloaking device. In Subsec.\ \ref{sec-losslesscase}, under the lossless assumption, we show that if one can perfectly cloak or even approximately cloak at one frequency then it constraints cloaking on the considered frequency interval (see Theorem \ref{Th.boundtransp} and Cor.\ \ref{Cor.lossless}). On the other hand, in Subsec.\ \ref{sec-lossycase} we consider the more general case of a lossy cloaking device and derive inequalities that constraint cloaking in this setting (see Theorem \ref{Th.lossy} and Cor.\ \ref{Cor.boundlossy}). Finally, in Subsec.\ \ref{sec:GenOurCloakingBnds}, we generalize the results obtained in Subsec.\  \ref{sec-losslesscase} for a lossless cloak to the case of a dispersive obstacle (see Cor.\ \ref{cor:ApproxCloakingWDisperObj}).

We conclude our paper with Appendix \ref{sec:Appendix}. First, Subsec.\ \ref{sec-annexe} gives the proof of Prop.\ \ref{pro.DtNDirBVPResults}. Next, in Subsec.\ \ref{sec:StrConvInvOps} we prove a technical result on strong convergence of a sequence of invertible operators. Finally, in Subsec.\ \ref{sec:HergFuncBanachSpaces} we describe the functional framework for bounded linear operators acting between a Banach space and its dual space that justifies calling the DtN map an operator-valued Herglotz function. The merit of presenting this structure for this important operator is to allow for further investigations.

\section{Preliminaries}\label{sec:Preliminaries}
\subsection{Notations and  conventions}\label{sec:NotationsConventions}

Let $\Omega$ be a nonempty bounded connected open set of $\mathbb{R}^d$ ($d\geq 2$) with Lipschitz boundary (i.e., a bounded Lipschitz domain). We consider the complex Hilbert space
$$
\mathcal{H}=\BL^2(\Omega)=[L^2(\Omega)]^d
$$
with standard inner product $(\cdot,\cdot)_{\mathcal{H}}$ defined by
$$
(\Bu,\Bv)_{\mathcal{H}}=\int_{\Omega} \Bu(\Bx) \cdot \overline{\Bv(\Bx)} \,\mathrm{d} \Bx,  \ \forall \Bu, \, \Bv \in \mathcal{H}
$$
and  its associated Hilbert norm which we denote by $\|\cdot\|_{\mathcal{H}}$.
Next, we introduce the following classical Sobolev spaces associated with the gradient $\nabla$ and divergence $\nabla\cdot$ operators along with the Green's formula:
\begin{itemize}
\item $H^1(\Omega)=\left\{ u\in L^2(\Omega) \,\mid \,\nabla u \in \BL^2(\Omega)\right\}$
endowed with its Hilbertian norm:
$$
\left\|u\right\|_{H^{1}(\Omega)}^2= \left\|u\right\|_{L^2(\Omega)}^2+\left\|\nabla  u\right\|_{\BL^2(\Omega)}^2.
$$

\item $H^1_0(\Omega)=\{u \in H^{1}(\Omega) \mid u=0 \mbox{ on } \partial \Omega \}$.  Recall, $H^1_0(\Omega)$ is a closed subspace in $H^1(\Omega)$ (for the $\|\cdot\|_{H^{1}(\Omega)}$ norm). We denote by $H^{-1}(\Omega)$  its topological dual, i.e., $H^{-1}(\Omega)=(H^{1}_0(\Omega))^*$, 
and  $\langle  \cdot ,\cdot  \rangle_{H^{-1}(\Omega),H^1_0(\Omega)}$  the duality product between $H^{-1}(\Omega)$ and $H^1_0(\Omega)$.

\item $H_{div}(\Omega)=\left\{ \Bu\in \BL^2(\Omega) \mid \,\Grad \cdot \Bu \in \BL^2(\Omega)\right\}$
endowed with its Hilbertian norm:
$$
\left\|\Bu\right\|_{H_{div}(\Omega)}^2= \left\|\Bu\right\|_{\BL^2(\Omega)}^2+\left\|\nabla \cdot \Bu\right\|_{\BL^2(\Omega)}^2.
$$
 \item $H_{div,0}(\Omega)=\left\{ \Bu\in H_{div}(\Omega)  \mid \Bu\cdot \Bn=0 \mbox{ on } \partial \Omega \right\}.$ Recall,
$H_{div,0}(\Omega)$ is a closed subspace in $H_{div}(\Omega)$ (for the $\|\cdot \|_{H_{div}(\Omega)}$ norm).
 
\item $H_{harm}(\Omega)=\{u\in H^1(\Omega) \mid \Delta u=0 \}$.

\item The trace space:  $H^{\frac{1}{2}}(\partial \Omega)=\left\{ u|_{\partial \Omega} \mid u \in H^{1}(\Omega) \right\}$ endowed with Hilbertian norm (see, e.g., \cite[p.\ 8]{Girault:1986:FNS}):
\begin{gather*}
    \left\|V_0\right\|_{H^{\frac{1}{2}}(\partial\Omega)}=\inf_{u\in H^1(\Omega), u|_{\partial\Omega}=V_0}\left\|u\right\|_{H^{1}(\Omega)}.
\end{gather*}
Recall, the duality product $\langle \cdot,\cdot \rangle_{H^{-\frac{1}{2}}(\partial \Omega), H^{\frac{1}{2}}(\partial \Omega)}$ between $H^{\frac{1}{2}}(\partial \Omega)$ and it's topological dual $H^{-\frac{1}{2}}(\partial \Omega)=\big(H^{\frac{1}{2}}(\partial \Omega)\big)^*$ is an extension of the $L^2$-inner product on the boundary in the sense that 
 $$
\langle u,\overline{v} \rangle_{H^{-\frac{1}{2}}(\partial \Omega), H^{\frac{1}{2}}(\partial \Omega)}=\left( u,v \right)_{L^2(\partial \Omega)},\; \forall u\in L^2(\partial \Omega) \mbox{ and } \forall  v \in  H^{\frac{1}{2}}(\partial \Omega).
 $$
 Moreover, $H^{-\frac{1}{2}}(\partial \Omega)$ is endowed with  its natural dual norm: 
 $$
 \|u\|_{H^{-\frac{1}{2}}(\partial \Omega)}=\sup_{v\in H^{\frac{1}{2}}(\partial \Omega) \setminus \{ 0\} } \frac{ \big|\langle u,\overline{v} \rangle_{H^{-\frac{1}{2}}(\partial \Omega), H^{\frac{1}{2}}(\partial \Omega)}\big|}{\|v\|_{H^{\frac{1}{2}}(\partial \Omega)}}.
 $$
 \item The Green's formula holds (see, e.g., \cite[p.\ 28, Eq.\ (2.17)]{Girault:1986:FNS}): 
 For all $
        \Bv\in H_{div}(\Omega)$ and all $ u\in H^{1}(\Omega)$,
    \begin{gather}
         (\Bv,\Grad u)_{\mathcal{H}}+(\Grad \cdot \Bv, u)_{L^2(\Omega)}=\langle \gamma_{\Bn}\Bv,\overline{u|_{\partial\Omega} }\rangle_{H^{-\frac{1}{2}}(\partial \Omega), H^{\frac{1}{2}}(\partial \Omega)}, \label{eq:GreensFormula}
    \end{gather}
    where $\gamma_{\Bn}\Bv=\Bv\cdot\Bn$ on $\partial\Omega$ and $\Bn$ is the unit normal to $\partial \Omega$, oriented to the exterior of $\Omega$. The operator $\gamma_n$ is referred to as the normal trace operator. It is a bounded linear operator from $H_{div}(\Omega)$ to $H^{-\frac{1}{2}}(\partial \Omega)$. In addition, for a function $u\in H^1(\Omega)$ such that $\Delta u\in L^2(\Omega)$ [or, equivalently, such that $\nabla u\in H_{div}(\Omega)$], we introduce the standard notation $\frac{\partial u}{\partial \Bn}\in H^{-\frac{1}{2}}(\partial \Omega)$ for the normal trace of the gradient, i.e., $\frac{\partial u}{\partial \Bn}:=\Gg_{\Bn}(\nabla u)=\nabla u\cdot \Bn$ on $\partial \Omega$.
\end{itemize} 

\begin{Rem} 
For the cloaking problem we assume that $d = 3$, but we retain the notation $d$ since some important results in the paper (which also apply outside the cloaking setting) hold for dimensions $d \geq 2$.
\end{Rem}

The Banach algebra of all bounded linear operators from a Banach space $E$ into a Banach space $F$ is denoted by $\mathcal{L}(E,F)$ and, for simplicity, $\mathcal{L}(E)=\mathcal{L}(E,E)$. For any complex Hilbert space $\mathcal{H}$ with inner product $(\cdot,\cdot)$ (with the convention that it is linear in the first component and antilinear in the second component), if $\mathbb{A}\in \mathcal{L}(\mathcal{H})$ then we denote its (Hilbert space) adjoint by $\mathbb{A}^\dagger$. The real $\mathfrak{R}(\mathbb{A})$ and imaginary $\mathfrak{I}(\mathbb{A})$ parts of $\mathbb{A}\in \mathcal{L}(\mathcal{H})$ are defined by\footnote{For a complex number $z=a+\mathrm{i}b\ (a,b\in \mathbb{R})$, we use the notation $\overline{z}=a-\mathrm{i}b$ for complex conjugation, $\operatorname{Re}z:=\frac{1}{2}(z+\overline{z})=a, \operatorname{Im}z:=\frac{1}{2\mathrm{i}}(z-\overline{z})=b$, for the real and imaginary part, resp., of $z$. In contrast, we use \eq{def:OperatorRealImagnaryParts} for the operator real and imaginary parts of $\mathbb{A}$ to avoid any potential confusion that could arise, especially in the context of matrices. Indeed, the real part  and the imaginary part of a matrix $\BA\in M_{d}(\bbC)$ are often defined  by  $\operatorname{Re}(\BA)=(\BA+\overline{\BA})/2$ and  $\operatorname{Im}(\BA)=(\BA-\overline{\BA})/2 \mathrm{i}$ (namely by taking the real and imaginary parts of the entries of $\BA$, e.g., see \cite[p.\ 7]{Horn:2013:MAN}), which in general is different from the definition \eqref{def:OperatorRealImagnaryParts}.}
    \begin{gather}
        \mathfrak{R}(\mathbb{A}):=\frac{1}{2}(\mathbb{A}+\mathbb{A}^\dagger), \ \mathfrak{I}(\mathbb{A}):=\frac{1}{2\mathrm{i}}(\mathbb{A}-\mathbb{A}^\dagger),\label{def:OperatorRealImagnaryParts}
    \end{gather}
and, in particular, their associated quadratic forms are related by
\begin{gather*}
    \operatorname{Re}(\mathbb{A}v,v)=(\mathfrak{R}(\mathbb{A})v,v),\;\operatorname{Im}(\mathbb{A}v,v)=(\mathfrak{I}(\mathbb{A})v,v),\;\forall v\in \mathcal{H}.
\end{gather*}
We say $\mathbb{A}\in \mathcal{L}(\mathcal{H})$ is positive semidefinite, denoted by $\mathbb{A}\geq 0$, if $(\mathbb{A} v,v)\geq 0, \forall v\in \mathcal{H}$ and $\mathbb{A}>0$, if $\mathbb{A}$ is also invertible. 
We write $\mathbb{A}\leq \mathbb{B}$ (i.e., $\mathbb{A}\geq \mathbb{B}$) if $\mathbb{A}, \mathbb{B}\in \mathcal{L}(\mathcal{H})$ and $\mathbb{B}-\mathbb{A}\geq 0$. An operator $\mathbb{A}\in \mathcal{L}(\mathcal{H})$ is called coercive\footnote{See, e.g., \cite[p.\ 156, Def.\ 4.2.6 \& Rem.\ 4.2.7]{Assou:2018:MFE} for several other equivalent definitions to \eq{def:CoerciveOperator} for coercivity.} if
\begin{gather}
    \exists c>0,\gamma\in[0,2\pi)\;\text{such that}\;\operatorname{Im}\left[e^{\mathrm{i}\gamma}(\mathbb{A}v,v)\right]\geq c\,(v,v),\;\forall v\in \mathcal{H},\label{def:CoerciveOperator}
\end{gather}
that is, $0<cI_{\mathcal{H}}\leq  \mathfrak{I}(e^{\mathrm{i}\gamma}\mathbb{A})$, where $I_{\mathcal{H}}$ denotes the identity operator on $\mathcal{H}$.

For the Hilbert space $\mathbb{C}^d$, we use the standard inner product and denote its norm by $||\cdot||_{\mathbb{C}^d}$. We denote by  $M_d(\mathbb{C})$ the finite-dimensional complex Banach space of $d \times d$ complex matrices endowed with the spectral norm  $\|\cdot\|$ induced by the norm $||\cdot||_{\mathbb{C}^d}$ on $\bbC^d$. 
In this setting, for a matrix ${\bf A} \in M_d(\mathbb{C})$, one denotes by $\BA^{\top}$ its transpose, $\BA^\dagger=\overline{\BA}^{\top}$ its adjoint  (i.e., conjugate-transpose) and define similarly the notation $\mathfrak{R}(\BA), \mathfrak{I}(\BA)$, and the notion of positive semi-definite as well as coercive matrices as above for operators. In addition, $\BI$ will denote the identity matrix.

For an open set $\mathcal{O}$ of $\mathbb{R}^d$ ($d\geq 2$), we denote by $ M_d( L^{\infty}(\mathcal{O}))$  the space of  essentially bounded measurable functions from $\mathcal{O}$ to $M_d(\bbC)$ endowed with  the norm $\|\cdot\|_{\infty}$ defined  by
$$
\|\Ba\|_{\infty}=\mathrm{ess} \, \underset{\Bx\in \mathcal{O}}{\sup} \|\Ba(\Bx)\|, \quad  \, \forall \Ba\in M_d( L^{\infty}(\mathcal{O})).$$
An $\Ba\in M_d( L^{\infty}(\mathcal{O}))$ is called uniformly coercive if 
\begin{equation}\label{def.coercivitytensor}
 \hspace{-0.75em}\exists c>0,\gamma\in[0,2\pi)\;\text{such that}\;\operatorname{Im}\left[e^{\mathrm{i \gamma}}\Ba(\Bx) \BW\cdot \overline{\BW}\right] \geq c   ||\BW||_{\mathbb{C}^d}^2, \,  \ \forall \BW \in \bbC^d, \, \mbox{ for a.e. } \Bx \in \mathcal{O}.
\end{equation}

Finally, we use the notation for some subsets of the complex plane $\mathbb{C}$ and the real line $\mathbb{R}$:
$$\mathbb{C}^{+}=\{z\in\mathbb{C}\mid\operatorname{Im}z>0\},  \ \mathbb{R}^{\pm}=\{x\in\mathbb{R}\mid \pm x\geq 0\}, \ \mathbb{R}^{\pm,*}=\{ x\in \mathbb{R}\mid \pm x>0\}.$$

\subsection{A review of some properties of Herglotz and Stieltjes functions}\label{subsec:ReviewHerglotzStieltjesFuncs}
We recall now the definition of two classes of analytic functions: Herglotz and Stieltjes function which model passive materials and which will be used throughout this paper.
For more details about the properties of these functions, we refer to \cite[Appendix]{Krein:1977:MMP} and \cite{Krein:1974:RFU,Gesztesy:2000:MVH, Berg:2008:SPBS}. 

We first introduce the class of scalar Herglotz functions.
\begin{Def}
An analytic function $h:\bbC^{+}\to \bbC$ is a Herglotz function (also called a Nevanlinna or Pick function) if
$$
\Imag h(z)\geq 0, \ \forall z \in \bbC^{+}. 
$$
\end{Def}
A useful property of Herglotz functions is the following representation theorem due to Nevanlinna \cite{Nevanlinna:1922:AEF}.
\begin{Thm}\label{thm.Herg}
 $h$ is Herglotz function if and only if it admits the following representation
\begin{equation}\label{eq.defhergl}
h(z)=\alpha \, z+\beta + \displaystyle \int_{\bbR} \left(  \frac{1}{\xi-z}- \frac{\xi}{1+\xi^2}\right)\md \mm( \xi), \ \mbox{ for } \Imag(z)>0,
\end{equation}
where $\alpha\in\bbR^{+}$, $\beta\in \bbR$, and $\mm$ is a positive regular Borel measure for which   $\int_{\bbR} \
 \md \mm(\xi)/(1+\xi^2)$ is finite.  
\end{Thm}

\begin{Rem}
If  for the measure $\mm$ introduced in the representation \eqref{eq.defhergl}, $\int_{\bbR}| \xi| \, \md \mm(\xi)/(1+\xi^2)$ is  finite, then we can rewrite  the relation (\ref{eq.defhergl}) as:
\begin{equation*}\label{eq.herg2}
h(z)=\alpha \, z+\gamma  + \displaystyle \int_{\bbR} \frac{\md  \mm( \xi) }{\xi-z}  \ \quad \mbox{ with } \gamma=\beta- \int_{\bbR}  \frac{\xi \,\md \mm( \xi)}{1+\xi^2}\in \bbR.
\end{equation*}
\end{Rem}
For a given Herglotz function $h$, the triple $(\alpha, \beta, \mm)$  introduced in the representation \eqref{eq.defhergl} of $h$  is uniquely defined by the following corollary.
\begin{Cor}\label{cor.Herg}
Let $h$ be a Herglotz function defined by its representation (\ref{eq.defhergl}), then we have:
\begin{eqnarray}
&& \alpha= \lim_{y\to +\infty}\displaystyle \frac{h(\ii y)}{\ii y}, \ \beta=\Real h(\ii), \quad \forall \; a \in \bbR,  \quad \operatorname{m}(\{a\})=\lim_{y\to 0^{+}}y \; \operatorname{Im}h(a+  \mathrm{i} y),\label{eq.coeffdomherg} \\[10pt] 
 && \, \forall  (a,b) \subset \bbR \mbox{ with $a< b$},  \, \frac{\mm([a,b])+\mm((a,b))}{2} =\lim_{y\to 0^{+}} \frac{1}{\pi} \int_{a}^{b}  \Imag h(x+\ii y) \, \md x \label{eq.mesHerg}.
\end{eqnarray}
\end{Cor}
Indeed, the representation theorem \ref{thm.Herg} and the Lebesgues dominated convergence theorem imply (see, e.g., \cite{Bernland:2011:SRC}) that a Herglotz function satisfies the following asymptotics on $\mathrm{i}\bbR^{+,*}$:
\begin{equation}\label{eq.asymherglotzfunction}
h(z)=-\mm(\{0\}) z^{-1}+o(z^{-1} ) \ \mbox{ as } |z| \to  0 \ \mbox{ and } \ h(z)=\alpha \, z+o(z) \mbox{ as } |z| \to + \infty.
\end{equation}
In other words, a Herglotz function grows at most as rapidly $z$ when $ |z|$  tends to $+\infty$ and cannot be more singular than $z^{-1}$ when $ |z|$ tends to $0$ along the positive imaginary axis. 
More generally, one defines matrix-valued Herglotz functions  \cite{Gesztesy:2000:MVH} as follows. 
\begin{Def}\label{def.matHerg}
An analytic function $h:\bbC^{+}\to M_n(\bbC)$ is a matrix-valued Herglotz function  if
$$
\mathfrak{I} [h(z)]\geq 0, \ \forall z \in \bbC^{+}. 
$$
\end{Def}

The following  straightforward lemma allows one to prove that a matrix-valued function is a Herglotz function by just considering its quadratic form.
\begin{Lem}\label{lem:matHergViaQuadForms}
A matrix-valued function $h:\bbC^{+}\to M_n(\bbC)$ is a Herglotz function if and only if the scalar-valued function $z\mapsto(h(z)u)\cdot \overline{u}$ is a Herglotz function for every $u\in \mathbb{C}^n.$ 
\end{Lem}

We now introduce a second class of analytic function, namely, the Stieltjes functions which are closely related to Herglotz functions:
\begin{Def}\label{Def.Stielt}
A Stieltjes function is an analytic function $s:\bbC\setminus{\bbR}^{-} \to \bbC$ which satisfies:
$$
\Imag s(z)\leq 0 \ \,  \forall z \in \bbC^{+} \ \mbox{ and } \ s(x)\geq 0 \  \mbox{ for } x > 0.
$$
\end{Def}
\noindent From this definition follows by the analytic continuation and the Schwarz reflection principle that a Stieltjes function $s$ has to satisfy 
\begin{equation}\label{eq.StieltSchwprinciple}
s(z) =\overline{ s(\overline{z})}, \quad \forall z\in \bbC\setminus \bbR^{-}.
\end{equation}
As Herglotz functions, Stieltjes functions are characterized by a representation theorem.

\begin{Thm}\label{thm.Stielt}
A necessary and sufficient condition for $s$ to be a Stieltjes function is given by the following representation:
\begin{equation}\label{eq.representationStieltjes}
s(z)=\alpha+ \displaystyle \int_{\bbR^{+}}  \frac{\md \mm( \xi)}{\xi+z} \quad \forall z\in \bbC\setminus \bbR^{-},
\end{equation}
where $\alpha\in \mathbb{R}^+$  and $\mm$ is a positive regular Borel measure for which  $\int_{\bbR^{+}} \
 \md \mm(\xi)/(1+\xi)$ is finite. Moreover,  $\alpha= \lim \limits_{x\to +\infty} s(x)$ (for $x>0$)  and  the measure  $\mm$ in \eqref{eq.representationStieltjes} are uniquely defined by $s$.  
\end{Thm}

\begin{Rem}
Using representation Theorems \ref{thm.Herg} and \ref{thm.Stielt}, we point out that it is straightforward to notice that if $s$ is a Stieltjes function, then  $h$ defined by $h(z)=s(-z)$ on $\bbC^+$ is an Herglotz function whose measure $\mm$ has a support included in $\bbR^{+}$ in the relation (\ref{eq.defhergl}). Another connection between Herglotz and Stieltjes functions is given in Section \ref{sec-Herg-phys}  by Corollary \ref{cor.HergStielt}.
\end{Rem}

\section{Formulation of the cloaking problem}\label{sec:FormulationCloakingProb}

\subsection{Setting of the problem}\label{subsec:SettingCloakingProb}
Here we consider the following challenging question: is it possible to use a passive cloak to make invisible a dielectric inclusion on a finite frequency interval in the quasistatic regime of Maxwell's equations for an observer close to the object? We prove, due to the passivity of the cloaking device, that the answer to this question is negative. Furthermore, based on the properties of two classes of analytic functions (namely, Herglotz and Stieltjes functions) that model passive devices, we prove some bounds which impose fundamental limits on passive cloaking device over a finite frequency interval. This article extends the results of \cite{Cassier:2017:BHF} to the case of the near-field problem. In \cite{Cassier:2017:BHF}, the authors were dealing with the far-field cloaking problem where the observer is far from the cloaking device. In this context, the main scattering contribution to detect the device is its associated polarizability tensor (i.e., a $3\times 3$ complex-valued matrix function of the frequency). In our present setting, the observer is close to the device so that no mathematical far-field approximation is made for this problem. This observer can probe the system by imposing a complex voltage at the boundary of the cloaking device (or at some outer boundary within which lies the cloaking device) and measures on this boundary the resulting normal displacement of the electric field. In other words, one can probe the medium by knowing the 
Dirichlet-to-Neumann (DtN) operator associated to the system over the considered frequency interval. As the DtN operator is, for each considered frequency, defined on an infinite-dimensional functional space, it seems that compared to \cite{Cassier:2017:BHF} where the authors were dealing with a $3 \times 3$ polarizability tensor, one has more information to limit the cloaking effect over a frequency interval. However, we show, as in \cite{Cassier:2017:BHF}, that one only needs to chose one measurement to impose fundamental limits on this effect. Moreover, one can even take this measurement to simply be affine.

We assume (see Figure \ref{fig.med}) that the passive cloak together with the object fills a bounded, open, and non-empty simply connected Lipschitz domain $\Omega$ of $\mathbb{R}^d$ (see Remark \ref{rem.cloakotuside} for more details). The object to be cloaked is a passive dielectric, isotropic or anisotropic, inclusion that occupies a non-empty open (or more generally, any Lebesgue--non-negligible measurable) subset $\mathcal{O} \subseteq \Omega$.
 It should be noted here that $\mathcal{O}$ need not be simply connected (physically, this means we can have several inclusions to cloak or inclusions with a non-trivial topology). We assume also that its associated permittivity is a real-valued  tensor $\BGve_{\mathrm{ob}}(\Bx,\omega)$ on the frequency interval $[\Go_-,\Go_+]\subset \mathbb{R}^{+,*}$ (
with $\omega_-<\omega_+$). To make this object invisible, one uses a  cloak occupying $\Omega\setminus \mathcal{O}$ (with positive volume given by its Lebesgue measure $|\Omega\setminus \mathcal{O}|$), that is made of passive (and possibly dissipative) material whose permittivity tensor is given by $\Gve_{c}(\Bx,\omega)$ for a.e.\ $\Bx \in \Omega \setminus \CO$.
 Thus, one defines the permittivity function of the cloak plus the cloaked object for frequency $\Go$ in $[\Go_-,\Go_+]$ by:
 \begin{equation}\label{eq.passive-cloak}
 \BGve(\Bx,\Go)=\BGve_{\mathrm{ob}}(\Bx,\omega) \mbox{ for a.e.\ \Bx }\in \CO \mbox{ and }
 \BGve(\Bx,\Go)=\BGve_{c}(\Bx,\omega) \mbox{ for a.e.\ } \Bx \in \Omega \setminus \CO   .
\end{equation}
Moreover, the object is made of a standard non-dispersive, lossless and reciprocal material whose permittivity is strictly larger than the permittivity of the vacuum  $\Gve_0>0$  on $[\Go_-,\Go_+]
$. Furthermore, mathematically, we extend this latter relation analytically to the upper half-plane $\mathbb{C}^+$, so that   $\mbox{for a.e.\ \Bx }\in \CO \mbox{ and } \forall \Go \in [\Go_-,\Go_+] \cup \bbC^+$:
\begin{equation}\label{eq.permitivtyob}
    \mathfrak{I}[\BGve_{\mathrm{ob}}(\Bx,\Go)]=0 \  \mbox{ with } \ \BGve_{\mathrm{ob}}(\Bx,\Go)=\BGve_{\mathrm{ob}}(\Bx) =\BGve_{\mathrm{ob}}(\Bx)^{\top} \mbox{ and } \BGve_{\mathrm{ob}}(\Bx)\geq \Gve \BI \mbox{ with }   \Gve >\Gve_0.
\end{equation}

\begin{Rem}\label{rem.cloakotuside}
More generally, $\Omega$ can be considered as a bounded open simply connected domain with Lipschitz boundary which contains the object and the cloak.  Thus, the object and the cloak doesn't have to fill $\Omega$. For passive cloaking,  we only require  that $\Omega\setminus \mathcal{O}$ is occupied by passive media. Therefore,  our results hold also for passive cloaking methods where the cloak does not surround the object as for cloaking using anomalous resonances \cite{Milton:2006:CEA, Nicorovici:2007:QCT, Ammari:2013:STN, Nguyen:2017:CAO,McPhedran:2020:RAR} or complementary media \cite{Lai:2009:CMI, Nguyen:2016:CCM}, see also references within.
\end{Rem}

\subsection{The main assumptions on the permittivity of the passive cloaking device} \label{subsec:MainAssumpOnPermittivity}
The permittivity $\BGve(\Bx,\Go)$ of the cloaking device  is defined by \eqref{eq.passive-cloak} in the considered frequency interval $[\Go_-,\Go_+]$.
However, as the material is passive, $\BGve(\Bx,\Go)$ has to satisfy  some analytical properties for frequency $\omega$ in the upper-half plane $\bbC^+$. These properties introduce correlation between the response of the system at different frequencies and constrain the electromagnetic properties of the system response on a frequency interval  \cite{Milton:1997:FFR,Zemanian:2005:RTC,Bernland:2011:SRC,Welters:2014:SLL, Cassier:2017:BHF,Cassier:2017:MMD,Cassier-Joy:2025operator1}.  Therefore, to detail these properties in the following paragraph, one extends the definition of $\BGve(\Bx,\Go)$  in \eqref{eq.passive-cloak} to frequencies in the upper-half plane $\bbC^+$.
 
\begin{itemize}
\item[H1] {\bf Passivity:} As the cloak  is made a passive material, 
for a.e.\ $\Bx \in \Omega\setminus \mathcal{O}$,
$\omega \mapsto \omega \BGve(\Bx,\omega)$ is a matrix-valued Herglotz function in the sense of the Definition \ref{def.matHerg}. 

\item[H2] {\bf Reality principle:}  For a.e.\ $\Bx \in \Omega\setminus \mathcal{O}, \, \forall \Go \in \bbC^{+},\ \BGve(\Bx,-\overline{\Go})=\overline{\BGve(\Bx,\Go)}$. We point out that this symmetry  reflects the fact that the  electric susceptibility tensor of the cloak is a real-valued tensor in the time-domain [and the time Fourier-Laplace transform of the susceptibility  is precisely the tensor $\BGve(\Bx,\Go)/\varepsilon_0- \BI$].

\item[H3]{\bf High-frequency behavior:} We assume the following limit behavior of the permittivity of the cloak   for complex frequencies which lies \mbox{on the positive imaginary axis} $\mathrm{i}\mathbb{R}^{+,*}$: 
\begin{equation}\label{eq.hfbv}
\mbox{for a.e.\ $\Bx \in \Omega\setminus \CO$,}  \ \BGve(\Bx,\Go)\to \Gve_0 \BI \,   \mbox{ as } |\Go| \to \infty \mbox{ on }  \mathrm{i}\mathbb{R}^{+,*}.  
\end{equation}
The equation \eqref{eq.hfbv} is a mathematical minimal assumption which is indeed related to high-frequency behavior on the material.
Except for some mathematical ``pathological cases" (such as a meromorphic permittivity  function $\BGve$ with a  discrete number of poles on the real axis, e.g., see equations (80) and (81) of \cite{Cassier:2017:MMD}), in physical models (such as the generalized Drude-Lorentz models, e.g., see \cite{Cassier:2017:MMD,Cassier:2023:LBS,Cassier:2025:LBS})  the limit \eqref{eq.hfbv} 
 coincides with the high-frequency limit on the real axis. Indeed, due the electrons inertia, in the high frequency regime, the cloak behaves as a non-dispersive material  regime whose permittivity high frequency limit is the one of the vacuum  $\Gve_0$  (cf.\ \cite[Sec.\ 78]{Landau:1984:ECM}).

 \item [H4] \textbf{Frequency continuity on  $[\Go_-,\Go_+]$ for the cloak:} We assume that the permittivity of the cloaking device admits a continuous extension from the upper-half plane $\bbC^+$  to  the considered frequency interval $[\Go_-,\Go_+]$, i.e., the function $\Go\mapsto \omega \BGve(\Bx,\cdot)$ is continuous on $[\Go_-,\Go_+] \cup \bbC^+$ for a.e. $\Bx\in \Omega\setminus \mathcal{O}$.
\end{itemize}

\begin{Rem}
    We point out that using \eqref{eq.permitivtyob}, it is clear that H1, H2 and H4 holds for the permittivity of the whole cloaking device, i.e., those statements are then also true if we replace $\Omega\setminus \mathcal{O}$ by $\Omega$.
\end{Rem}
 
As in \cite{Cassier:2017:BHF}, the two following assumptions are made on the permittivity tensor to ensure the well-posedness of the considered partial differential equations (PDEs) in this paper. Thus, they ensure the existence of the DtN operator used to describe the cloaking problem.
 \begin{itemize}
 \item[H5] \hspace{-0.49417pt}{\bf Locally uniformly bounded in frequency}: $\forall \Go_0 \in[\Go_-,\Go_+]\cup \bbC^+$, $\BGve(\cdot,\Go_0)\in M_d(L^{\infty}(\Omega))$  and there exists  $c_1(\omega_0),\delta>0$  such that $$\forall \omega\in B(\Go_0,\delta) \cap ([\Go_-,\Go_+]\cup \bbC^+), \,  \displaystyle  \|\BGve(\cdot,\Go)\|_{\infty} \leq c_1(\omega_0).$$
Moreover, we suppose that this last property holds also in a neighborhood of $\Go_0=\infty$ by replacing, in the previous relation, $B(\Go_0,\delta)$ with $\{\Go=\mathrm{i}y \mid y \in \bbR^{+,*}, \ y>1/\delta\}$.

\item[H6]{\bf Coercivity assumption:}\footnote{We point out that a  typo has to be corrected in the coercivity assumption ${\tilde{\textnormal{H}}}$7 in \cite{Cassier:2017:BHF} page 071504-16 which is similar to the assumption $\textnormal{H6}$ made in this paper. The terms $|\Imag(e^{i\, \gamma(\Go)}\BGve(\Bx,\Go)\BE\cdot\overline{\BE})|$ and $|\Imag(e^{i\, \gamma(\Go_0)}\BGve(\Bx,\Go)\BE\cdot\overline{\BE})|$ should be replaced respectively by $\Imag(e^{i\, \gamma(\Go)}\BGve(\Bx,\Go)\BE\cdot\overline{\BE})$ and $\Imag(e^{i\, \gamma(\Go_0)}\BGve(\Bx,\Go)\BE\cdot\overline{\BE})$. We point out that the same correction has to be made in the assumptions $\widetilde{\textnormal{H}}$7a and $\widetilde{\textnormal{H}}$7b  page 20  of \cite{Ou:2022:AHN}, where the authors in a survey on Herglotz functions and their applications, dedicates several pages to sum up the results of \cite{Cassier:2017:BHF}.}
\begin{itemize}

\item 
$\forall \Go_0 \in [\Go_-,\Go_+]$, $\exists \delta, c_2(\Go_0)>0 \mid $  $\forall \omega \in B(\Go_0,\delta) \cap ([\Go_-,\Go_+]\cup \mathbb{C}^+)$ there exists  $\gamma(\Go)\in [0,2 \pi)$ such that:
$$
 \  \Imag\left[e^{\mathrm{i}\, \gamma(\Go)} \omega \BGve(\Bx,\Go)\BE\cdot\overline{\BE}\right] \geq c_2(\Go_0)   ||\BE||_{\mathbb{C}^d}^2, \,  \ \forall \BE \in \bbC^d, \, \mbox{ for a.e. } \Bx \in \Omega.
$$
\end{itemize}

\begin{Rem}\label{rem:DtNMapTCoercivityRelaxedHypWellPosedness}
    We point out that H5 and the above coercive assumption H6 are standard sufficient conditions for the mathematical well-posedness of the PDE related to our cloaking problem. However, it is a bit restrictive in the sense it does not apply for some specific   cloaking devices. For instance, if one considers the case of a cloak made of a lossless homogeneous isotropic negative index metamaterial whose permittivity $\BGve_c(\cdot, \Go)$ is a negative constant function on $[\omega_-, \omega_+]$ with a dielectric obstacle with a positive permittivity tensor $\BGve_{ob}(\Bx)$ satisfying \eqref{eq.permitivtyob}, such a coercive assumption does not hold. However, it can be relaxed and the well-posedness of the PDE can be obtained using mathematical methods developed for second-order divergence form elliptic PDE with sign changing coefficients such as the $T-$coercivity \cite{BonnetBenDhia:2012:TIP, Assou:2018:MFE} approach or the method developed  in \cite{Nguyen:2016:LAP} for anisotropic media.
    However, for simplicity, we consider here the standard coercivity assumption.
    Furthermore, we stress that the  bounds  we derive  in this paper does not depend on the positive coercive constants $c_2(\omega_0)$ which could be therefore  arbitrarily small.
    \end{Rem}

    \begin{Rem}\label{rem:CoercivityAssumpForVarEpsilon}
        It is straightforward to check that if $\omega\BGve(\cdot,\omega)$ satisfies the assumption  $\textnormal{H5}$ and the coercivity assumption H6 then so does $\BGve(\cdot,\omega)$.
    \end{Rem}

\end{itemize}

The next lemma tells us that because of \eqref{eq.permitivtyob} and under assumptions H1 and H3, we automatically have coercivity of $\BGve(\cdot,\omega)$ in the upper half-plane $\mathbb{C}^+$.  
\begin{Lem}\label{Lem.comptnesor-vacuum}
If assumptions H1 and H3 hold then 
\begin{align}\label{eq.permittivitybound}
        \mathfrak{I}[\omega \BGve(\cdot,\Go)]\geq \operatorname{Im}(\omega) \varepsilon_0 \mathbf{I}>0,\;\; \forall \omega \in \bbC^+.
    \end{align}
\end{Lem}
\begin{proof}
First, by assumptions H1 and H3, $\omega\mapsto \omega \BGve(\cdot,\Go)$ is a matrix-valued  Herglotz function which satisfies the limit \eqref{eq.hfbv}. Hence, for any fixed $\BU \in \bbC^d$ and a.e.\ $\Bx\in \Omega\setminus \mathcal{O}$, 
the function  $ \omega \mapsto  \omega \BGve(\Bx,\Go)\bU \cdot \overline{\BU}$ defined on $\bbC^+$ is a scalar-valued Herglotz function so that by Theorem \ref{thm.Herg} it admits the following representation, for any $\omega\in \bbC^+$:
\begin{equation}\label{eq.Herglotzeps}
  \omega \BGve(\Bx,\Go)\bU \cdot \overline{\BU} =  \varepsilon_0 \, |\bU|^2\, \omega  + b_{\Bx,\bU}+  \int_{\xi \in \bbR} \left(\frac{1}{\xi-\omega}- \frac{\xi}{1+\xi^2}\right)\md \mm_{\Bx,\BU}( \xi),  
\end{equation}
where $b_{\Bx,\bU}\in \bbR$ (we point out that if H2 holds then $b_{\Bx,\bU}=0$) and $ \mm_{\Bx,\BU}$ is a positive regular Borel measure for which $\int_{\bbR} \
  \md \mm_{\Bx,\BU}( \xi) /(1+\xi^2)$ is finite.  
Thus, as $b_{\Bx,\bU}\in \bbR$, $\operatorname{Im}[(\xi-\omega)^{-1}]$ is positive for $\omega\in \bbC^+$, and  $ \mm_{\Bx,\BU}$ is a positive measure,  \eqref{eq.Herglotzeps} implies that 
\begin{equation}\label{eq.coercivecloak}
\mbox{ for a.e.}\ \Bx\in \Omega\setminus \mathcal{O} \mbox{ and } \forall \omega \in \bbC^+, \quad \operatorname{Im}\left[\omega  \BGve(\Bx,\Go) \bU\cdot \overline{\BU}\right] \geq \varepsilon_0\, \operatorname{Im}(\omega) ||\bU||_{\mathbb{C}^d}^2, \quad  \forall \bU \in \bbC^d.
\end{equation}
Furthermore, using \eqref{eq.permitivtyob}, one has  for a.e. $\Bx\in \mathcal{O}$ and  $\forall \omega \in \bbC^+$:
\begin{equation}\label{eq.coerciveobject}
\operatorname{Im}\left[\omega  \BGve(\Bx,\Go) \bU\cdot \overline{\BU}\right] =\operatorname{Im}(\omega) \left[\BGve_{ob}(\Bx) \bU\cdot \overline{\BU}\right]  \geq \varepsilon\, \operatorname{Im}(\omega) ||\bU||_{\mathbb{C}^d}^2\geq\varepsilon_0\, \operatorname{Im}(\omega) ||\bU||_{\mathbb{C}^d}^2, \  \forall \bU \in \bbC^d.
\end{equation}
Combining \eqref{eq.coercivecloak} and \eqref{eq.coerciveobject} leads to the inequality \eqref{eq.permittivitybound} at the tensor level.
\end{proof}

\subsection{Mathematical formulation  of the cloaking problem in quasistatic regime}\label{subsec:MathFormulationCloakingProbRegime}
In our cloaking problem, we consider the quasistatic regime of Maxwell's equations.
This regime is justified as an approximation under certain assumptions, see \cite[Chap.\ VII]{Landau:1984:ECM}, \cite{Cheney:1999,Thaler:2014}, \cite[Chap.\ 11]{Milton:2022:TOC}, and \cite{Ando:2016:PSF} (which has been extended to include eddy currents \cite{Buffa:2000:JEC}), that allows bodies that are not necessarily small and/or frequencies not necessarily near zero. It should also be noted that the quasi-static regime can also be applicable when electrical fields are large, e.g., the quasistatic approximation is valid in a region of anomalous localized resonance around the interface due to the fact that the field gradients are extremely large, see \cite[Sec.\ 3]{McPhedran:2020:RAR}. Thus, our considerations in this paper are relevant for a large class of frequency and spatially dependent permittivities. For instance, in our context, the quasistatic regime regime holds if the cloaking device is small compared to wavelengths associated to the frequency interval $[\omega_-,\omega_+]$ (and, more precisely, to the wavelength corresponding to the maximal frequency greater than $\omega_+$).

In this setting, the propagative part of the time-harmonic Maxwell equations can be neglected, leading to a decoupling of the equations. Thus, it follows that the complex-valued electric field satisfies  $\nabla \times \BE(\cdot, \omega)=0$ holds on $\Omega$. As $\Omega$ is simply connected, this is equivalent to the electric field being derived from a complex potential, i.e., $\BE(\cdot,\omega)=-\Grad V(\cdot,\omega)$. Furthermore, as $\nabla \cdot \BD(\cdot, \Go)=\nabla \cdot  \BGve(\cdot,\omega) \BE(\cdot, \Go)=0$, it means that the potential field $V\in H^1(\Omega)$ satisfies the following elliptic equations:
\begin{equation*}\label{TM}
(\mathcal{P})\left\{ \begin{array}{ll}
\nabla \cdot \BGve(\cdot, \omega) \Grad V= 0  & \mbox{ in }  \Omega, \\[10pt]
 V=V_0 &  \mbox{ on }  \partial \Omega,
\end{array} \right.
\end{equation*}
where the surface potential $V_0$ 
is imposed on the boundary $\partial \Omega$ by the observer (see Figure \ref{fig.cloackprob}). Then, the observer can measure the resulting normal component of the displacement field, namely: 
\begin{align*}
    \BD(\cdot, \Go)\cdot \Bn=-\BGve(\cdot, \Go) \Grad V \cdot \Bn \mbox{ on } \partial \Omega, 
\end{align*}
where $\Bn$ denotes the unit outward normal  to $\Omega$.
Thus, the observer has access to the Dirichlet-to-Neumann operator of the cloaking device which maps the data $V_0$ on $\partial \Omega$ to the measurement $\BGve(\cdot, \Go) \Grad V \cdot \Bn$ on $\partial \Omega$.

\begin{figure}[!ht]
\centering
 \includegraphics[width=0.95\textwidth]{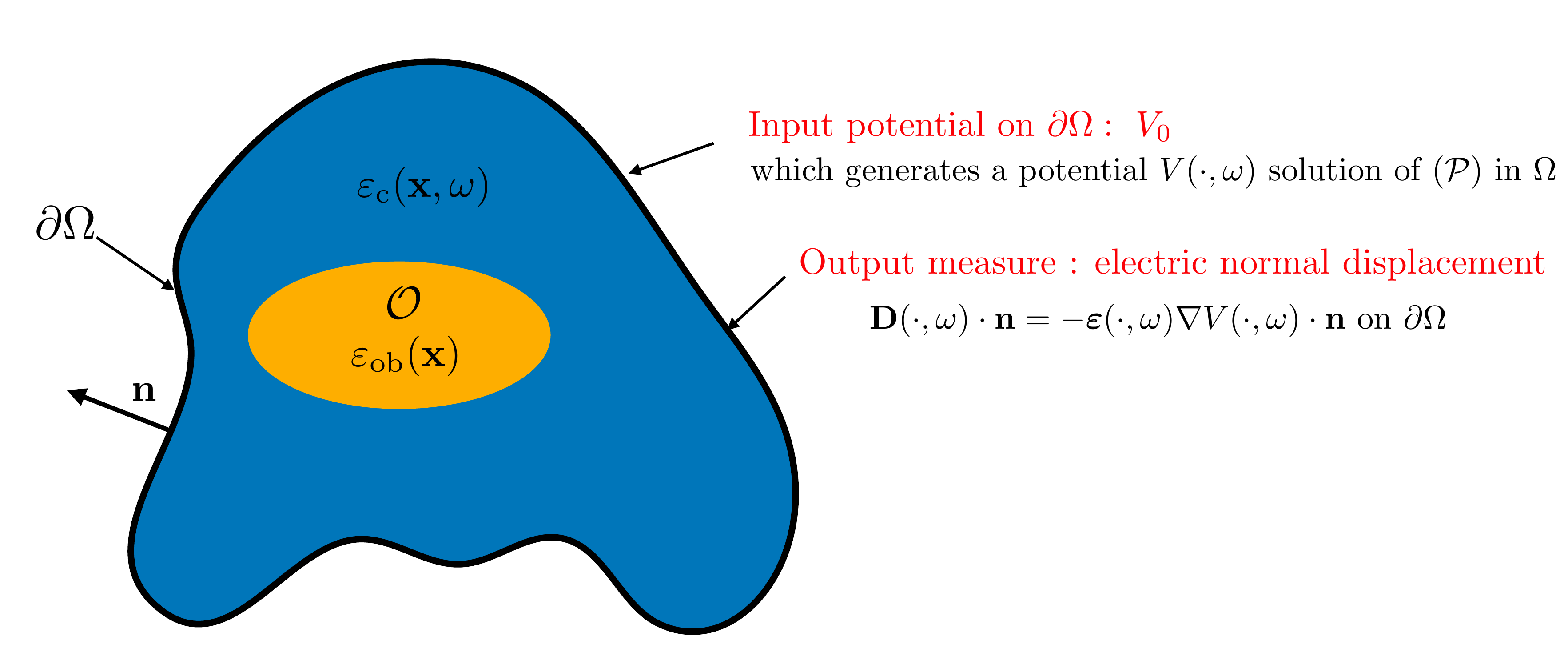}
 \caption{Description of the cloaking problem: the observer imposes a surface potential  
$V_0 \in H^{1/2}(\partial \Omega)$ at the boundary of the cloaking device  $\partial \Omega$
and measures, at each frequency $\omega$ of $[\omega_-, \omega_+]$, 
the resulting normal component of the electric displacement: $ \mathbf{D}(\cdot,\omega)\cdot \mathbf{n}
  = -\varepsilon(\cdot,\omega)\,\nabla V \cdot \mathbf{n}$ on $\partial\Omega$, where the potential 
$V(\cdot,\omega)$ is the solution of the elliptic problem $(\mathcal{P})$.
 }
 \label{fig.cloackprob}
\end{figure}

\begin{Rem}
We emphasize that the elliptic equation $(\mathcal{P})$ is  physically relevant only in the frequency interval of interest 
$[\omega_{-},\omega_{+}]$ where the quasistatic approximation holds.
Nevertheless, since the dielectric tensor of the cloak
$\boldsymbol{\varepsilon}(\mathbf{x},\omega)$ is  defined also (via the passive constitutive laws) 
for all frequencies $\omega$ in $\bbC^+$, one can  still consider these equations 
mathematically for $\omega \in [\Go_-,\Go_+]\cup \mathbb{C}^{+}$ to derive fundamental limits on cloaking  on the considered frequency range $[\Go_-,\Go_+]$.
\end{Rem}
We recall now the mathematical definition of this operator and its associated functional framework in the following proposition (the proof of this proposition is standard, but for the benefit of the reader we prove it in the appendix and, specifically, Subsec.\ \ref{sec-annexe}). This proposition is stated for general $\Ba$, which later in the paper will be taken to have various different properties, but includes $\Ba=\BGve(\cdot, \Go)$.

\begin{Pro}\label{pro.DtNDirBVPResults}
Let $\Omega$ be a bounded Lipschitz domain of $\bbR^d$ (with $d\geq 2$) and assume $\Ba\in M_d( L^{\infty}(\Omega))$ is uniformly coercive [in the sense  \eq{def.coercivitytensor}]. Then, for any data $V_0\in H^{\frac{1}{2}}(\partial \Omega)$, the Dirichlet boundary-value problem:
\begin{align}
\nabla \cdot \Ba \Grad u= 0 \text{ in }\Omega \quad \mbox{ with } u=V_0 \ \mbox{ on }  \partial \Omega \label{def.DirBVPForCoercivityTensors}
\end{align}
has a unique solution $u\in H^1(\Omega)$ that depends continuously on the boundary data $V_0$. Thus, one can define the Dirichlet-to-Neumann operator $\Lambda_{\Ba}: H^{\frac{1}{2}}(\partial \Omega)  \to H^{-\frac{1}{2}}(\partial \Omega)$ by
\begin{align}
\Lambda_{\Ba} V_0= \gamma_n( \Ba\, \nabla u)=\Ba \nabla  u\cdot  \Bn \ \mbox{ on } \partial \Omega\label{def.DtNOpDirBVPForCoercivityTensors}
\end{align}
and $\Lambda_{\Ba}\in \mathcal{L}(H^{\frac{1}{2}}(\partial \Omega),H^{-\frac{1}{2}}(\partial \Omega) )$. 
\end{Pro}

Using the assumptions H1, H3, H5, and H6 (see also Remark \ref{rem:CoercivityAssumpForVarEpsilon} and Lemma \ref{Lem.comptnesor-vacuum}), Proposition \ref{pro.DtNDirBVPResults} justifies the introduction of the Dirichlet-to-Neumann (DtN) operator of the problem $(\mathcal{P})$ which is the function 
\begin{gather}
\Lambda_{\BGve(\cdot, \omega)}\in \mathcal{L}(H^{\frac{1}{2}}(\partial \Omega), H^{-\frac{1}{2}}(\partial \Omega))    
\end{gather}
defined, for each $V_0 \in H^{\frac{1}{2}}(\partial \Omega)$ and each  frequency $\omega \in [\Go_-,\Go_+]\cup \bbC^+$, by
\begin{gather}\label{eq.Defdtnfreqdomain}
    \Lambda_{\BGve(\cdot, \omega)} V_0= \BGve(\cdot, \omega) \nabla  V\cdot  \Bn\mbox{ on } \partial \Omega,
\end{gather}
where $V$ is a the unique solution  of the elliptic equations $(\mathcal{P})$. 

In this setting, the object is cloaked  on the frequency interval $[\omega_-, \omega_+]$ if the observer is not able to distinguished from their measurements between the DtN operator associated to the cloaking device and the DtN operator corresponding to the case of empty space occupying all of $\Omega$ whose permittivity is that of the the vacuum.

Thus, it leads one to introduce  
for any fixed potential   $V_0\in H^{1/2}(\partial \Omega)$ imposed by the observer, the function $F_{V_0}:[\Go_-,\Go_+]\cup \bbC^+ \to \mathbb{C}$ defined by 
\begin{gather}\label{eq.defFV}
   F_{V_0}:\omega \mapsto  \langle [\Lambda_{\BGve(\cdot, \omega)} -\Lambda_{\Gve_0}]V_0,\overline{V_0} \rangle_{H^{-\frac{1}{2}}(\partial \Omega), H^{\frac{1}{2}}(\partial \Omega)}. 
\end{gather}
The modulus of $F_{V_0}$ measures the quality of the cloaking. 
It leads to the following definition of perfect cloaking on the frequency interval $[\omega_-,\omega_+]$ as well as approximate cloaking at a single frequency $\omega_0\in [\omega_-,\omega_+]$.
\begin{Def}\label{Def-perfectcloaking}
The obstacle is perfectly cloaked on the frequency interval $[\omega_-,\omega_+]$ if 
$$
\forall \omega \in [\omega_-,\omega_+] \ \mbox{ and }\ \forall V_0\in H^{\frac{1}{2}}(\partial \Omega),  \mbox{ one has }  F_{V_0}(\omega)=0.
$$
Moreover, we will also use this terminology for a subinterval (including a singleton set consisting of a single frequency).
Furthermore, we say the obstacle is approximately cloaked (with a tolerance $\eta>0$) at a frequency $\omega_0\in [\omega_-,\omega_+]$ if \begin{equation}\label{eq.defaprroximatecloaking}
|F_{V_0}(\Go_0)|\leq\eta \,  G^{\operatorname{vac}}_{V_0} \quad \mbox{ with } \quad G^{\operatorname{vac}}_{V_0}=\langle \Lambda_{\varepsilon_0} V_0,\overline{V_0} \rangle_{H^{-\frac{1}{2}}(\partial \Omega), H^{\frac{1}{2}}(\partial \Omega)} ,  \quad \, \forall  V_0\in H^{1/2}(\partial \Omega).
\end{equation}
\end{Def} 
\noindent We point out that in the above definition   $G^{\operatorname{vac}}_{V_0}$ is non-negative (and even positive if $V_0$ is not a constant surface potential) since  if one  denotes by $u$ the unique solution in $H^1(\Omega)$ of \eqref{def.DirBVPForCoercivityTensors} for $\Ba=\varepsilon_0 \,\mathbf{I}$, one has by the Green's formula \eqref{eq:GreensFormula} applied to  $u\in H^1(\Omega)$ and $\Bv=\Ba\nabla u\in H_{div}(\Omega)$:
\begin{equation}\label{eq.defGvac}
G^{\operatorname{vac}}_{V_0}=\langle \Lambda_{\varepsilon_0}V_0,\overline{V_0} \rangle_{H^{-\frac{1}{2}}(\partial \Omega), H^{\frac{1}{2}}(\partial \Omega)}=\varepsilon_0\|\nabla u\|^2_{\mathcal{H}}.
\end{equation}

For any $\Ba \in M_d(L^{\infty}(\Omega))$ which is uniformly coercive [in the sense of \eqref{def.coercivitytensor}],  the operator norm of the DtN map $\Lambda_{\Ba}$ (defined in Proposition \ref{pro.DtNDirBVPResults}) is given by
\begin{equation}\label{eq.opnormdtN}
\|\Lambda_{\Ba}\|=\sup_{U, V\in H^{1/2}(\partial \Omega)\setminus \{0\}} \frac{|\langle \Lambda_{\Ba} U,\overline{V} \rangle_{H^{-\frac{1}{2}}(\partial \Omega), H^{\frac{1}{2}}(\partial \Omega)}|}{\|U\|_{ H^{\frac{1}{2}}(\partial \Omega)}\, \|V\|_{H^{\frac{1}{2}}(\partial \Omega)}}.
\end{equation}
Moreover, by virtue of the polarization identity \cite[p.\ 25]{Teschl:2014:MMQ}, one can express the sesquilinear form $\langle \Lambda_{\Ba} U,\overline{V} \rangle_{H^{-\frac{1}{2}}(\partial \Omega), H^{\frac{1}{2}}(\partial \Omega)}$  as follows
\begin{eqnarray}\label{eq.pola}
\langle \Lambda_{\Ba} U,\overline{V} \rangle_{H^{-\frac{1}{2}}(\partial \Omega), H^{\frac{1}{2}}(\partial \Omega)}
&=&\frac{1}{4} \sum_{k=0}^{3}
\big[i^{k} \langle \Lambda_{\Ba}(U + i^{k} V,\, \overline{ U + i^{k} V}\rangle_{H^{-\frac{1}{2}}(\partial \Omega), H^{\frac{1}{2}}(\partial \Omega)}\big].
\end{eqnarray} 
Thus, using the above polarization identity, one deduces that perfect cloaking on the frequency interval $[\omega_-,\omega_+]$  is equivalent to the following equality in $ \mathcal{L}( H^{\frac{1}{2}}(\partial \Omega), H^{-\frac{1}{2}}(\partial \Omega))$:
$$
\Lambda_{\BGve(\cdot, \omega)}=\Lambda_{\varepsilon_0}, \quad \forall \Go \in  [\omega_-,\omega_+].
$$
In other words, the DtN operator of the cloaked object coincides with the one of the vacuum in the frequency interval $[\omega_-,\omega_+]$.

If the modulus of the function $F_{V_0}(\cdot)$ is not zero, but remains small on $[\Go_-,\Go_+]$  for all imposed potentials $V_0\in H^{\frac{1}{2}}(\partial \Omega)$ on $\partial \Omega$ compared to the non-negative quadratic form  $G_{V_0}^{\operatorname{vac}}$ associated  to the DtN map $\Lambda_{\varepsilon_0}$ when the medium  $\Omega$ is filled with vacuum, then one says that the object is \textit{approximately cloaked} in the sense that \eqref{eq.defaprroximatecloaking} holds.
Indeed, applying the polarization identity \eqref{eq.pola}  to the operator 
$\Lambda_{\BGve(\cdot, \omega_0)} - \Lambda_{\varepsilon_0}$ for any normalized $U$ and $V$ in $H^{\frac{1}{2}}(\partial \Omega)$, and using \eqref{eq.opnormdtN} yield, 
after a  simple computation (left to the reader), that  the definition \eqref{eq.defaprroximatecloaking} of approximate cloaking  at the frequency $\omega_0$ implies that 
\begin{equation}\label{inequalityaproxcloack}
\|\Lambda_{\BGve(\cdot, \omega_0)}-\Lambda_{\varepsilon_0}\|\leq 2\, \eta \|\Lambda_{\varepsilon_0}\|.
\end{equation}

In this paper, we show that perfect cloaking cannot occur even if the observer doesn't have access to the full infinite dimensional space of potentials $H^{\frac{1}{2}}(\partial \Omega)$ to probe the medium. But our main result for this cloaking problem are bounds related  to the function $F_{V_0}(\cdot)$ on $[\omega_-,\omega_+]$ for $d$ measurements
associated to $d$ input independent affine potentials $V_0$. Furthermore, one only needs one of these $d$ measurements to prevent with these bounds approximate cloaking. Such bounds show that even approximate cloaking  on a frequency interval is subject to fundamental limits imposed by the passivity of the device.  These limits are related to physical parameters appearing in the bounds: the bandwidth $\omega_+-\omega_-$, the central frequency $(\omega_-+ \omega_+)/2$, the volume of the obstacle $|\mathcal{O}|$, the volume of the cloak  $|\Omega\setminus \mathcal{O}|$, and a lower bound $\varepsilon/\varepsilon_0$ on the relative  permittivity  of the obstacle.

\section{Effective operator representations for the DtN map}\label{sec:EffOpReprDtNMap}
In the abstract theory of composites, a Hilbert space framework exists for defining an effective operator which is based on a generalized notion of a Hodge decomposition and constitutive material relation that is encapsulated in a concept known as the $Z$-problem (see \cite{Milton:1990:CSP, Milton:2022:TOC, Grabovsky:2016:CMM, Milton:2016:ETC, Welters:EOM:2026} as well as \cite{Grabovsky:2018:BRE} for an excellent review of this problem, \cite[Sec.\ 1.2]{Beard:2022:RVP} for a brief history on the $Z$-problem, \cite{Beard:2023:EOVP}, and references within). This theory will be recalled briefly in Subsec.\ \ref{subsec:AbstrThyComposites} after we introduce, in Subsec.\ \ref{subsec:HodgeDecompDtNMap}, the Hodge decomposition that we use in our paper associated with the Dirichlet problem for Laplace's equation (see Theorem \ref{th.tpro}). A motivation for using this framework is that it allows one to develop systematically bounds on the effective operator in a similar manner as in the theory of composites, for instance, the two variational principles known as the Dirichlet and Thomson variational principles, see \cite{Milton:1990:CSP, Milton:2022:TOC, Milton:2016:ETC, Beard:2023:EOVP, Welters:EOM:2026}, from which one can derive, e.g., the Wiener (1912) arithmetic and harmonic
mean bounds on the effective conductivity \cite{Wiener:1912:TMF}. On this, we postpone our discussion until Sec.\ \ref{Sec:VarPrincBndsEffOps}. With this in mind, the main result of this section (i.e., Theorem \ref{thm.effopreprDtNmap} in Subsec.\ \ref{sec.linkDtnsigmastar}) gives a new representation of the DtN map in terms of an effective operator for a certain $Z$-problem, that we will call the ``Dirichlet $Z$-problem" (see Def.\ \ref{DefKZProb} in Subsec.\ \ref{Dirichlet-Z-problem}), based on the Hodge decomposition from Subsec.\ \ref{subsec:HodgeDecompDtNMap}. In later sections, we will show how to use this representation to derive bounds on the DtN map using those two variational principles and then apply it to our cloaking problem.

\subsection{Hodge decomposition for the DtN map representation}\label{subsec:HodgeDecompDtNMap}
Hodge (also known as Helmholtz or Helmholtz-Hodge) decompositions \cite{Dautray:1980:MAN, Monk:2003:FEM, Schwarz:2006:HDM, Bhatia:2013:HHD} and associated orthogonal projection methods \cite{Weyl:1940:MOP} for Hilbert spaces of functions are key tools in the theory of composites \cite{DellAntonio:1986:ATO, DellAntonio:1988:GRE, Milton:2022:TOC, Milton:2016:ETC, Grabovsky:2016:CMM, Grabovsky:2018:BRE, Beard:2023:EOVP, Welters:EOM:2026} and, in particular, see \cite[Sec.\ 12.7]{Milton:2022:TOC}. For our problem, in representing the DtN map in terms of an effective operator, we need the following standard Hodge decomposition for the Dirichlet problem [i.e., $(\mathcal{P})$ with $\BGve(\cdot,\omega)= \BI$] whose statement and proof can be found, for instance, in \cite[Proposition 1, pp.\ 215-217]{Dautray:1980:MAN}. 

\begin{Thm}[Hodge decomposition for the Dirichlet problem]\label{th.tpro}
Let $\Omega$ be a bounded  simply connected Lipschitz domain of $\bbR^d$.  Then the three spaces $\mathcal{U},\;\mathcal{E},\;\mathcal{J}$  defined by 
\begin{align}
\mathcal{U} & = \nabla H_{harm}(\Omega)=\{\BE_0=-\nabla u \mid u\in H^1(\Omega) \mbox{ and } \Delta u=0 \}, \label{USpHodgeDecomp}\\
\mathcal{E} & = \{\BE=-\nabla V \mid V \in H^1_0(\Omega)\}, \label{ESpHodgeDecomp}\\
\mathcal{J} & = \{ \BJ \in H_{div,0}(\Omega) \mid  \nabla\cdot\BJ=0\},\label{JSpHodgeDecomp}
\end{align}
are mutually orthogonal subspaces in the Hilbert space  $\mathcal{H}=\BL^2(\Omega)$.
Moreover, \begin{gather} \mathcal{H}=\BL^2(\Omega)=\mathcal{U} \overset{\perp}{\oplus} \mathcal{E} \overset{\perp}{\oplus} \mathcal{J}, \label{orthotripleDtNZproblem}\\
	\operatorname{range}(\nabla|_{H^1(\Omega)})= \mathcal{U} \overset{\perp}{\oplus} \mathcal{E},\
	\operatorname{kernel}(\nabla \cdot|_{H_{div}(\Omega)}) = \mathcal{U} \overset{\perp}{\oplus} \mathcal{J}, \label{KerDivIsUPlusJ}
	\end{gather}
where $\nabla|_{H^1(\Omega)}: H^1(\Omega) \rightarrow \BL^2(\Omega)$ is the gradient operator $\nabla$ on $H^1(\Omega)$, and $\nabla \cdot|_{H_{div}(\Omega)}: H_{div}(\Omega) \rightarrow L^2(\Omega)$ is the divergence operator $\nabla \cdot$ on $H_{div}(\Omega)$.
\end{Thm}

\begin{Def}\label{def:HodgeProjectionsForDirichProb}
    We define $\BGG_0, \BGG_1, \BGG_2$ to be the orthogonal projections of the Hilbert space $\mathcal{H}=\BL^2(\Omega)$ onto $\mathcal{U}, \mathcal{E}, \mathcal{J}$ defined by \eqref{USpHodgeDecomp}, \eqref{ESpHodgeDecomp}, \eqref{JSpHodgeDecomp}, respectively.  
\end{Def}

\begin{Rem}\label{rem:SketchPrfDirHodgeDecomp}
    As this theorem is crucial to our results, we want to give here an insight into its proof, in particular, that it is an immediate consequence of the following facts. First, $\nabla|_{H^1(\Omega)}$ and $\nabla \cdot|_{H_{div}(\Omega)}$ are densely defined closed operators with closed ranges (proving one of these operators has closed range is essentially the difficult part of the proof, cf.\ \cite{Amrouche:2015:JLL}) whose adjoints are
   \begin{gather*}
        (\nabla|_{H^1(\Omega)})^\dagger=-\nabla \cdot|_{H_{div,0}(\Omega)},\;(\nabla |_{H^1_0(\Omega)})^\dagger=-\nabla \cdot|_{H_{div}(\Omega)},
    \end{gather*}
    which the latter is a consequence of the Green's formula \eqref{eq:GreensFormula}. Second, by general results in the theory of adjoints \cite[Theorem 5.13, p.\ 234]{Kato:1995:PTO}, it follows that
    \begin{gather}
        \mathcal{H}=\operatorname{range}(\nabla|_{H^1(\Omega)})\overset{\perp}{\oplus}\operatorname{kernel} (\nabla \cdot|_{H_{div,0}(\Omega)})=\operatorname{range}(\nabla |_{H^1_0(\Omega)})\overset{\perp}{\oplus}\operatorname{kernel} (\nabla \cdot|_{H_{div}(\Omega)}).\label{OpVerDtNMapHodgeDecomp}
    \end{gather}
    Finally, the Hodge decomposition theorem is an immediate consequence of this and the fact that
    \begin{eqnarray}
       \mathcal{U}&=&\operatorname{range}(\nabla|_{H^1(\Omega)})\cap\operatorname{kernel} (\nabla \cdot|_{H_{div}(\Omega)}), \label{OpDefOfUEJInDtNHodgeDecomp1} \\[4pt] \mathcal{E}&=&\operatorname{range}(\nabla |_{H^1_0(\Omega)}) \ \mbox{ and } \ \mathcal{J}=\operatorname{kernel}(\nabla \cdot|_{H_{div,0}(\Omega)}).\label{OpDefOfUEJInDtNHodgeDecomp2}
    \end{eqnarray}
    Let us elaborate on this. First, its clear from the orthogonal decompositions \eqref{OpVerDtNMapHodgeDecomp} and the equalities \eqref{OpDefOfUEJInDtNHodgeDecomp1} and \eqref{OpDefOfUEJInDtNHodgeDecomp2}, that $\mathcal{U}, \mathcal{E},$ and $\mathcal{J}$ are mutually orthogonal (closed) subspaces of $\mathcal{H}$. Hence, the theorem follows once we prove the claim: if $\BF\in\mathcal{H}$ then $\BF=\BE_0+\BE+\BJ$ for some $\BE_0\in \mathcal{U}, \BE\in \mathcal{E}, \BJ\in \mathcal{J}$. Let us prove this. Suppose $\BF\in \mathcal{H}$. Then [by the first equality of \eqref{OpVerDtNMapHodgeDecomp}] $\BF=\nabla u+\BJ$ with $\BJ\in \operatorname{kernel} (\nabla \cdot|_{H_{div,0}(\Omega)})=\mathcal{J}$ and $u\in H^1(\Omega)$. As $\nabla u\in \mathcal{H}$ then [by the second equality of \eqref{OpVerDtNMapHodgeDecomp})] $\nabla u=\nabla v+\BE_0$ with $v\in H^1_0(\Omega)$ (thus $\BE=\nabla v\in \mathcal{E}$) and $\BE_0\in \operatorname{kernel} (\nabla \cdot|_{H_{div}(\Omega)})$ implying $\BE_0=\nabla(u-v)\in \operatorname{range}(\nabla|_{H^1(\Omega)})\cap \operatorname{kernel} (\nabla \cdot|_{H_{div}(\Omega)})$. Therefore, one has $\BE_0\in \mathcal{U}$ (since the harmonic function $u-v\in H^1(\Omega)$). Hence we have proven that $\BF=\BE_0+\BE+\BJ$, where $\BE_0\in \mathcal{U}$, $\BE\in \mathcal{E}$ and $\mathcal{\BJ}\in \mathcal{J}$. This proves the claim.
\end{Rem}

\subsection{Abstract theory of composites: \texorpdfstring{$Z$}{Z}-problem and effective operator}\label{subsec:AbstrThyComposites}

Here we give a brief review of the essential aspects of the abstract theory of composites regarding the $Z$-problem and effective operator. The starting point of the definition of a $Z$-problem is a Hilbert space $\mathcal{H}$ which is assumed to have an orthogonal triple decomposition into subspaces $\mathcal{H=U}\overset{\bot}{\mathcal{\oplus}}\mathcal{E}\overset{\bot
			}{\mathcal{\oplus}}\mathcal{J}$ which can be considered as an abstraction of a Hodge decomposition in the case that $\mathcal{H}$ is a space of square integrable functions. Next, a (bounded) linear operator $\mathbf{L}:\mathcal{H}\rightarrow \mathcal{H}$ is given which we think of as a material tensor (more precisely, a left-multiplication operator by this tensor). The $Z$-problem then involves solving a constrained linear equation on $\mathcal{H}$  which is the  analog of a linear constitutive relation with material tensor $\mathbf{L}$ such that the input and output fields belong respectively to $\mathcal{U}\overset{\bot}{\mathcal{\oplus}}\mathcal{E}$ and $\mathcal{U}\overset{\bot}{\mathcal{\oplus}}\mathcal{J}$.  Thus,  the  constrained linear equation   satisfies  some PDE constraints implicitly encoded in the Hodge decomposition of $\mathcal{H}$. From this $Z$-problem, an effective operator $\mathbf{L}_*:\mathcal{U}\rightarrow \mathcal{U}$ is defined by the linear constitutive relations which focuses only on a restricted part of the whole system, namely, that part from subspace $\mathcal{U}$. For instance, in the setting in which $\mathcal{H}$ is a space of periodic functions and $\mathcal{U}$ is the space of their periodic averages, then $\mathbf{L}_*$ would map the average of the input to the average of the output for the material constitutive relations with material tensor $\mathbf{L}$. This is often the starting point for the theory of composites which is closely related to homogenization theory for PDEs; for more on this, see, e.g., \cite{Milton:2022:TOC, Grabovsky:2016:CMM, Milton:2016:ETC, Beard:2023:EOVP, Welters:EOM:2026}. In this setting, the effective tensor $\BL_*$ coincides with the one  obtained by the homogenization of the periodic composite.

Below we give a precise definition of the $Z$-problem and the associated effective operator (Def.\ \ref{DefZProbMain}). After this we will discuss the existence and uniqueness of solutions to the $Z$-problem for a coercive operator $\mathbf{L}$, well-definedness of the effective operator $\mathbf{L}_*$, and a formula for the effective operator as a Schur complement of $\mathbf{L}$. The main results in this regard are Theorem \ref{ThmMainClassicalZProbEffOp} and Lemma \ref{lem.uniformcoerciveimpliessolvabilityofZproblem}. In the next section we will consider the $Z$-problem associated with the Hodge decomposition introduced in Subsec.\ \ref{subsec:HodgeDecompDtNMap} and the relationship between the effective operator for this $Z$-problem and DtN map.  

\begin{Def}[$Z$-problem and effective operator]\label{DefZProbMain}
	The $Z$-problem
	\begin{equation}
		(\mathcal{H},\mathcal{U},\mathcal{E},\mathcal{J},\BL), \label{DefZProb}
	\end{equation}
	is the following problem associated with a Hilbert space $\mathcal{H}$ endowed with its inner product $(\cdot,\cdot) $, an orthogonal triple decomposition of $\mathcal{H}$ as 
	\begin{equation}
		\mathcal{H=U}\overset{\bot}{\mathcal{\oplus}}\mathcal{E}\overset{\bot
			}{\mathcal{\oplus}}\mathcal{J},\label{DefZProbHOrthTri}
	\end{equation}
	and a bounded linear operator $\BL\in\mathcal{L}(\mathcal{H})$:
	given $\BE_{0}\in\mathcal{U}$, find triples $\left(\BJ_{0},\BE,\BJ\right)\in\mathcal{U}\times\mathcal{E}\times\mathcal{J}$ satisfying 
	\begin{equation}
		\BJ_{0}+\BJ=\BL \left(  \BE_{0}+\BE\right).
		\label{DefZProbEq} 
	\end{equation}
If such a triple $\left(  \BJ_{0},\BE,\BJ\right)$ exists then it is called a solution of the $Z$-problem at $\BE_{0}$.
In addition, if there exists a bounded linear operator $\BL_*\in\mathcal{L}(\mathcal{U}$) such that 
	\begin{equation}
		\BJ_{0}=\BL_{\ast}\BE_{0}, \label{DefZProbEffOp}
	\end{equation}
	whenever $\BE_0\in\mathcal{U}$ and $\left(  \BJ_{0},\BE,\BJ\right)$ is a solution of the $Z$-problem at $\BE_0$, then $\BL_*$ is called an effective operator of the $Z$-problem.
\end{Def}

Consider a $Z$-problem $	(\mathcal{H},\mathcal{U},\mathcal{E},\mathcal{J},\BL)$.
Then we can write the operator
\begin{align}
 \BL=[\BL_{ij}]_{i,j=0,1,2}\in\mathcal{L}(\mathcal{H}) \label{3b3BlockOpReprOfSigma}  
\end{align}
as a $3\times 3$ block operator matrix with respect to the orthogonal triple decomposition (\ref{DefZProbHOrthTri}) of the Hilbert space $\mathcal{H}=\mathcal{U}\ho\mathcal{E}\ho\mathcal{J}$. More precisely, we introduce the orthogonal projections $\BGG_0,\BGG_1,\BGG_2$ of $\mathcal{H}$ onto $H_0=\mathcal{U}, H_1=\mathcal{E}, H_2=\mathcal{J},$ respectively, and define
\begin{align}
    \BL_{ij}\in \mathcal{L}(H_j,H_i),\;\BL_{ij}=\BGG_i\BL\BGG_j:H_j\rightarrow H_i,\label{DefOfSigmaSubblocks}
\end{align}
for $i,j=0,1,2$. In particular, $\BL_{11}$ is the compression of $\BL$ to $\mathcal{E}$, that is,
\begin{align}
   \BL_{11}= \BGG_1\BL\BGG_1|_{\mathcal{E}},\label{DefAltSigma11Subblock}
\end{align}
i.e., the restriction of the operator $\BGG_1\BL\BGG_1$ on $\mathcal{H}$ to the closed subspace $\mathcal{E}$. Then equation (\ref{DefZProbEq}) is equivalent to the system
\begin{gather}
    \BL_{00}\BE_0+\BL_{01}\BE=\BJ_0,\label{ZProbEquivFormPart1}\\
        \BL_{10}\BE_0+\BL_{11}\BE=0,\label{ZProbEquivFormPart2}\\
        \BL_{20}\BE_0+\BL_{21}\BE=\BJ.\label{ZProbEquivFormPart3}
\end{gather}
Finally, from this and assuming $\BL_{11}$ is invertible, we get the classical formulas for the solution of equation (\ref{DefZProbEq}) and a representation formula for the effective operator as a Schur complement (the formula is also in \cite[Eq.\ (12.57)]{Milton:2022:TOC}):
\begin{gather}
    \BJ_0=\BL_*\BE_0,\;\BE=-\BL_{11}^{-1}\BL_{10}\BE_0,\; \BJ=\BL_{20}\BE_0+\BL_{21}\BE,\label{ClassicSolnZProb}\\
    \BL_*=\begin{bmatrix}
        \BL_{00}&\BL_{01}\\
        \BL_{10}&\BL_{11}
    \end{bmatrix}/\BL_{11}=\BL_{00}-\BL_{01}\BL_{11}^{-1}\BL_{10}.\label{ClassicEffOperFormula}
\end{gather}
Thus, we have proven the following theorem. 
\begin{Thm}\label{ThmMainClassicalZProbEffOp}
If $(\mathcal{H},\mathcal{U},\mathcal{E},\mathcal{J},\BL)$ is a $Z$-problem (as in Def.\ \ref{DefZProbMain}) and $\BL_{11}$ [as defined by (\ref{DefAltSigma11Subblock})] is invertible then equation (\ref{DefZProbEq}) (i.e., the $Z$-problem at $\BE_0$) has a unique solution for each $\BE_0\in \mathcal{U}$ and it is given by the formulas (\ref{ClassicSolnZProb}), (\ref{ClassicEffOperFormula}). Moreover, the effective operator $\BL_*$ of the $Z$-problem exists, is unique, and is given by the Schur complement formula (\ref{ClassicEffOperFormula}).
\end{Thm}

The following lemma gives some important properties of effective operators, including an alternative representation (see also \cite[Prop.\ 17]{Beard:2023:EOVP}) for $\BL_*$, which will be useful later.
\begin{Lem}\label{lem:ZProbEffOpAltFormulaViaInv}
    Let $(\mathcal{H},\mathcal{U},\mathcal{E},\mathcal{J},\BL)$ be a $Z$-problem such that $\BL_{11}\in \mathcal{L}(\mathcal{U})$ is invertible. Then:
    \begin{itemize}
        \item[(i)] If $\lambda$ is a nonzero scalar then $(\lambda \BL)_*=\lambda \BL_*$.
        \item[(ii)] If $\BL(\mathcal{U})\subseteq \mathcal{U}$ then $\BL_*=\BL$, i.e., $\BL_*\BE_0=\BL\BE_0, \ \forall \BE_0\in \mathcal{U}$.
        \item[(iii)] If $[\BL_{ij}]_{i,j=0,1}\in \mathcal{L}(\mathcal{U}\overset{\bot}{\mathcal{\oplus}}\mathcal{E})$ is invertible then $([\BL_{ij}]_{i,j=0,1}^{-1})_{00}\in \mathcal{L}(\mathcal{U})$ is invertible and
    \begin{gather}
        \BL_*=([\BL_{ij}]_{i,j=0,1}^{-1})_{00}^{-1}.\label{ZProbEffOpAltFormulaViaInvL}
    \end{gather}
        \item[(iv)] The operator $(\BL^\dagger)_{11}=(\BL_{11})^\dagger\in \mathcal{L}(\mathcal{U})$ is invertible and $(\BL_*)^\dagger=(\BL^\dagger)_*$. 
    \end{itemize}
\end{Lem}
\begin{proof}
    The proof of statements (i), (ii), and (iv) are immediate consequences of the Schur complement formula (\ref{ClassicEffOperFormula}) and also for (iv) that $\BGG_{j}^\dagger=\BGG_j, j=0,1,2$ since they are orthogonal projections hence self-adjoint. Let us now prove (iii). As $\BL_{11}$ is invertible, then Theorem \ref{ThmMainClassicalZProbEffOp} is true so that the effective operator $\BL_*$ exists and is given by \eqref{ClassicEffOperFormula}. By the assumption that $[\BL_{ij}]_{i,j=0,1}$ is invertible and from the Aitken block-diagonalization formula (see, for instance, \cite[Lemma 11]{Beard:2023:EOVP} or \cite{Zhang:2005:SCA}), 
    \begin{gather*}
        \begin{bmatrix}
            \BL_{00} & \BL_{01}\\
            \BL_{10} & \BL_{11}
        \end{bmatrix}=\begin{bmatrix}
            \BI & \BL_{01}\BL_{11}^{-1}\\
            0 & \BI
        \end{bmatrix}\begin{bmatrix}
            \BL_* & 0\\
            0 & \BL_{11}
        \end{bmatrix}\begin{bmatrix}
            \BI & 0\\
            \BL_{11}^{-1} \BL_{10} & \BI
        \end{bmatrix},
    \end{gather*}
    it follows that $\BL_*$ is invertible and the Banachiewicz inversion formula (see, for instance, \cite[Lemma 14]{Beard:2023:EOVP} or \cite{Zhang:2005:SCA}) holds:
    \begin{gather*}
    [\BL_{ij}]_{i,j=0,1}^{-1} =\begin{bmatrix}
        (\BL_*)^{-1} &-(\BL_*)^{-1}\BL_{01}\BL_{11}^{-1}\\
        -\BL_{11}^{-1}\BL_{10}(\BL_*)^{-1} & \quad\BL_{11}^{-1}+\BL_{11}^{-1}\BL_{10}(\BL_*)^{-1}\BL_{01}\BL_{11}^{-1}
    \end{bmatrix}.
    \end{gather*}
    Thus, it follows that
    \begin{gather*}
        ([\BL_{ij}]_{i,j=0,1}^{-1})_{00}=(\BL_*)^{-1},
    \end{gather*}
    which implies the desired formula \eqref{ZProbEffOpAltFormulaViaInvL}.
\end{proof}

The following lemma is useful as it gives a sufficient condition for both $\BL_{11}$ and $[\BL_{ij}]_{i,j=0,1}$ to be invertible, we omit the proof as it is straightforward.
\begin{Lem}\label{lem.uniformcoerciveimpliessolvabilityofZproblem}
If $\BL$ is coercive [in the sense \eq{def:CoerciveOperator}] then $\BL_{11}$ and $[\BL_{ij}]_{i,j=0,1}$ are coercive and the operators $\BL$, $\BL_{11}$, and $[\BL_{ij}]_{i,j=0,1}$ are all invertible. In particular, Theorem \ref{ThmMainClassicalZProbEffOp} is true for any $Z$-problem $(\mathcal{H},\mathcal{U},\mathcal{E},\mathcal{J},\BL)$ such that $\BL$ is coercive.
\end{Lem}

\subsection{Dirichlet \texorpdfstring{$Z$}{Z}-problem and associated effective operator}\label{Dirichlet-Z-problem}

We now define a $Z$-problem and associated effective operator which arises from the Hodge decomposition (\ref{orthotripleDtNZproblem}) in Theorem \ref{th.tpro}, the Dirichlet boundary-value problem (\ref{def.DirBVPForCoercivityTensors}), and the DtN operator (\ref{def.DtNOpDirBVPForCoercivityTensors}).

\begin{Def}[Dirichlet $Z$-problem]\label{DefKZProb}
	Let $\Ba:\Omega \rightarrow M_d(\mathbb{C})$ be a matrix-valued function such that $\Ba\in M_d(L^{\infty}(\Omega))$ is uniformly coercive [in the sense of  \eqref{def.coercivitytensor}]. 
    The $Z$-problem $(\mathcal{H},\mathcal{U},\mathcal{E},\mathcal{J},\BL_\Ba)$ associated with the Hilbert space $\mathcal{H}=\BL^2(\Omega),$ orthogonal triple decomposition of $\mathcal{H}$ in \eqref{orthotripleDtNZproblem}, and the bounded linear operator $\BL_\Ba\in \mathcal{L}(\mathcal{H})$ of left multiplication by $\Ba$, is called the Dirichlet $Z$-problem. For simplicity, $\BL_{\Ba}$ will be denoted by $\Ba$ and with this abuse of notation we refer to $(\mathcal{H},\mathcal{U},\mathcal{E},\mathcal{J},\Ba)$ as the Dirichlet $Z$-problem associated with $\Ba$ and effective operator $(\BL_{\Ba})_*$ by $\Ba_*$.
\end{Def}

\begin{Pro}\label{pro.ExistenceUniquenessDirichletZProbEffOp}
For any Dirichlet $Z$-problem $(\mathcal{H},\mathcal{U},\mathcal{E},\mathcal{J},\Ba)$, each of the operators $\Ba\in \mathcal{L}(\mathcal{H})$, $[\Ba_{ij}]_{i,j=0,1}=(\BGG_0+\BGG_1)\Ba(\BGG_0+\BGG_1)|_{\mathcal{U}\overset{\bot}{\mathcal{\oplus}}\mathcal{E}}\in \mathcal{L}(\mathcal{U}\overset{\bot}{\mathcal{\oplus}}\mathcal{E}),$ and $\Ba_{11}=\Gamma_1\Ba\Gamma_1|_{\mathcal{E}}\in \mathcal{L}(\mathcal{E})$ are coercive [in the sense of \eq{def:CoerciveOperator}] and invertible. Thus, Theorem \ref{ThmMainClassicalZProbEffOp} is true for any Dirichlet $Z$-problem and, in particular, the Dirichlet $Z$-problem at $\BE_0$ admits a unique solution for each $\BE_0\in \mathcal{U}$.
\end{Pro}
\begin{proof}
Consider a Dirichlet $Z$-problem $(\mathcal{H},\mathcal{U},\mathcal{E},\mathcal{J},\Ba)$. Then $\Ba: \Omega \to M_d(\bbC) $ is a matrix-valued function such that $\Ba\in M_d( L^{\infty}(\Omega))$ and there exists $c>0$ and $\gamma\in[0,2\pi)$ with
\begin{equation*}
\operatorname{Im}(e^{\mathrm{i}\gamma}\Ba(\Bx)\BW\cdot \overline{\BW})\geq c||\BW||_{\mathbb{C}^d}^2,\ \forall \BW \in \bbC^d, \, \mbox{ for a.e. } \Bx \in \Omega.
\end{equation*}
This implies that $\Ba$ is coercive since for all $\BW\in \mathcal{H}$:
\begin{equation*} \operatorname{Im}\left[e^{\mathrm{i}\gamma}(\Ba\BW,\BW)_{\mathcal{H}}\right]    =\int_{\Omega}\operatorname{Im}\left[e^{\mathrm{i}\gamma}\Ba(\Bx)\BW(\Bx)\cdot \overline{\BW(\Bx)}\right] \md \Bx
    \geq\int_{\Omega}c \, ||\BW(\Bx)||_{\mathbb{C}^d}^2 \md \Bx= c \, \|\BW\|^2_{\mathcal{H}}.
\end{equation*}
The proof of this proposition now follows from this and Lemma \ref{lem.uniformcoerciveimpliessolvabilityofZproblem}.
\end{proof}

The next result follows directly from Lemma \ref{lem:ZProbEffOpAltFormulaViaInv} and Proposition \ref{pro.ExistenceUniquenessDirichletZProbEffOp}, so we omit the proof.
\begin{Cor}\label{cor:BasicPropertiesEffOpDiricZProb}
    Let $(\mathcal{H},\mathcal{U},\mathcal{E},\mathcal{J},\Ba)$ be a Dirichlet $Z$-problem. Then:
    \begin{itemize}
        \item[(i)] If $\lambda\in \mathbb{C}\setminus\{0\}$ then $(\lambda \Ba)_*=\lambda \Ba_*$.
        \item[(ii)] If $\Ba$ is a constant matrix then $\Ba_*=\Ba$, i.e., $\Ba_*\BE_0=\Ba\BE_0, \ \forall \BE_0\in \mathcal{U}$.
        \item[(iii)] The operator $[\Ba_{ij}]_{i,j=0,1}\in \mathcal{L}(\mathcal{U}\overset{\bot}{\mathcal{\oplus}}\mathcal{E})$ is invertible, $([\Ba_{ij}]_{i,j=0,1}^{-1})_{00}=\BGG_0[\Ba_{ij}]_{i,j=0,1}^{-1}\BGG_0|_{\mathcal{U}}\in \mathcal{L}(\mathcal{U})$ is invertible, and
    \begin{gather}
        \Ba_*=([\Ba_{ij}]_{i,j=0,1}^{-1})_{00}^{-1}.\label{ZProbEffOpAltFormulaViaInv}
    \end{gather}
    \item[(iv)] The matrix-valued function $\Ba^{\dagger}:\Omega\rightarrow M_d(\mathbb{C})$ defined by $\Ba^\dagger(\Bx)=\Ba(\Bx)^\dagger$ for all $\Bx\in \Omega$ has the properties: $\Ba^\dagger\in M_d(L^{\infty}(\Omega))$, $\Ba^\dagger$ is uniformly coercive [in the sense of  \eqref{def.coercivitytensor}], and in terms of the effective operator $(\Ba^\dagger)_*\in \mathcal{L}(\mathcal{U})$ of the Dirichlet $Z$-problem $(\mathcal{H},\mathcal{U},\mathcal{E},\mathcal{J},\Ba^\dagger)$ we have
    \begin{gather*}
        (\Ba_*)^\dagger=(\Ba^\dagger)_*,
    \end{gather*}
    where $(\Ba_*)^\dagger$ is the (Hilbert space) adjoint of the operator $\Ba_*\in \mathcal{L}(\mathcal{U})$. 
    \end{itemize}    
\end{Cor}

\subsection{Link between the DtN operator and the effective operator}\label{sec.linkDtnsigmastar}

Here we formalize and prove our main result (Theorem \ref{thm.effopreprDtNmap}) of Sec.\ \ref{sec:EffOpReprDtNMap} which links the DtN map $\Lambda_{\Ba}$ (defined in Prop.\ \ref{pro.DtNDirBVPResults}) to the effective operator $\Ba_*$ of the Dirichlet $Z$-problem from Subsec.\ \ref{Dirichlet-Z-problem}. We can characterize this relationship between $\Lambda_{\Ba}$ and $\Ba_*$ in terms of the lift operator $\Pi$ (Def.\ \ref{def:LiftOperator}) and its adjoint $\Pi^\dagger$ by the formula: $\Lambda_{\Ba}=\Pi^\dagger\Ba_*\Pi$. One might think that $\Lambda_{\Ba}$ and $\Ba_*$ are congruent operators, but this isn't true since $\Pi$ is not invertible only surjective (by Prop.\ \ref{prop:LiftOpIsBoundedLinearOnto}).

\begin{Def}[Lift operator]\label{def:LiftOperator}
The lift operator $\Pi: H^{\frac{1}{2}}(\partial \Omega)\to \mathcal{U}$ is the function defined by
\begin{align}
    \Pi(V_0)=\nabla u,\;\;V_0\in H^{\frac{1}{2}}(\partial \Omega),
\end{align}
where $u$ is the harmonic function [i.e., $u\in H_{harm}(\Omega)$] uniquely defined by the following Dirichlet boundary-value problem:
\begin{equation}
\Delta u=0  \mbox{ in }  \Omega  \ \mbox{ with }  \ u=V_0 \mbox{ on } \partial \Omega.\label{def.DirichletBVP}
\end{equation}    
\end{Def}

\begin{Pro}\label{prop:LiftOpIsBoundedLinearOnto}
The lift operator $\Pi: H^{\frac{1}{2}}(\partial \Omega)\to \mathcal{U}$ is a bounded linear operator, i.e., $\mathcal{L}(H^{\frac{1}{2}}(\partial \Omega), \mathcal{U})$ which is surjective, but not injective. It has one-dimensional kernel which contains the constant fields on $\partial \Omega$. Moreover, its adjoint $\Pi^{\dagger}\in \mathcal{L}(\mathcal{U},H^{-\frac{1}{2}}(\partial \Omega))$ is the ``Neumann" trace operator $\Pi^{\dagger}:\mathcal{U} \to H^{-\frac{1}{2}}(\partial \Omega)$ defined by:
\begin{align}
   \Pi^\dagger(\nabla v)=\frac{\partial v}{\partial \Bn}, \ \forall \, v\in H_{harm}(\Omega).
\end{align}
\end{Pro}
\begin{proof} 
For each $V_0\in H^{\frac{1}{2}}(\partial \Omega)$, there exists a unique $u \in H^1(\Omega)$ (depending  linearly and continuously on $V_0$) that solves the equation \eqref{def.DirichletBVP}. This defines a bounded linear operator
$P_{\Omega} \in \mathcal{L}(H^{\frac{1}{2}}(\partial \Omega), H^1(\Omega))$ by $P_{\Omega}(V_0)=u$ (see, e.g., \cite[Theorem 3.14, p.\ 46]{Monk:2003:FEM}). Then, as $\nabla \in \mathcal{L}(H^1(\Omega), \BL^2(\Omega))$, it follows immediately that $\Pi=\nabla P_{\Omega}\in \mathcal{L}(H^{\frac{1}{2}}(\partial \Omega), \mathcal{U})$. Moreover, $\Pi$ is not injective. Indeed, 
$\Pi (V_0)=0=\nabla u$ with $u$ defined by \eqref{def.DirichletBVP}, if and only if (as $\Omega$ connected) $u$ is constant in $\Omega$.  Thus, it is equivalent that its trace $V_0$ is constant on $\partial \Omega$. Furthermore, $\Pi$ is surjective since by definition  any element of $\mathcal{U}$ can be written $\nabla u$  with  $u\in H_{harm}(\Omega)$, i.e., $u$ is harmonic (i.e., $\Delta u=0$) and  $u\in H^1(\Omega)$. Hence, as $u\in H^1(\Omega)$, its trace on $\partial \Omega$ exists, which we will denote by $V_0$, and $V_0\in H^{1/2}(\partial \Omega)$. Thus, by construction, $u$ solves \eqref{def.DirichletBVP} and therefore  $\Pi(V_0)=\nabla u$.

For the adjoint formula, one uses that for any $\nabla v\in \mathcal{U}$ [for some $v\in H_{harm}(\Omega)$] and any $V_0\in H^{\frac{1}{2}}(\partial \Omega)$:
\begin{eqnarray*}
( \nabla v, \Pi V_0)_{\mathcal{H}}&=& ( \nabla v, \nabla u)_{\mathcal{H}} \\[4pt]
&=&   (  -\nabla^2 v,  u)_{L^2(\Omega)}+ \left\langle \frac{\partial v}{\partial \Bn} ,\overline{u}\right\rangle_{H^{-\frac{1}{2}}(\partial \Omega), H^{\frac{1}{2}}(\partial \Omega)}     \mbox{ [by the Green's formula \eqref{eq:GreensFormula}] }\\
&=&  \left\langle \frac{\partial v}{\partial \Bn} ,\overline{V_0}\right\rangle_{H^{-\frac{1}{2}}(\partial \Omega), H^{\frac{1}{2}}(\partial \Omega)} \quad \mbox{(since $v$ is harmonic and $u=V_0$ on $\partial \Omega$)}
\end{eqnarray*}
which implies that $\Pi^{\dagger}(\nabla v)= \partial v /\partial \Bn$.
\end{proof}
The following Proposition (preceded by a technical lemma) establishes, for a given surface potential $V_0 \in H^{1/2}(\partial \Omega)$, the relationship between the solution $u$  to the Dirichlet boundary-value problem (\ref{def.DirBVPForCoercivityTensors}) associated with $V_0$, and the solution $(\BJ_{0}, \BE, \BJ) \in \mathcal{U} \times \mathcal{E} \times \mathcal{J}$ of the Dirichlet $Z$-problem $(\mathcal{H}, \mathcal{U}, \mathcal{E}, \mathcal{J}, \Ba)$ corresponding to $\BE_{0} = \Pi (V_{0})  \in \mathcal{U}$. This relationship will be the key ingredient to established the representation of the DtN operator in terms of an effective operator in Theorem \ref{thm.effopreprDtNmap}. 

\begin{Lem}\label{Lem-equiv-set}
Let $V_0\in H^{\frac{1}{2}}(\partial \Omega)$. Then
\begin{equation}\label{eq.setequality}
\{ \nabla \widetilde{v} \mid  \widetilde{v} \in H^1(\Omega) \text{ with } \widetilde{v}|_{\partial \Omega}=V_0\}= \{ \BE_0+\widetilde{\BE} \mid \widetilde{\BE} \in \mathcal{E} \text{ and }\BE_0\in \mathcal{U} \text{ with }\BE_0=\Pi (V_0) \}.
\end{equation}
\end{Lem}
\begin{proof}
Let $V_0 \in  H^{\frac{1}{2}}(\partial \Omega)$ be fixed. \\[4pt]
{\bf Step 1: Proof of the left inclusion of \eqref{eq.setequality}} \\[4pt]
Assume that $v\in H^1(\Omega)$ with $v|_{\partial \Omega}=V_0$. Then, by Theorem \ref{th.tpro}, we have $$\nabla v\in\operatorname{range}(\nabla|_{H^1(\Omega)})= \mathcal{U} \overset{\perp}{\oplus} \mathcal{E}.$$
Thus, one can decompose uniquely $\nabla v$ as 
\begin{equation}\label{eq.relationdecompelecfield}
\nabla v=\widetilde{\BE}_0+\widetilde{\BE}  \ \mbox{ with } \ \widetilde{\BE}_0=-\nabla V_{\widetilde{\BE}_0} \in  \mathcal{U} \ \mbox{ and } \ \widetilde{\BE}=- \nabla V_{\widetilde{\BE}}\in  \mathcal{E},
\end{equation}
for some  $V_{\widetilde{\BE}_0}\in H^1(\Omega)$ with $\Delta V_{\widetilde{\BE}_0}=0$ and some $V_{\widetilde{\BE}}\in H^1_0(\Omega)$.
Define $\BE_0:=\Pi (V_0)$ and note that $\BE_0\in \mathcal{U}$ by definition of the lift operator $\Pi$. Thus to prove the considered set inclusion, one only needs to show that $\widetilde{\BE}_0=\BE_0$. To prove this, first note by \eqref{eq.relationdecompelecfield} it follows that
$$
\nabla (v+ V_{\widetilde{\BE}_0}+V_{\widetilde{\BE}})=0.
$$ 
Hence as $\Omega$ is connected, then there exists a constant $C\in \bbC$ such that  $v+ V_{\widetilde{\BE}_0}+V_{\widetilde{\BE}}=C$ on $\overline{\Omega}=\partial \Omega\cup \Omega$ and it follows that $V_0+V_{\widetilde{\BE}_0}=C$ on $\partial \Omega$ since $V_{\widetilde{\BE}}\in H^1_0(\Omega)$.
Thus, $V_0=-V_{\widetilde{\BE}_0}+C$ on $\partial \Omega$ and one gets that the function $\psi \in H^1(\Omega)$ defined by $\psi=-V_{\widetilde{\BE}_0}+C$ on $\Omega$ satisfies:
$$
 \Delta \psi=0   \ \mbox { in }\ \Omega  \quad \mbox{and} \quad \psi=V_0 \ \mbox { on } \ \partial \Omega.
$$
Hence it follows by definition of $\Pi$ that $\BE_0=\Pi(V_0)=\nabla \psi$ and, as $\nabla \psi=-\nabla V_{\widetilde{\BE}_0}=\widetilde{\BE}_0$, this proves  that $\widetilde{\BE}_0=\BE_0$, as desired.\\[4pt]
{\bf Step 2: Proof of the right inclusion of \eqref{eq.setequality}} \\[4pt]
Assume $\widetilde{\BE}\in\mathcal{E}$ and $\BE_0\in \mathcal{U}$ with $\BE_0=\Pi(V_0)$ and define $\BE_1:=\BE_0+\widetilde{\BE}$. Then there exists $V_{\widetilde{\BE}}\in H^1_{0}(\Omega)$ such that $\widetilde{\BE}=-\nabla V_{\widetilde{\BE}}$ and $V_{{\BE}_0}\in H^1(\Omega)$ such that $\BE_0=-\nabla V_{{\BE}_0}$ where $-V_{{\BE}_0}\in H^1(\Omega)$ is the solution of the Dirichlet boundary-value problem \eqref{def.DirichletBVP} satisfying $-V_{{\BE}_0}|_{\partial \Omega}=V_0$. Thus, it is clear that $\BE_1=\nabla v$ where  $v=-(V_{{\BE}_0}+ V_{{\BE}_1})\in H^1(\Omega)$ and has a trace equal to $V_0$ on $\partial \Omega$, i.e., $v|_{\partial\Omega}=V_0$. This proves the  the right inclusion of \eqref{eq.setequality}.
\end{proof}

\begin{Pro}\label{Pro-link-DtNpDE-Z-problem} Assume $\Ba\in M_d( L^{\infty}(\Omega))$ is uniformly coercive [in the sense  \eq{def.coercivitytensor}]
and let $V_0\in H^{1/2}(\partial \Omega)$. Then  the unique solution  $u$ (defined  by Proposition \ref{pro.DtNDirBVPResults}) to the Dirichlet boundary-problem (\ref{def.DirBVPForCoercivityTensors}) associated to  $V_0$  and the  unique solution $\left(\BJ_{0},\BE,\BJ\right)\in\mathcal{U}\times\mathcal{E}\times\mathcal{J}$ of the  Dirichlet $Z$-problem (introduced in Definition \ref{DefKZProb})  $(\mathcal{H},\mathcal{U},\mathcal{E},\mathcal{J},\Ba)$   associated to  $\BE_0=\Pi (V_0)=\Pi( u|_{\partial \Omega})\in \mathcal{U}$ are related by
\begin{gather}
    \nabla u=\BE_0+\BE \ \mbox{ and }  \ \BJ_0+\BJ=\Ba \nabla u.
\end{gather}\end{Pro}
\begin{proof}
Let $V_0 \in  H^{\frac{1}{2}}(\partial \Omega)$ be fixed and $u$ be the corresponding  unique solution in $H^{1}(\Omega)$ of the Dirichlet boundary-value problem (\ref{def.DirBVPForCoercivityTensors}). In particular, the trace of $u$ on $\partial \Omega$ is equal to $V_0$, i.e., $u|_{\partial \Omega}=V_0$, and by virtue of Lemma \ref{Lem-equiv-set}, one gets that
$$
\nabla u=\BE_0+\widetilde{\BE}  \ \mbox{ with } \ \BE_0=\Pi (V_0)=\Pi( u|_{\partial \Omega})\in \mathcal{U} \ \mbox{ and } \ \widetilde{\BE}=- \nabla V_{\widetilde{\BE}}\in  \mathcal{E}, \text{ for some } V_{\widetilde{\BE}}\in H_0^1(\Omega). 
$$
Then, as $\Ba\nabla u\in\operatorname{kernel}(\nabla \cdot|_{H_{div}(\Omega)})$, one deduces from Theorem \ref{th.tpro} that $$\Ba(\BE_0+\widetilde{\BE})=\Ba\nabla u\in\operatorname{kernel}(\nabla \cdot|_{H_{div}(\Omega)})=\mathcal{U}\overset{\perp}{\oplus}  \mathcal{J}.$$ 
By the uniqueness of the solution of the Dirichlet $Z$-problem $(\mathcal{H},\mathcal{U},\mathcal{E},\mathcal{J},\Ba)$ at $\BE_0$ (see  Proposition \ref{pro.ExistenceUniquenessDirichletZProbEffOp}), it follows that $\widetilde{\BE}=\BE$ and $\BJ_0+\BJ=\Ba \nabla u$.
\end{proof}
The following theorem is our main result in Sec.\ \ref{sec:EffOpReprDtNMap}. It was Milton in \cite[Chap.\ 2]{Milton:2016:ETC} who recognized that the DtN operator can be represented in terms of an effective operator. Its mathematical reformulation and proof as given below, is the fruit of a collaboration of the three authors of the present article. The proof was completed in 2019 and based on this, the third author and his collaborators developed a discrete analog of this result for electrical networks, see \cite[Theorem 63]{Beard:2023:EOVP}.
\begin{Thm}[Effective operator representation of the DtN map]\label{thm.effopreprDtNmap}
The DtN operator $\Lambda_{\Ba}: H^{\frac{1}{2}}(\partial \Omega) \to H^{-\frac{1}{2}}(\partial \Omega)$ (as defined in Prop.\ \ref{pro.DtNDirBVPResults}) and the effective operator $\Ba_*$ of the Dirichlet $Z$-problem $(\mathcal{H},\mathcal{U},\mathcal{E},\mathcal{J},\Ba)$ (as defined in Def.\ \ref{DefKZProb}) satisfy
\begin{align}\label{eq.factDtNmap}
\Lambda_{\Ba}=\Pi^{\dagger}\Ba_*\Pi,    
\end{align}
where $\Pi: H^{\frac{1}{2}}(\partial \Omega)\to \mathcal{U}$ is the lift operator (as defined in Def.\ \ref{def:LiftOperator}), $\Pi^{\dagger}:\mathcal{U} \to H^{-\frac{1}{2}}(\partial \Omega)$ is its adjoint (see Prop.\ \ref{prop:LiftOpIsBoundedLinearOnto}), and $\Ba_*$ is given by the Schur complement formula
\begin{gather}
    \Ba_*=\Ba_{00}-\Ba_{01}\Ba_{11}^{-1}\Ba_{10}, \label{eq.Schurcomplement1}\\
    \Ba_{ij}=\BGG_i\Ba \BGG_j:\operatorname{range}(\mathbf{\Gamma}_j)\rightarrow \operatorname{range}(\mathbf{\Gamma}_i),\;i,j=0,1,\label{eq.Schurcomplement2}
\end{gather}
where $\BGG_0$ and $\BGG_1$ are the orthogonal projections of the Hilbert space $\mathcal{H}=\BL^2(\Omega)$ onto the spaces $\mathcal{U}$ and $ \mathcal{E}$, respectively (as defined in Def.\ \ref{def:HodgeProjectionsForDirichProb}).
\end{Thm}
\begin{proof}
Let $V_0 \in  H^{\frac{1}{2}}(\partial \Omega)$ be fixed.  If one denotes by $u$ the corresponding  unique solution in $H^{1}(\Omega)$ of the Dirichlet boundary-value problem (\ref{def.DirBVPForCoercivityTensors}) and  by $\left(\BJ_{0},\BE,\BJ\right)\in\mathcal{U}\times\mathcal{E}\times\mathcal{J}$  the  unique solution  of the  Dirichlet $Z$-problem $(\mathcal{H},\mathcal{U},\mathcal{E},\mathcal{J},\Ba)$  associated to $\BE_0=\Pi V_0\in \mathcal{U}$, then by  virtue of the Proposition \ref{Pro-link-DtNpDE-Z-problem},  one has the following relation
$$
 \nabla u=\BE_0+\BE\in \mathcal{U}\overset{\perp}{\oplus} \mathcal{E}\ \ \mbox{ and }  \ \BJ_0+\BJ=\Ba \nabla u \in \mathcal{U}\overset{\perp}{\oplus}  \mathcal{J}.
$$
Moreover, by Proposition \ref{pro.ExistenceUniquenessDirichletZProbEffOp}, we know that the effective operator $\Ba_{*}$ of the Dirichlet $Z$-problem $(\mathcal{H},\mathcal{U},\mathcal{E},\mathcal{J},\Ba)$ exists and from the definition of $\Ba_*$, one has
$$
\BJ_0=\Ba_* \BE_0.
$$
Then, one gets
\begin{eqnarray*}
\langle \Pi^{*}\Ba_*\Pi V_0, \overline{V_0}\rangle_{H^{-\frac{1}{2}}(\partial \Omega),H^{\frac{1}{2}}(\partial \Omega)}
&=&(\Ba_* \BE_0, \BE_0 )_{\mathcal{H}} \;\; (\mbox{since } \BE_0=\Pi \,V_0) \\
&=&(\BJ_0,\BE_0)_{\mathcal{H}}\\
&=&(\BJ_0+\BJ, \BE_0+\BE)_{\mathcal{H}} \, \;\; (\mbox{as $\mathcal{U}$, $\mathcal{E}$ and $\mathcal{J}$ are mutually orthogonal})\\
&=&(\Ba\nabla u, \nabla u)_{\mathcal{H}}\\
&=&(\Ba\nabla u, \nabla u)_{\mathcal{H}}+(\nabla\cdot\Ba\nabla u, u)_{L^2(\Omega)} \;\; (\mbox{since $\nabla\cdot\Ba\nabla u=0$ on $\Omega$})\\
&=& \langle \Lambda_{\Ba}V_0, \overline{V_0 }\rangle_{H^{-\frac{1}{2}}(\partial \Omega), H^{\frac{1}{2}}(\partial \Omega)},
\end{eqnarray*}
where for the last equality we have used the Green's formula \eqref{eq:GreensFormula} and that $\Lambda_{\Ba}V_0=\Ba\nabla u\cdot\Bn$ on $\partial\Omega$.
Therefore, from this and the polarization identity for sesquilinear forms \cite[p.\ 25]{Teschl:2014:MMQ}, it follows that $ \Pi^{*}\Ba_* \, \Pi =\Lambda_{\Ba}$.
\end{proof}

\section{Variational principles and bounds on effective operators}\label{Sec:VarPrincBndsEffOps}

In this section, we introduce (in Subsec.\ \ref{subsec:VarPrinAbstrResults}) two standard variational principles (i.e., Theorems \ref{ThmClassicalDiriMinPrin} and \ref{ThmClassicalThomMinPrin} below) from the abstract theory of composites that will be needed in this paper, which require certain operator positivity assumptions (in particular, it doesn't apply to every coercive operator). These variational principles are related to the classical Dirichlet and Thomson variational principles \cite{Courant:1953:MMP} 
for the DtN map \cite{Berryman:1990:VCE, Borcea:2002:EIT, Borcea:2003:VCN, Milton:2011:UBE}, 
and we elaborate on this relationship in Subsec.\ \ref{subsec:VarPrincAppDtNMap}. We also develop as a corollary (Corollary \ref{cor:EffOpIneqsForCoerciveOps}) of the first variational principle (Theorem \ref{ThmClassicalDiriMinPrin}), an operator inequality for the effective operator $\BL_*$ of a $Z$-problem that applies to any coercive operator $\BL$. Later, in Sec.\ \ref{sec:DiriPrbAffineBCsDtNMap}, we will use Theorems \ref{ThmClassicalDiriMinPrin} and \ref{ThmClassicalThomMinPrin} to derive bounds on the DtN map with affine boundary conditions. Then, in Sec.\ \ref{sec:DiffHerglotzFuncsDtNMapsIsHerglotz}, we will use Corollary \ref{cor:EffOpIneqsForCoerciveOps} to prove a key result of our paper: the difference of two Herglotz functions, that represent two certain different DtN maps, is a Herglotz function (see Theorem \ref{thm.differenceHerg}).

\subsection{Variational principles in the abstract theory of composites}\label{subsec:VarPrinAbstrResults}

The next two theorems are well-known (see, e.g., \cite{Milton:1988:VBE, Milton:1990:CSP, Milton:2022:TOC,Milton:2016:ETC, Beard:2023:EOVP, Welters:EOM:2026}) and their proofs can be found in, e.g., \cite{Beard:2023:EOVP}.
\begin{Thm}[Dirichlet minimization principle]\label{ThmClassicalDiriMinPrin}
If $(\mathcal{H},\mathcal{U},\mathcal{E},\mathcal{J},\BL)$ is a $Z$-problem such that $\BL^\dagger=\BL$, $\BL_{11}\geq 0$, and $\BL_{11}$ is invertible then the effective operator $\BL_*$ is the unique self-adjoint operator satisfying the minimization principle:
\begin{equation*}
		( \BL_*\BE_0,\BE_0 )=\min_{\BE\in\mathcal{E}}( \BL(\BE_0+\BE),\BE_0+\BE ),\;\forall \BE_0\in\mathcal{U},
	\end{equation*}
	and, for each $\BE_0\in\mathcal{U}$, the minimizer is unique and given by
	\begin{equation*}
		\BE=-\BL_{11}^{-1}\BL_{10}\BE_0.
	\end{equation*}
	Moreover, we have the following upper bound on the effective operator:
	\begin{equation*}
		\BL_*\leq \BL_{00}=\BGG_0\BL\BGG_0|_{\mathcal{U}}.
	\end{equation*}
\end{Thm}
An immediately corollary of Theorem \ref{ThmClassicalDiriMinPrin} is the following monotonicity result (also known in the theory of composites, e.g., see \cite[Sec.\ 13.2]{Milton:2022:TOC}, \cite[Corollary 4]{Welters:EOM:2026}).
\begin{Cor}\label{cor:MonotonocityOfEffOps}
    If $(\mathcal{H},\mathcal{U},\mathcal{E},\mathcal{J},\BL)$ and $(\mathcal{H},\mathcal{U},\mathcal{E},\mathcal{J},\BM)$ are $Z$-problems such that $\BL^\dagger=\BL, \BM^\dagger=\BM$, both $\BL_{11}$ and $\BM_{11}$ are positive semidefinite and invertible, and $\BL\leq \BM$
    then
    \begin{gather*}
        \BL_*\leq \BM_*.
    \end{gather*}
\end{Cor}
An improvement on Theorem \ref{ThmClassicalDiriMinPrin} is the following theorem which also gives a lower bound on the effective operator.
\begin{Thm}[Thomson minimization principle]\label{ThmClassicalThomMinPrin}
If $(\mathcal{H},\mathcal{U},\mathcal{E},\mathcal{J},\BL)$ is a $Z$-problem such that $\BL^\dagger=\BL\geq 0,$ and $\BL$ is invertible then $(\BL_*)^{-1}$ is the unique self-adjoint operator satisfying the minimization principle:
	\begin{gather*}
		((\BL_*)^{-1}\BJ_0,\BJ_0 )=\min_{\BJ\in\mathcal{J}}( \BL^{-1}(\BJ_0+\BJ),\BJ_0+\BJ),\ \forall \BJ_0\in\mathcal{U},
	\end{gather*}
	and, for each $\BJ_0\in\mathcal{U}$, the minimizer is unique and given by
	\begin{equation*}
		\BJ=-\BL_{22}^{-1}\BL_{20}\BJ_0.
	\end{equation*}
	Moreover, we have the upper and lower bounds on the effective operator:
	\begin{equation}
		0\leq [(\BL^{-1})_{00}]^{-1}]\leq\BL_*\leq \BL_{00},\label{ClassicalUpperLowerBddsEffOp}
	\end{equation}
	where 
	\begin{align*}
	    (\BL^{-1})_{00}=\BGG_0\BL^{-1}\BGG_0|_{\mathcal{U}}.
	\end{align*}
\end{Thm}
This following corollary [which appears to be new, at least within in the (abstract) theory of composites]
is a key result we need in our paper (for instance, see the proof of Theorem in Sec.\ \ref{sec:DiffHerglotzFuncsDtNMapsIsHerglotz}).
\begin{Cor}\label{cor:EffOpIneqsForCoerciveOps}
If $(\mathcal{H},\mathcal{U},\mathcal{E},\mathcal{J},\BL)$ is a $Z$-problem and $\BL$ is coercive, i.e.,
\begin{gather*}
        \exists c>0, \gamma\in [0,2\pi)\;|\;\operatorname{Im}[e^{\mathrm{i}\gamma}(\BL v,v)]\geq c(v,v),\;\forall v\in \mathcal{H},
    \end{gather*}
then the effective operator $\BL_*$ is coercive and
\begin{gather*}
    0<cI_{\mathcal{U}}\leq [\mathfrak{I}(e^{\mathrm{i}\gamma}\BL)]_*\leq \mathfrak{I}(e^{\mathrm{i}\gamma}\BL_*).
\end{gather*}
\end{Cor}
\begin{proof}
    Assume the hypotheses. First, consider the $Z$-problem $(\mathcal{H},\mathcal{U},\mathcal{E},\mathcal{J},\mathfrak{I}(e^{\mathrm{i}\gamma}\BL))$. Then applying Theorem \ref{ThmClassicalDiriMinPrin}, we have
    \begin{eqnarray*}
        ( [\mathfrak{I}(e^{\mathrm{i}\gamma}\BL)]_*\BE_0,\BE_0 )&=&\min_{\BE\in\mathcal{E}}( \mathfrak{I}(e^{\mathrm{i}\gamma}\BL)(\BE_0+\BE),\BE_0+\BE ) \\
       & =&\min_{\BE\in\mathcal{E}}\operatorname{Im}[e^{\mathrm{i}\gamma}( \BL(\BE_0+\BE),\BE_0+\BE)],\ \forall\BE_0\in\mathcal{U}.
    \end{eqnarray*}
    Second, consider the $Z$-problem $(\mathcal{H},\mathcal{U},\mathcal{E},\mathcal{J},e^{\mathrm{i}\gamma}\BL)$. Then by Lemma \ref{lem:ZProbEffOpAltFormulaViaInv}.(i) and Lemma \ref{lem.uniformcoerciveimpliessolvabilityofZproblem} it follows that $(e^{\mathrm{i}\gamma}\BL)_*=e^{\mathrm{i}\gamma}\BL_*$ and for any $\BE_0\in \mathcal{U}$ there exists a solution $(\BJ_0,\BE,\BJ)\in \mathcal{U}\times\mathcal{E}\times\mathcal{J}$ to this $Z$-problem at $\BE_0$ which implies
    \begin{gather*}
        e^{\mathrm{i}\gamma}(\BL_*\BE_0, \BE_0)=((e^{\mathrm{i}\gamma}\BL)_*\BE_0, \BE_0)= (e^{\mathrm{i}\gamma}\BL(\BE_0+\BE),\BE_0+\BE)
    \end{gather*}
    so that
     \begin{eqnarray*}
       (\mathfrak{I}(e^{\mathrm{i}\gamma}\BL_*)\BE_0, \BE_0) &=&\operatorname{Im}[e^{\mathrm{i}\gamma}(\BL_*\BE_0, \BE_0)]=\operatorname{Im}[e^{\mathrm{i}\gamma}(\BL(\BE_0+\BE),\BE_0+\BE)] \\[4pt]
       & \geq &\min_{\BE\in\mathcal{E}}\operatorname{Im}[e^{\mathrm{i}\gamma}( \BL(\BE_0+\BE),\BE_0+\BE)]=( [\mathfrak{I}(e^{\mathrm{i}\gamma}\BL)]_*\BE_0,\BE_0 ),
    \end{eqnarray*}
    which proves $[\mathfrak{I}(e^{\mathrm{i}\gamma}\BL)]_*\leq \mathfrak{I}(e^{\mathrm{i}\gamma}\BL_*)$. Finally, since $0<cI_{\mathcal{H}}\leq \mathfrak{I}(e^{\mathrm{i}\gamma}\BL)$ then it follows by Lemma \ref{lem:ZProbEffOpAltFormulaViaInv}.(ii) and Theorem \ref{ThmClassicalDiriMinPrin} that $0<cI_{\mathcal{U}}=(cI_{\mathcal{H}})_*\leq [\mathfrak{I}(e^{\mathrm{i}\gamma}\BL)]_*$. This completes the proof.
\end{proof}

\subsection{Application to the DtN map}\label{subsec:VarPrincAppDtNMap}

Consider now the DtN map $\Lambda_{\Ba}$ in the special case that $\Ba=\Ba(\Bx)\in M_d(L^{\infty}(\Omega))$ is uniformly positive definite, i.e.,
\begin{gather}
  \exists c>0\ \mid \ \Ba(\Bx) \BW\cdot \overline{\BW} \geq c  ||\BW||_{\mathbb{C}^d}^2, \,  \ \forall \BW \in \bbC^d, \, \mbox{ for a.e. } \Bx \in \Omega.\label{def:UniformlyPosDefinite}
\end{gather}
In particular, $\Ba$ is uniformly positive definite in the sense of (\ref{def.coercivitytensor}) (with the same $c$ and $\gamma=\pi/2$) and $\BL_{\Ba}\in \mathcal{L}(\BL^2(\Omega))$ (i.e., the operator of left multiplication by $\Ba$, see Def.\ \ref{DefKZProb}) is positive definite ($\BL_{\Ba}^\dagger=\BL_{\Ba}\geq 0$) and invertible [in particular, it is coercive in the sense of \eq{def:CoerciveOperator} with the same $c$ and $\gamma=\pi/2$].

Next, because of these hypotheses on $\Ba$, the classical Dirichlet variational principle for the Dirichlet boundary-value problem (\ref{def.DirBVPForCoercivityTensors}) can be described as follows \cite[Eq.\ (2.9)]{Borcea:2003:VCN}, \cite[Eq.\ (3.1)]{Milton:2011:UBE}, \cite[p.\ 240]{Courant:1953:MMP}: For each $V_0\in H^{\frac{1}{2}}(\partial \Omega)$,
\begin{gather}\label{eq.Dirichletpriciple}
    \min_{v\in H^1(\Omega),v|_{\partial \Omega}=V_0}(\Ba\nabla v, \nabla v)_{\mathcal{H}}=\langle \Lambda_{\Ba}V_0,\overline{V_0} \rangle_{H^{-\frac{1}{2}}(\partial \Omega), H^{\frac{1}{2}}(\partial \Omega)} 
\end{gather}
and the LHS has a unique minimizer $u$, namely, the unique solution $u\in H^1(\Omega)$ to the Dirichlet boundary-value problem (\ref{def.DirBVPForCoercivityTensors}). In addition, this also uniquely defines the DtN map $\Lambda_{\Ba}$. 
Let us first show that by virtue of the effective operator representation of the DtN map given in Theorem \ref{thm.effopreprDtNmap}, this classical variational principle is just Theorem \ref{ThmClassicalDiriMinPrin} in a different guise. This approach follows an interesting and more general viewpoint from the abstract theory of composites, which is not limited to the DtN map.

Let $V_0\in H^{1/2}(\partial \Omega)$ be fixed and defined $\BE_0\in \mathcal{U}$ by $\BE_0=\Pi(V_0)$. Then it follows  that:
 \begin{equation}\label{Dirichlet-minim-compo} 
	\hspace{-0.36443pt}\begin{array}{llll}
  \langle \Lambda_{\Ba}V_0,\overline{V_0} \rangle_{H^{-\frac{1}{2}}(\partial \Omega), H^{\frac{1}{2}}(\partial \Omega)}&=&( \Ba_*\BE_0,\BE_0 )_{\mathcal{H}} &  \mbox{(since $\Lambda_{\Ba}=\Pi^\dagger\Ba_*\Pi$ by Theorem \ref{thm.effopreprDtNmap}), }  \\[6pt]
    &=& \displaystyle \min_{\widetilde{\BE}\in\mathcal{E}}(\Ba( \BE_0\!+\!\widetilde{\BE}),\BE_0\!+\!\widetilde{\BE})_{\mathcal{H}} &\mbox{(using Theorem \ref{ThmClassicalDiriMinPrin}),}  \\[6pt]
    &=&  \displaystyle \min_{v\in H^1(\Omega),v|_{\partial \Omega}=V_0}( \Ba\nabla v,\nabla v )_{\mathcal{H}},& \mbox{(using Lemma \ref{Lem-equiv-set}).} 
    	\end{array}
\end{equation}
By Theorem  \ref{ThmClassicalDiriMinPrin},   the  minimizer $\widetilde{\BE}$ of  $\displaystyle \min_{\widetilde{\BE}\in\mathcal{E}}( \Ba(\BE_0\!+\!\widetilde{\BE}),\BE_0\!+\!\widetilde{\BE})_{\mathcal{H}}$
is unique and  given by $$\BE=-\Ba_{11}^{-1}\Ba_{10}\BE_0  \quad \mbox{ with }  \BE_0=\Pi(V_0).$$
Thus, using \eqref{ClassicSolnZProb}, one deduces that $\BE$ is precisely the field in the unique solution  $(\BJ_{0}, \BE, \BJ) \in \mathcal{U} \times \mathcal{E} \times \mathcal{J}$ of the Dirichlet $Z$-problem $(\mathcal{H}, \mathcal{U}, \mathcal{E}, \mathcal{J}, \Ba)$ corresponding to $\BE_{0} = \Pi (V_{0})  \in \mathcal{U}$.
Hence by Proposition \ref{Pro-link-DtNpDE-Z-problem}, one has 
\begin{gather*}
    \nabla u=\BE_0+\BE \ \mbox{ and }  \ \BJ_0+\BJ=\Ba \nabla u,
\end{gather*}
where $u$ is the unique solution to the Dirichlet boundary-problem (\ref{def.DirBVPForCoercivityTensors}). Thus, as
$$
( \Ba(\BE_0\!+\!\BE),\BE_0\!+\!\BE)_{\mathcal{H}} =( \Ba\nabla u,\nabla u )_{\mathcal{H}},
$$
one concludes that $u$ is a minimizer of the left-hand side of \eqref{eq.Dirichletpriciple}. 
Moreover,  this minimizer is unique since if $\tilde{u}\in H^1(\Omega)$ is a minimizer   of this problem, then by Lemma \ref{Lem-equiv-set}  and  the uniqueness  of the minimizer of the problem associated to the second line in equation \eqref{Dirichlet-minim-compo}, one has necessarily $\nabla{\tilde{u}}=\nabla u=\BE+\BE_0$. Thus, as $\Omega$ is connected, one has $\tilde{u}-u$ is constant on $\Omega$ and has zero trace on $\partial \Omega$, thus $u=\tilde{u}$ on $\Omega$ and therefore, one recovers that the minimizer of the left-hand side of \eqref{eq.Dirichletpriciple} is unique. This concludes the proof of the classical Dirichlet variational principle via the abstract theory of composite approach.

Finally, one additional application of the effective operator representation of the DtN map (Theorem \ref{thm.effopreprDtNmap}) is that, by apply
the Thomson minimization principle (see Theorem \ref{ThmClassicalThomMinPrin}), we have the following proposition on upper and lower bounds for this map.
\begin{Pro}\label{pro.ThomsVarPrincDtnMap}
Let $\Ba:\Omega \rightarrow M_d(\mathbb{C})$ be a matrix-valued function such that $\Ba\in M_d(L^{\infty}(\Omega))$ is uniformly positive definite [in the sense of  \eqref{def:UniformlyPosDefinite}]. Then the DtN operator $\Lambda_{\Ba}$ (as defined in Prop.\ \ref{pro.DtNDirBVPResults}) satisfies for any $V_0\in H^{\frac{1}{2}}(\partial\Omega)$ and $\BE_0=\Pi (V_0) \in \mathcal{U}$:
\begin{equation}\label{eq.ThomsonDtN}
    0\leq (\BE_0,[(\Ba^{-1})_{00}]^{-1}\BE_0 )_{\mathcal{H}}\leq \langle \Lambda_{\Ba}V_0,\overline{V_0} \rangle_{H^{-\frac{1}{2}}(\partial \Omega), H^{\frac{1}{2}}(\partial \Omega)}\leq ( \BE_0,\Ba_{00}\BE_0 )_{\mathcal{H}}.
\end{equation}
\end{Pro}
\begin{proof}
Let $V_0\in H^{\frac{1}{2}}(\partial \Omega)$. Then $\Lambda_{\Ba}=\Pi^{\dagger}\Ba_*\Pi$ (by Theorem \ref{thm.effopreprDtNmap}) which implies that
$$
\langle \Lambda_{\Ba}V_0,\overline{V_0} \rangle_{H^{-\frac{1}{2}}(\partial \Omega), H^{\frac{1}{2}}(\partial \Omega)}=( \Ba_*\BE_0,\BE_0 )_{\mathcal{H}}  \quad \mbox{ with }\BE_0=\Pi (V_0) \in \mathcal{U}.
$$
Hence,  the inequality \eqref{eq.ThomsonDtN} follows immediately from  the upper and lower bounds \eqref{ClassicalUpperLowerBddsEffOp} on  effective operators given in Theorem  \ref{ThmClassicalThomMinPrin}.
\end{proof}

\section{DtN map with affine boundary conditions}\label{sec:DiriPrbAffineBCsDtNMap}
Consider the Dirichlet boundary-value problem (\ref{def.DirBVPForCoercivityTensors}) with affine boundary conditions, i.e., the boundary-value $V_0\in H^{1/2}(\partial \Omega)$ is given by\footnote{The term ``affine" here might be confusing as the boundary-value $V_0$ is linear in $\Bx$, but adding any constant to \eqref{def.affineboundaryconditions} doesn't change the output of the DtN map \eqref{eq.Defdtnfreqdomain}. Thus we drop the constant part of the affine boundary condition.}
\begin{align}
    V_0=-{\bf e}_0\cdot \Bx |_{\partial\Omega},\label{def.affineboundaryconditions}
\end{align}
for fixed $\Be_0\in \mathbb{C}^d$. Our goal is to characterize the DtN operator $\Lambda_{\Ba}$ on these affine boundary conditions. This is carried out in the next section with a new effective operator representation for it. After that we are able to derive bounds on it in terms of spatial averaging when $\Ba$ is also uniformly positive definite [in the sense of \eq{def:UniformlyPosDefinite}]. The bounds that we derive were first formulated by Nemat-Nasser and Hori \cite[Sec.\ 2.7.3]{Nemat-Nasser:1993:MOP} following a comment in private communication by J. Willis (1989), and then further developed by G.\ W.\ Milton \cite{Milton:2011:UBE} to achieve tighter bounds incorporating more information about the inhomogeneous body $\Omega$, which we do not utilize here. 
The approach that we take here is different and is based on the abstract theory of composites as first formulated in \cite{Milton:2016:ETC}.

\subsection{DtN map for affine boundary conditions, averaging, and a new effective operator representation}\label{subsec:DtNMapAffBndCondsEffOpRepr}

Consider now the function $u(\Bx)=-{\bf e}_0\cdot \Bx, \forall \Bx\in \Omega$. It is the harmonic extension to $\Omega$ of $V_0$ in (\ref{def.affineboundaryconditions}) [since $u\in H^1(\Omega)$ with $\Delta u=0$ in $\Omega$ and $u=V_0$ on $\partial \Omega$]. Hence, in terms of the lift operator $\Pi$, we have
\begin{align}
    \Pi(V_0)=\Pi(-{\bf e}_0\cdot \Bx |_{\partial\Omega})=\nabla (-{\bf e}_0\cdot \Bx)=-{\bf e}_0=\BE_0\in\mathcal{U}.\label{LiftOpOnAffBCs}
\end{align}
In order to characterize the DtN map in terms of these affine boundary conditions, we need the following definition and lemma (whose proof is obvious and so is omitted).
\begin{Def}
Denote the average by
\begin{gather*}
    \langle \Bu \rangle = \frac{1}{|\Omega|}\int_{\Omega}\Bu(\Bx) \, \md \Bx,\ \forall \Bu\in \mathcal{H}  \quad \mbox{ where }  \ |\Omega|=\int_{\Omega}1\,  \md \Bx,
\end{gather*}
the space of constant (i.e., uniform) fields by $\langle\mathcal{U}\rangle$, i.e.,
\begin{gather*}
    \langle\mathcal{U}\rangle=\{\Bu\in \mathcal{H}\mid\Bu=\langle \mathbf{\Bu}\rangle\},
\end{gather*}
and define the average operator $\Gamma_{avg}:\mathcal{H}\rightarrow \langle\mathcal{U}\rangle$ by
\begin{align*}
    \Gamma_{avg}\mathbf{u}=\langle\mathbf{u} \rangle\in \langle\mathcal{U}\rangle,\;\forall \mathbf{u}\in \mathcal{H}.
\end{align*}
\end{Def}
\begin{Lem}\label{lem.avgoperatorisorthoproj}
The average operator $\Gamma_{avg}$ is the orthogonal projection of $\mathcal{H}$ onto the $d$-dimensional subspace $\langle\mathcal{U}\rangle\subseteq \mathcal{U}$.
\end{Lem}

The following two theorems are the main results of this section. The first result tells us that the operator $\Ba^D$, defined by \eqref{def.EffOpForDtNMapWithAffBCs}, is the compression of the operator $\Ba_*$ on $\mathcal{H}$ to the subspace $\langle\mathcal{U}\rangle$ (equivalently, $\Ba_*$ is a dilation of $\Ba^D$). The second result tells us that the $\Ba^D$ is an effective operator of a $Z$-problem.
\begin{Thm}\label{thm.EffOpForDtNMapWithAffBCs}
The DtN operator $\Lambda_{\Ba}: H^{\frac{1}{2}}(\partial \Omega) \to H^{-\frac{1}{2}}(\partial \Omega)$ and the effective operator $\Ba_*$ of the Dirichlet $Z$-problem $(\mathcal{H},\mathcal{U},\mathcal{E},\mathcal{J},\BL_\Ba)$ satisfy
\begin{gather*}
     \langle \Lambda_{\Ba}(-{\bf e}_0\cdot \Bx |_{\partial\Omega}),\overline{(-{\bf e}_0\cdot \Bx |_{\partial\Omega})} \rangle_{H^{-\frac{1}{2}}(\partial \Omega), H^{\frac{1}{2}}(\partial \Omega)}=(\Ba_*{\bf e}_0,{\bf e}_0 )_{\mathcal{H}}=( \Ba^D{\bf e}_0,{\bf e}_0 )_{\mathcal{H}},\ \forall \Be_0\in \mathbb{C}^d,
\end{gather*}
where $\Ba^D:\langle\mathcal{U}\rangle\rightarrow\langle\mathcal{U}\rangle$ is the bounded linear operator
\begin{align}
    \Ba^D=\Gamma_{avg}\Ba_*\Gamma_{avg}|_{\langle\mathcal{U}\rangle}\in \mathcal{L}(\langle\mathcal{U}\rangle).\label{def.EffOpForDtNMapWithAffBCs}
\end{align}
\end{Thm}
\begin{proof}
First, by Proposition \ref{pro.ExistenceUniquenessDirichletZProbEffOp} and Lemma \ref{lem.avgoperatorisorthoproj}, it follows that $\Ba^D=\Gamma_{avg}\Ba_*\Gamma_{avg}|_{\langle\mathcal{U}\rangle}\in \mathcal{L}(\langle\mathcal{U}\rangle)$ with $\Gamma_{avg}^\dagger=\Gamma_{avg}=\Gamma_{avg}^2$. Now, as usual, we identify the elements of $\mathbb{C}^d$ with their corresponding constant functions in $\langle\mathcal{U}\rangle$. Hence, if $\Be_0\in \mathbb{C}^d$ then by this identification we have $\Gamma_{avg}\Be_0=\Be_0$ and so by the equality (\ref{LiftOpOnAffBCs}) together with the identity $\Lambda_{\Ba} =\Pi^\dagger\Ba_*\Pi$ (by Theorem \ref{thm.effopreprDtNmap}), it follows that
\begin{eqnarray*}
    \langle \Lambda_{\Ba}(-{\bf e}_0\cdot \Bx |_{\partial\Omega}),\overline{(-{\bf e}_0\cdot \Bx |_{\partial\Omega})} \rangle_{H^{-\frac{1}{2}}(\partial \Omega), H^{\frac{1}{2}}(\partial \Omega)}&=&(\Ba_*\Pi(-{\bf e}_0\cdot \Bx |_{\partial\Omega}),\Pi(-{\bf e}_0\cdot \Bx |_{\partial\Omega}) )_{\mathcal{H}}\\&=&(\Ba_*{\bf e}_0,{\bf e}_0 )_{\mathcal{H}}=( \Ba^D{\bf e}_0,{\bf e}_0 )_{\mathcal{H}}.
\end{eqnarray*}
This completes the proof.
\end{proof}

\begin{Thm}\label{thm.aDisaneffoperator}
If $\Ba^D$ is the operator as defined by (\ref{def.EffOpForDtNMapWithAffBCs}) and $\Ba_{*^D}$ is the effective operator of the $Z$-problem $(\mathcal{H}^D,\mathcal{U}^D, \mathcal{E}^D,\mathcal{J}^D,\Ba)$, where
\begin{gather*} 
\mathcal{H}^D =\mathcal{H},\;\mathcal{U}^D=\langle\mathcal{U}\rangle,\;\mathcal{E}^D=\mathcal{E},\\
\mathcal{J}^D=\big(\mathcal{U}\overset{\perp}{\ominus}\langle\mathcal{U}\rangle\big)\overset{\perp}{\oplus}\mathcal{J}=\{ \BJ\in H_{div }(\Omega)  \mid  \nabla \cdot \BJ=0 \mbox{ and }\langle \BJ\rangle=0 \}
\end{gather*}
(here $\mathcal{U}\overset{\perp}{\ominus}\langle\mathcal{U}\rangle$ denotes the orthogonal complement of $\langle\mathcal{U}\rangle$ in $\mathcal{U})$, then
\begin{align*}
    \Ba^D=\Ba_{*^D}.
\end{align*}
\end{Thm}
\begin{proof}
First, by Theorem \ref{th.tpro} and Lemma \ref{lem.avgoperatorisorthoproj} it follows that
\begin{align*}
   \mathcal{H}=\langle\mathcal{U}\rangle \overset{\perp}{\oplus} \mathcal{E} \overset{\perp}{\oplus} \Big[\big(\mathcal{U}\overset{\perp}{\ominus}\langle\mathcal{U}\rangle\big)\overset{\perp}{\oplus}\mathcal{J}\Big].
\end{align*}
This proves that $(\mathcal{H}^D,\mathcal{U}^D, \mathcal{E}^D,\mathcal{J}^D, \Ba)$ is a $Z$-problem, where $\mathcal{J}^D=\big(\mathcal{U}\overset{\perp}{\ominus}\langle\mathcal{U}\rangle\big)\overset{\perp}{\oplus}\mathcal{J}$  can be characterized in the following way  $$\mathcal{J}^D=\operatorname{kernel}(\nabla \cdot|_{H_{div}(\Omega)})\cap \operatorname{range}(\Gamma_{avg})=\{ \BJ\in H_{div}(\Omega) \mid  \nabla \cdot \BJ=0  \mbox{ and } \langle \BJ\rangle=0 \},$$ since $\mathcal{U} \oplus \mathcal{J} $ is the kernel of divergence operator [see relation  \eqref{KerDivIsUPlusJ} of Theorem \ref{th.tpro}].
Next, as $\mathcal{E}^D=\mathcal{E}$, it follows from Proposition \ref{pro.ExistenceUniquenessDirichletZProbEffOp} that Theorem \ref{ThmMainClassicalZProbEffOp} holds for both the Dirichlet $Z$-problem $(\mathcal{H},\mathcal{U}, \mathcal{E},\mathcal{J}, \Ba)$ as well as the $Z$-problem $(\mathcal{H}^D,\mathcal{U}^D, \mathcal{E}^D,\mathcal{J}^D, \Ba)$. We claim that $\Ba^D=\Ba_{*^D}.$ Let $\BE_0\in \mathcal{U}^D$. Then, by Theorem \ref{ThmMainClassicalZProbEffOp}, there exists a $(\BJ_0,\BE, \BJ)\in \mathcal{U}^D\times \mathcal{E}^D\times \mathcal{J}^D$ such that
\begin{align*}
    \BJ_0+\BJ=\Ba(\BE_0+\BE),\;\BJ_0= \Ba_{*^D}\BE_0.
\end{align*}
On the other hand, $\BE\in \mathcal{E}^D=\mathcal{E}$, $\BE_0,\BJ_0\in \mathcal{U}$ (as $\mathcal{U}^D=\langle\mathcal{U}\rangle\subseteq \mathcal{U}$) and $(\BI-\Gamma_2)\BJ\in \mathcal{U}\overset{\perp}{\ominus}\langle\mathcal{U}\rangle$, where $\Gamma_2$ is the orthogonal projection of $\mathcal{H}$ onto $\mathcal{J}$, so that
\begin{align*}
    [\BJ_0+(\BI-\Gamma_2)\BJ]+\Gamma_2\BJ=\Ba(\BE_0+\BE),
\end{align*}
implying by Theorem \ref{ThmMainClassicalZProbEffOp} that
\begin{align*}
    \BJ_0+(\BI-\Gamma_2)\BJ=\Ba_*(\BE_0).
\end{align*}
It follows from this and
\begin{align*}
    \Gamma_{avg}\BE_0=\BE_0,\;\Gamma_{avg}[\BJ_0+(\BI-\Gamma_2)\BJ]=\BJ_0,
\end{align*}
that we have
\begin{align*}
\Ba^D(\BE_0)=\Gamma_{avg}\Ba_*\Gamma_{avg}|_{\langle\mathcal{U}\rangle}(\BE_0)=\BJ_0=\Ba_{*^D}(\BE_0).
\end{align*}
As this is true for every $\BE_0\in \mathcal{U}^D=\langle\mathcal{U}\rangle,$ we have proven the claim. This proves the theorem.
\end{proof}
It should remarked that this theorem and its proof has an analogy in the abstract theory of composites on certain operations that can be done on effective operators and their associated $Z$-problems, for instance, see Eq.\ (29.1) in \cite[Sec.\ 29.1]{Milton:2022:TOC} and \cite[Chap.\ 7, Sec.\ 9]{Milton:2016:ETC}.

\subsection{Elementary bounds on \texorpdfstring{$\Ba^D$}{\textbf{a}hatD} in terms of the averaged tensors \texorpdfstring{$\langle \Ba \rangle$}{<\textbf{a}>} and \texorpdfstring{$\langle \Ba^{-1} \rangle^{-1}$}{<\textbf{a}hat-1>hat-1}}\label{subsec:ElemBndsAvgLocalTensors}
The next theorem is the main result of this section. These inequalities are known (see, for instance, \cite{Nemat-Nasser:1995:UBO}, \cite[Eqs.\ (3.4) and (3.6)]{Milton:2011:UBE}), but we will give here a new and simple proof using the abstract theory of composites that we believe is an important result in its own right. 
\begin{Thm}\label{Thm:SigmaDBndsOfMilton}
If $\Ba$ is uniformly positive definite [in the sense of \eq{def:UniformlyPosDefinite}] then the following operator inequalities hold:
\begin{align}\label{eq.operatorineqaffine}
    0\leq\langle \Ba^{-1}\rangle^{-1}\leq \Ba^D\leq  \langle \Ba\rangle,
\end{align}
where
\begin{align}\label{eq.average op}
    \langle \Ba\rangle=\Gamma_{avg}\Ba\Gamma_{avg}|_{\langle\mathcal{U}\rangle},\;\langle \Ba^{-1}\rangle=\Gamma_{avg}\Ba^{-1}\Gamma_{avg}|_{\langle\mathcal{U}\rangle}.
\end{align}
Moreover, for each $\Be_0\in \mathbb{C}^d$,
\begin{gather}\label{eq.boundefftnesor}
    0\leq ( \langle \Ba^{-1}\rangle^{-1}{\bf e}_0,{\bf e}_0 )_{\mathcal{H}}\leq  \langle \Lambda_{\Ba}(-{\bf e}_0\cdot \Bx |_{\partial\Omega}),\overline{(-{\bf e}_0\cdot \Bx|_{\partial\Omega})} \rangle_{H^{-\frac{1}{2}}(\partial \Omega), H^{\frac{1}{2}}(\partial \Omega)}\leq  (\langle \Ba\rangle{\bf e}_0,{\bf e}_0 )_{\mathcal{H}}.
\end{gather}
\end{Thm}
\begin{proof}
First, it is clear that to prove this theorem we need only prove the operator inequalities \eqref{eq.operatorineqaffine} as the remaining inequalities then follow immediately from those by Theorem \ref{thm.EffOpForDtNMapWithAffBCs}. Now by Theorem \ref{thm.aDisaneffoperator} we know that $\Ba^D=\Ba_{*^D}$ is the effective operator of the $Z$-problem  $(\mathcal{H}^D,\mathcal{U}^D,\mathcal{E}^D,\mathcal{J}^D,\Ba)$ and so by the Thomson variational principle (Theorem \ref{ThmClassicalThomMinPrin}) applied to this $Z$-problem we have that
\begin{align*}
    0 \leq (\Gamma_{avg}\Ba^{-1}\Gamma_{avg}|_{\langle\mathcal{U}\rangle})^{-1}\leq \Ba_{*^D}\leq  \Gamma_{avg}\Ba\Gamma_{avg}|_{\langle\mathcal{U}\rangle}.
\end{align*}
This proves the theorem.
\end{proof}
Notice that the upper and lower bounds \eqref{eq.boundefftnesor} are analogues of the bounds \eqref{eq.ThomsonDtN} in Prop.\ \ref{pro.ThomsVarPrincDtnMap} and obtain by also using the Thomson variational principle although now applied to the new $Z$-problem  $(\mathcal{H}^D,\mathcal{U}^D,\mathcal{E}^D,\mathcal{J}^D,\Ba)$.
Compared to the bounds \eqref{eq.ThomsonDtN}, they have the advantage to provide a more explicit dependence on the material tensor $\Ba$ via averaging, but they are less general since they are restricted to affine boundary conditions.

\section{Analyticity properties  of the DtN map}\label{sec:AnalyticPropertiesDtNMap}

In this section we discuss the analytic properties of the DtN map as a function of frequency $\omega$. First, in Sec.\ \ref{sec:DiffHerglotzFuncsDtNMapsIsHerglotz}, we study its quadratic form  and prove, for a fixed surface potential $V_0$, that it is a Herglotz function in $\omega$. Second, in Sec.\ \ref{sec-Herg-Stielt} we recall from \cite{Cassier:2017:BHF} how to construct a Herglotz function from a Stieltjes function and use this in Sec.\ \ref{sec-Herg-phys} to connect this quadratic form to a Stieltjes function.

\subsection{A Herglotz function associated with the difference of two DtN maps}\label{sec:DiffHerglotzFuncsDtNMapsIsHerglotz}
 This section makes the connection between the quadratic form of the DtN map of our  passive device (as  a function of the frequency) and a well-known class of analytic functions: the Herglotz functions which are intensively used in the scientific literature to model passive systems \cite{Zemanian:2005:RTC,Bernland:2011:SRC, Milton:2016:ETC, Cassier:2017:MMD}. Furthermore, the mathematical properties of this class of functions
(integral representations \cite{Akhiezer:1993:TLO,Berg:2008:SPBS,Gesztesy:2000:MVH}, sum rules \cite{Bernland:2011:SRC,Gustafsson:2010:SRP,Welters:2014:SLL,Cassier:2017:BHF}, continued fraction expansions \cite{Milton:1987:MCEa, Milton:1987:MCEb,Pivovarchik:2013:DNI}) associated with complex analysis appear as key tool to derive quantitative bounds on the physical properties  of a passive systems on
a finite bandwidth (e.g., lensing \cite{Lind-Johansen:2009:PLF}, scattering \cite{Sohl:2007:PLB},  cloaking \cite{Cassier:2017:BHF}, propagation speed of electromagnetic waves \cite{Welters:2014:SLL}). Concerning electromagnetic systems, a first connection was made between the DtN map and Herglotz function based on a variational approach \cite{Cassier:2016:ETC}. Here, we adopt a another  point of view based on the abstract theory of composites to establish this bridge. This allows us  to use the DtN map representation as an effective operator (see Theorem \ref{thm.effopreprDtNmap}) and variational bounds based on Schur complements techniques  to prove stronger result such as Theorem \ref{thm.differenceHerg}. This theorem,  which compares the quadratic forms of the DtN maps of the  cloak to the uncloaked device, enables us to prove an unintuitive result that their difference remains a Herglotz function.

\begin{Pro}\label{Pro:DtNMapQuadFormIsHerglotz}
If the cloaking device satisfies assumptions \textnormal{H1} and \textnormal{H3--H6} then the function $\omega\mapsto \Lambda_{\omega \BGve(\cdot,\Go)}$ from $\mathbb{C}^+$ into $\mathcal{L}(H^{1/2}(\partial\Omega),H^{-1/2}(\partial\Omega))$ is an analytic operator-valued function. Moreover, for each $V_0\in H^{1/2}(\partial \Omega)$, the scalar function
\begin{align}\label{eq.DefHerglotzcloakingdevice}
 h_{V_0} :  \omega\mapsto\langle \Lambda_{\omega \BGve(\cdot,\omega)}V_0,\overline{V_0} \rangle_{H^{-\frac{1}{2}}(\partial \Omega), H^{\frac{1}{2}}(\partial \Omega)},
\end{align}
defined on $\mathbb{C}^+\cup [\Go_-, \Go_+]$, is a Herglotz function which is continuous on $[\omega_-, \omega_+]$.

\end{Pro}
\begin{proof}
Our proof is broken into four steps.\\[6pt]
{\bf Step 1:  Preliminaries}\\[4pt]
First, by \eqref{eq.permitivtyob} and assumptions H1, H3 we know that Lemma \ref{Lem.comptnesor-vacuum} holds so together with assumptions H5, H6, we have by Proposition \ref{pro.DtNDirBVPResults} that the function $\omega \mapsto \Lambda_{\omega \BGve(\cdot,\omega)}$  from  $\bbC^+\cup [\omega_-,\omega_+]$  to 
$\mathcal{L}( H^{\frac{1}{2}}(\partial \Omega), H^{-\frac{1}{2}}(\partial \Omega))$ is well-defined. Second, those assumptions imply the representation formula \eqref{eq.factDtNmap} of Theorem \ref{thm.effopreprDtNmap} of $\Lambda_{\omega \BGve (\cdot,\omega)}$ holds on $\bbC^+\cup [\omega_-,\omega_+]$ with $\Ba=\omega \BGve(\cdot,\omega)$ and the effective tensor $\Ba_*=\big(\omega \BGve(\cdot,\omega)\big)_*$ is given by the Schur complement formula \eqref{eq.Schurcomplement1} and \eqref{eq.Schurcomplement2}.\\[6pt]
\noindent 
{\bf Step 2: Analyticity of $\omega\mapsto\Lambda_{\omega \BGve(\cdot,\omega)}$}\\[4pt]
Next, we will show that $\omega\mapsto\Lambda_{\omega \BGve(\cdot,\omega)}$ is analytic operator-valued function on $\bbC^+$, from which it follows that the scalar function  $h_{V_0}$ is analytic on $\bbC^+$. 
To this aim, we first prove that the operator $\BL_{\omega \BGve(\cdot,\omega)}$ corresponding to left multiplication by $\omega \BGve(\cdot,\omega)$ is an analytic operator-valued function from $\mathbb{C}^+$ into $\mathcal{L}(\mathcal{H})$. Combining   \cite[Theorem 3.12]{Kato:1995:PTO} (which states the equivalence between weak analyticity and  analyticity  for the operator norm topology) alongside the polarization identity for sesquilinear forms \cite[p.\ 25]{Teschl:2014:MMQ}, shows this is equivalent to proving the analyticity of the function $$\omega \mapsto (\omega \BGve(\cdot,\omega) \BF, \BF )_{\mathcal{H}} \mbox{ on } \bbC^+, \mbox{ for each } \BF\in \mathcal{H}.$$  As $\omega  \mapsto \BGve(\Bx,\omega)$ is analytic for a.e.\ $\Bx\in \Omega$ (by the passivity assumption H1), one proves this analyticity by applying the theorem of complex differentiation under the integral presented in \cite{Mattner:2001:CDI} (using the hypothesis H5 for the domination condition, i.e., assumption A3 in \cite[Theorem]{Mattner:2001:CDI}, required in its assumptions).  As $L_{\omega \BGve(\cdot,\omega)}$, which we also denote by $\omega \BGve(\cdot,\omega)$ (see Def.\ \eqref{DefKZProb} on this abuse of notation), is analytic then this implies that the operator-valued functions $\omega \mapsto \omega \BGve(\cdot,\omega)_{ij}=\Gamma_i\omega \BGve(\cdot,\omega)\Gamma_j\in \mathcal{L}(\operatorname{range}(\BGG_j), \operatorname{range}(\BGG_i))$ are analytic on $\bbC^+$. As $\omega \mapsto \omega\BGve(\cdot,\omega)_{11}\in \mathcal{L}(\mathcal{E})$ is analytic on $\bbC^+$ and since $\omega\BGve(\cdot,\omega)_{11}$ is an invertible operator as an element of $\mathcal{L}(\mathcal{E})$ for each $\omega \in \mathbb{C}^+$ (from Lemma \ref{lem.uniformcoerciveimpliessolvabilityofZproblem} using Lemma \ref{Lem.comptnesor-vacuum}), then the map from $\omega$ to the inverse operator, i.e., $\omega \mapsto [\omega \BGve(\cdot,\omega)_{11}]^{-1} \in \mathcal{L}(\mathcal{E})$, is analytic on $\mathbb{C}^+$ (see, for instance, \cite[Chap.\ 7, Sec.\ 1, pp.\ 365--366]{Kato:1995:PTO}). This proves the Schur complement function $\omega \mapsto \big(\omega \BGve(\cdot,\omega)\big)_*\in \mathcal{L}(\mathcal{U})$ [see \eqref{eq.Schurcomplement1} and \eqref{eq.Schurcomplement2}] is an analytic operator-valued function  on $\bbC^+$. Thus,  as  the lift operator $\Pi$ and its adjoint $\Pi^\dagger$ are bounded and frequency-independent, it implies with \eqref{eq.factDtNmap} that  $\omega\mapsto\Lambda_{\omega \BGve(\cdot,\omega)}$ is analytic on $\bbC^+$.\\[6pt]
{\bf Step 3:  $h_{V_0}$ is an Herglotz function}\\[4pt]
To prove that the function $h_{V_0}$  has a positive semidefinite imaginary part on $\bbC^+$, we connect the DtN map to the effective tensor of the corresponding Dirichlet $Z$-problem using Theorem \eqref{thm.effopreprDtNmap}. 
Let $\BE_0=\Pi V_0$. Then, since $\BE_0\in \mathcal{U}$, for a fixed  $\omega \in \bbC^+$, there exists a solution $(\BJ_0,\BE,\BJ)\in \mathcal{U}\times\mathcal{E}\times\mathcal{J}$ to the Dirichlet $Z$-problem $(\mathcal{H},\mathcal{U},\mathcal{E},\mathcal{J},\bf a)$, where $\bf a = \omega \BGve(\cdot,\omega)$, at $\BE_0$ which implies  that
    \begin{eqnarray*}
      \langle \Lambda_{\omega \BGve(\cdot,\omega)}V_0,\overline{V_0} \rangle_{H^{-\frac{1}{2}}(\partial \Omega), H^{\frac{1}{2}}(\partial \Omega)}&=& ( [\omega \BGve(\cdot,\omega)]_*\Pi V_0,\Pi V_0 )_{\mathcal{H}} \\ &=&( [\omega \BGve(\cdot,\omega)]_*\BE_0, \BE_0 )_{\mathcal{H}}\\
        &=&( [\omega \BGve(\cdot,\omega)](\BE_0+\BE), \BE_0+\BE )_{\mathcal{H}}.
    \end{eqnarray*}
  (in the above equalities, we use that $\BJ_0=[\BGve(\cdot,\omega)]_*\BE_0$, $\BJ_0+\BJ=\omega \BGve(\cdot,\omega) (\BE_0+\BE)$ and the fact that $\mathcal{U},\mathcal{E},\mathcal{J}$ are orthogonal to each other).  
Thus, it follows that
\begin{eqnarray*}
        \operatorname{Im} \langle \Lambda_{\omega \BGve (\cdot,\omega)}V_0,\overline{V_0} \rangle_{H^{-\frac{1}{2}}(\partial \Omega), H^{\frac{1}{2}}(\partial \Omega)}&=&( \mathfrak{I}[\omega \BGve (\cdot,\omega)](\BE_0+\BE),\BE_0+\BE )_{\mathcal{H}}\geq 0,
            \end{eqnarray*}
where the last inequality comes from  the passive assumption H1 on the cloaking device which implies that for a.e.\ $\Bx\in \Omega$, $\omega\mapsto\omega \BGve(\Bx,\omega)$ is a matrix-valued Herglotz-function and thus $\mathfrak{I}[\omega \BGve(\Bx,\omega)]\geq 0$ for $\omega\in \bbC^+$. This proves that  $h_{V_0}$ is a scalar-valued Herglotz function.\\[6pt]
{\bf Step 4:  Continuity of $h_{V_0}$ on $[\omega_-, \omega_+]$}\\[4pt]
 Finally, we show the continuity of the map $$\omega  \mapsto \Lambda_{\omega \BGve (\cdot,\omega)}V_0=\Pi^\dagger (\omega \BGve(\cdot,\omega))_* \Pi V_0$$ in the $H^{-1/2}$-norm on $[\omega_-, \omega_+]$. This continuity implies, in particular, the continuity of $h_{V_0}$  on $[\omega_-, \omega_+]$. As $\Pi$ and its adjoint $\Pi^\dagger$ are bounded operators which do not depend on the frequency $\Go$, it is sufficient to prove the continuity of the map $\Go\mapsto (\omega \BGve(\cdot,\omega))_* \BE_0$  for the  $\mathcal{H}$-norm for any $\BE_0\in \mathcal{U} $ on $[\omega_-, \omega_+]$. 
To this aim, one uses the formula \eqref{ZProbEffOpAltFormulaViaInv} in Corollary \ref{cor:BasicPropertiesEffOpDiricZProb}
[which holds due to Lemma \ref{Lem.comptnesor-vacuum} and the coercivity assumption H6 on the tensor $\omega \BGve(\cdot,\omega)$], namely, for any fixed    $\BE_0\in \mathcal{U}$:
\begin{equation}\label{eq.continuity}
\big(\Go\,\BGve(\cdot,\omega)\big)_* \BE_0= \big[ [\Go\,\BGve(\cdot,\omega)_{ij}]_{i,j=0,1}^{-1}\big]^{-1}_{00}\BE_0, \quad \forall \Go \in \bbC^+\cup[\omega_-,\omega_+],
\end{equation}
[where $\Go\,\BGve(\cdot,\omega)$ corresponds to the operator   $\BL_{\Go\,\BGve(\cdot,\omega)}$ in $\mathcal{L}(\mathcal{H})$].
Let $(\Go_n)_{n\in \mathbb{N}}$ be a sequence of $\bbC^{+}\cup [\omega_-, \omega_+]$ which tends to $\omega \in [\omega_-, \omega_+]$. 
First, the sequence of elements $[\Go_n\BGve(\cdot,\omega_n)_{ij}]_{i,j=0,1}^{-1}\in \mathcal{L}(\mathcal{U}\overset{\perp}{\oplus} \mathcal{E})$ is bounded since by the coercivity assumption H6 there exists $\delta>0$ such that 
\begin{equation}
    \label{eq.uniformboundinv}
\forall \omega_n \in B(\omega, \delta)\cap (\bbC^+ \cup [\omega_-, \omega_+]), \, \|[\Go_n\,\BGve(\cdot,\omega_n)_{ij}]_{i,j=0,1}^{-1}\|_{\mathcal{L}(\mathcal{U}\overset{\perp}{\oplus} \mathcal{E})}\leq c_2(\omega)^{-1},
\end{equation}
where $c_{2}(\omega)$ is the coercivity constant appearing in H6. 
Furthermore, for each $\tilde{\BE}\in \mathcal{U}\overset{\perp}{\oplus} \mathcal{E}$,
\begin{equation}\label{eq.strongconv}
\big\|\big(\big[ \omega_n\BGve(\cdot, \omega_n)_{ij}\big]_{i,j=0,1}-  \big[ \Go\,\BGve(\cdot,\omega)_{ij}\big]_{i,j=0,1}\big) \tilde{\BE}  \big\|_{\mathcal{H}}= \big\| \big[\omega_n\BGve(\cdot,\omega_n)-  \Go\,\BGve(\cdot,\omega) \big]\tilde{\BE}\big\|_{\mathcal{H}}\to 0
\end{equation}
 as $n\to + \infty$,
where the limit is obtained via the Lebesgue's dominated convergence theorem (the application of this theorem is justified by the assumptions H4 and H5). Thus, combining \eqref{eq.uniformboundinv} and \eqref{eq.strongconv} together with the Lemma \ref{Lem.strongconv} yields the following strong convergence: For each $\tilde{\BE}\in \mathcal{U}\overset{\perp}{\oplus} \mathcal{E}$,
\begin{equation*}
\| \big( [\Go_n\,\BGve(\cdot,\omega_n)_{ij}]_{i,j=0,1}^{-1}  -[\Go\,\BGve(\cdot,\omega)_{ij}]_{i,j=0,1}^{-1} \big)\, \tilde{\BE} \|_{\mathcal{H}}\to 0 \mbox{ as } n \to +\infty.
\end{equation*}
Thus, it follows that
for any $\BE_0\in \mathcal{U}$:
\begin{equation}\label{eq.strongconv2}
\| \big( \big[[\Go_n\,\BGve(\cdot,\omega_n)_{ij}]_{i,j=0,1}^{-1} \big]_{00} -[\Go\,\BGve(\cdot,\omega)_{ij}]_{i,j=0,1}^{-1} ]_{00} \big)\, \BE_0\|_{\mathcal{H}}\to 0 \mbox{ as } n \to +\infty.
\end{equation}
Then, as $\big(\omega_n \BGve(\cdot,\omega_n)\big)_*$ is also given by the Schur complement formula \eqref{eq.Schurcomplement1} and \eqref{eq.Schurcomplement2}, this yields [using \eqref{eq.continuity} and  the assumptions H5 and H6 along with the fact that the orthogonal projections $\Gamma_j$, for $j=0,1,2$, all have norm $1$] that there exists $\delta>0$ such that for all $ \omega_n \in B(\omega, \delta)\cap (\bbC^+ \cup [\omega_-, \omega_+])$:
\begin{equation}\label{eq.uniformboundinv2}
\|\big[[\Go_n\,\BGve(\cdot,\omega_n)_{ij}]_{i,j=0,1}^{-1} \big]_{00}^{-1}\|_{\mathcal{L}(\mathcal{U})}=\|  \big(\omega_n  \BGve(\cdot,\omega_n)\big)_* \|_{\mathcal{L}(\mathcal{U})}\leq  c_1(\omega)+  c_1(\omega)^2 c_2(\omega)^{-1},
\end{equation}
where $c_1(\omega)$ is the constant from assumption H5.
Combining \eqref{eq.strongconv2}, \eqref{eq.uniformboundinv2} together with Lemma \ref{Lem.strongconv} yields that
for any $\BE_0\in \mathcal{U}$:
\begin{equation*}
\| \big( \big[[\Go_n\,\BGve(\cdot,\omega_n)_{ij}]_{i,j=0,1}^{-1} \big]_{00}^{-1} -[\Go\,\BGve(\cdot,\omega)_{ij}]_{i,j=0,1}^{-1} ]_{00}^{-1} \big)\, \BE_0\|_{\mathcal{H}}\to 0 \mbox{ as } n \to +\infty.
\end{equation*}
This implies, by virtue of \eqref{eq.continuity}, the continuity (for any $\BE_0\in \mathcal{U} $) of  $\Go\mapsto (\omega \BGve(\cdot,\omega))_* \BE_0$  for the  $\mathcal{H}$-norm on $[\omega_-, \omega_+]$. 
    \end{proof}

Now, we introduce an additional modeling assumption, that is not mandatory for all results. 
\begin{itemize}
    \item[H7] {\bf Reciprocity principle:} The cloak is made of a reciprocal medium, i.e., its permittivity satisfies: 
$$
\mbox{for a.e. } \Bx \in \Omega\setminus \CO, \ 
\forall \Go \in [\Go_-,\Go_+] \cup \bbC^+, \  \BGve(\Bx,\Go)^{\top}=\BGve(\Bx,\Go).$$ 
\end{itemize}
    
 \noindent   We are now ready to give our main result of this section.
    To this aim, we introduce the matrix-valued permittivity
    $\BGve_{\infty}(\cdot)$  defined by: 
    \begin{equation}\label{eq.defGVeinf}
    \mbox{ for  a.e. } \Bx \in \mathcal{O},\, \BGve_{\infty}(\Bx)= \BGve_{\mathrm{ob}}(\Bx) \ \mbox{ and } \
    \mbox{ for  a.e. } \Bx \in \Omega\setminus \mathcal{O}, \, \BGve_{\infty}(\Bx)= \varepsilon_0 \BI .
    \end{equation}

In the next theorem, we point out that assumption H2 is only needed in order to prove statement 2.
    \begin{Thm}\label{thm.differenceHerg}
Suppose the cloaking device satisfies assumptions \textnormal{H1--H6} and define,  for each fixed $V_0\in H^{1/2}(\partial \Omega)$, the scalar function  on $[\Go_-,\Go_+] \cup \bbC^+$ by
    \begin{align}
 h_{\operatorname{dif},V_0} :      \omega\mapsto\langle [\Lambda_{\omega\BGve(\cdot,\Go)}-\Lambda_{\omega\varepsilon_0}]V_0,\overline{V_0} \rangle_{H^{-\frac{1}{2}}(\partial \Omega), H^{\frac{1}{2}}(\partial \Omega)}.
    \end{align}
Then it has the following properties:
\begin{enumerate}
\item $h_{\operatorname{dif},V_0}$ is a Herglotz function which is continuous on $[\omega_-, \omega_+]$.

\item If the cloaking device is made of a reciprocal material, i.e.,  if assumption  $\mathrm{H}7$ holds then
\begin{equation}\label{eq.Hergsym}
    h_{\operatorname{dif},V_0}(-\overline{\omega})=-\overline{h_{\operatorname{dif},V_0}(\omega)}, \quad \forall \Go \in \bbC^+. 
\end{equation}
If the cloaking device contains some non-reciprocal material, i.e., if assumption $\mathrm{H}7$ is not satisfied, then \eqref{eq.Hergsym} holds whenever 
the potential $V_0\in H^{\frac{1}{2}}(\partial \Omega)$ is a real-valued function.
\item   \, $\forall \Go \in [\Go_-,\Go_+] \cup \bbC^+$, \begin{equation}\label{eq.homogDtnrelation}
\frac{h_{\operatorname{dif},V_0} (\Go) }{\Go}=\langle [\Lambda_{\BGve(\cdot,\Go)}-\Lambda_{\varepsilon_0}]V_0,\overline{V_0} \rangle_{H^{-\frac{1}{2}}(\partial \Omega), H^{\frac{1}{2}}(\partial \Omega)}= F_{V_0}(\Go).
\end{equation}
\item  $h_{\operatorname{dif},V_0}$ admits the following non-negative limit on the positive imaginary axis:
\begin{equation}\label{eq.limitinfherg}
F_{V_0}(\mathrm{i} y)=\frac{h_{\operatorname{dif},V_0} (\mathrm{i} y)}{\mathrm{i} y} \to F_{V_0,\infty}:=\langle[ \Lambda_{\BGve_{\infty}(\cdot)}-\Lambda_{\varepsilon_0}]V_0,\overline{V_0} \rangle_{H^{-\frac{1}{2}}(\partial \Omega), H^{\frac{1}{2}}(\partial \Omega)} \geq 0
\end{equation}
 as $y\to +\infty$, where the matrix-valued permittivity function $\BGve_{\infty}(\cdot)$ is defined in \eqref{eq.defGVeinf}.
 \end{enumerate}
\end{Thm}
    \begin{proof}
Our proof is broken into four steps.\\[6pt]
{\bf{Step 1}: Proof of assertion 1}\\[4pt]
First, we point out that $h_{\operatorname{dif},V_0}$ is analytic 
on $\bbC^+$ and continuous on $[\omega_-,\omega_+]$ as it is the difference of two functions, namely, $h_{V_0}$ [defined by  \eqref{eq.DefHerglotzcloakingdevice}] and
$$h_{\operatorname{vac}, V_0}: \Go \mapsto \langle \Lambda_{\omega\varepsilon_0} V_0,\overline{V_0} \rangle_{H^{-\frac{1}{2}}(\partial \Omega), H^{\frac{1}{2}}(\partial(\Omega))} \ \mbox{ defined on } \ [\Go_-,\Go_+] \cup \bbC^+,$$ which share this regularity by Proposition \ref{Pro:DtNMapQuadFormIsHerglotz}. Note that in Proposition \ref{Pro:DtNMapQuadFormIsHerglotz} we prove this regularity for  the function $h_{V_0}$ associated to the permittivity tensor $\BGve(\cdot,\omega)$ of the cloaking device  which contains the obstacle. However, it is straightforward to check that  such result  holds also for the function $h_{\operatorname{vac}, V_0}$ (when $\Omega$ is filled by vacuum) via the same proof  using that assumptions H1-H6 hold  if one replaces $\BGve(\cdot,\omega)$ by the permittivity tensor of the vacuum $\varepsilon_0 \BI$.  

The main point of this step is to show now that the difference of the two Herglotz functions $h_{V_0}$ and $h_{\operatorname{vac}, V_0}$ remains a Herglotz function (which is not obvious). The key tools for the proof of this are based on the variational principles associated to the abstract theory of composites (see Section \ref{Sec:VarPrincBndsEffOps}).  
Let $\omega\in \mathbb{C}^+$. Then we have
    \begin{align}\label{eq.diff}
       \operatorname{Im}h_{\operatorname{dif},V_0} (\Go)=\operatorname{Im}\langle \Lambda_{\omega\BGve(\cdot,\Go)}V_0,\overline{V_0} \rangle_{H^{-\frac{1}{2}}(\partial \Omega), H^{\frac{1}{2}}(\partial \Omega)}-\operatorname{Im}\langle \Lambda_{\omega\varepsilon_0}V_0,\overline{V_0} \rangle_{H^{-\frac{1}{2}}(\partial \Omega), H^{\frac{1}{2}}(\partial \Omega)}.
    \end{align}
Defining  $\BE_0=\Pi V_0\in \mathcal{U}$ and using the  representations \eqref{eq.factDtNmap}  of the DtN operators, namely $\Lambda_{\omega\BGve(\cdot,\Go)}= \Pi^\dagger [\omega\BGve(\cdot,\Go)]_* \Pi$ and  $\Lambda_{\omega\BGve_0(\omega,\cdot)}= \Pi^\dagger [\omega \BGve_0(\omega,\cdot)]_* \Pi$, yields the equalities
    \begin{eqnarray*}
\operatorname{Im}\langle \Lambda_{\omega\BGve(\cdot,\Go)}V_0,\overline{V_0} \rangle_{H^{-\frac{1}{2}}(\partial \Omega), H^{\frac{1}{2}}(\partial \Omega)}&=&( \mathfrak{I}\{[\omega \BGve(\cdot,\Go)]_*\}\BE_0,\BE_0 )_{\mathcal{H}} , \\ \operatorname{Im}\langle \Lambda_{\omega\BGve_0}V_0,\overline{V_0} \rangle_{H^{-\frac{1}{2}}(\partial \Omega), H^{\frac{1}{2}}(\partial \Omega)}&=&( \mathfrak{I}\{[\omega \BGve_0]_*\}\BE_0,\BE_0 )_{\mathcal{H}}. 
\end{eqnarray*}
Applying Corollary \ref{cor:EffOpIneqsForCoerciveOps} followed by Lemma \ref{Lem.comptnesor-vacuum} together with Corollary \ref{cor:MonotonocityOfEffOps} and then using Corollary \ref{cor:BasicPropertiesEffOpDiricZProb}(ii) yields the series of inequalities:
 \begin{align*}
        \operatorname{Im}\langle \Lambda_{\omega\BGve(\cdot,\Go)}V_0,\overline{V_0} \rangle_{H^{-\frac{1}{2}}(\partial \Omega), H^{\frac{1}{2}}(\partial \Omega)}&=( \mathfrak{I}\{[\omega \BGve(\cdot,\Go)]_*\}\BE_0,\BE_0 )_{\mathcal{H}}\\
        &\geq ( \{\mathfrak{I}[\omega \BGve(\cdot,\Go)]\}_*\BE_0,\BE_0 )_{\mathcal{H}}\\
        &\geq  ( [(\operatorname{Im}\omega) \BGve_0]_*\BE_0,\BE_0 )_{\mathcal{H}}\\
        &=( \mathfrak{I}\{[\omega \BGve_0]_*\}\BE_0,\BE_0 )_{\mathcal{H}}\\
        &=\operatorname{Im}\langle \Lambda_{\omega\BGve_0}V_0,\overline{V_0} \rangle_{H^{-\frac{1}{2}}(\partial \Omega), H^{\frac{1}{2}}(\partial \Omega)},
    \end{align*}
and thus with \eqref{eq.diff}, $ \operatorname{Im}h_{\operatorname{dif},V_0} (\Go)
\geq 0$. Therefore,  $h_{\operatorname{dif},V_0}$ is a Herglotz function.\\[6pt]
{\bf{Step 2}: Proof of assertion 2}\\[4pt]
Let $\omega\in \bbC^+$. Then one has $-\overline{\omega}\in \bbC^+$ and
\begin{equation}\label{eq.difHergsym}
 h_{\operatorname{dif},V_0}(-\overline{\omega})= h_{V_0}(-\overline{\omega})-{h_{\operatorname{vac}, V_0}}(-\overline{\omega}).
 \end{equation}
First, using  the  representations \eqref{eq.factDtNmap} of the DtN map and the fact that $\BE_0=\Pi V_0$, $(-\overline{\omega} \varepsilon_0)_*=-\overline{\omega} \varepsilon_0$,  $(\omega \varepsilon_0)_*=\omega\varepsilon_0$ [by Corollary \ref{cor:BasicPropertiesEffOpDiricZProb}(ii)] on the above expression yields 
\begin{eqnarray}\label{eq.sym-vaccum}
h_{\operatorname{vac}, V_0}(-\overline{\omega})&=&\big((-\overline{\omega} \varepsilon_0)_* \BE_0, \BE_0\big)_{\mathcal{H}}=-\big(\overline{\omega} \varepsilon_0\BE_0, \BE_0\big)_{\mathcal{H}}=-\overline{\big( \omega \varepsilon_0\BE_0, \BE_0\big)}_{\mathcal{H}} \nonumber\\
&=&-\overline{ \big((\omega \, \varepsilon_0)_*\BE_0, \BE_0\big)_{\mathcal{H}}}=-\overline{h_{\operatorname{vac}, V_0}(\omega)}.
\end{eqnarray}
So far, we did not use the assumption H7. It remains to show that $ h_{V_0}(-\overline{\omega})=-\overline{h_{V_0}(\omega)}$ and this will require the reciprocity principle assumption H7 except if $V_0$ is a real-valued function.
 By virtue of the reality principle H2,
 one has 
 \begin{eqnarray*}
 h_{V_0}(-\overline{\omega})&=&\langle \Lambda_{-\overline{\omega}\BGve(\cdot,-\overline{\omega})} V_0, \overline{V_0} \rangle_{H^{-\frac{1}{2}}(\partial \Omega), H^{\frac{1}{2}}(\partial \Omega)}=\langle \Lambda_{-\overline{\omega\,\BGve(\cdot, \omega)}} V_0, \overline{V_0 }\rangle_{H^{-\frac{1}{2}}(\partial \Omega), H^{\frac{1}{2}}(\partial \Omega)}\\ &=&-\langle \Lambda_{\overline{\omega\,\BGve(\cdot, \omega)}} V_0, \overline{V_0} \rangle_{H^{-\frac{1}{2}}(\partial \Omega), H^{\frac{1}{2}}(\partial \Omega)},
 \end{eqnarray*}
where the last inequality comes from the fact that $\Lambda_{-\overline{\omega} \,\overline{\BGve(\cdot,\omega)}}=- \Lambda_{\overline{\omega} \,\overline{\BGve(\cdot,\omega)}}$.
We denote by $w$ (resp.\ $u$) the unique solution  (see Proposition \ref{pro.DtNDirBVPResults}) in $H^1(\Omega)$ of the  Dirichlet problem \eqref{def.DirBVPForCoercivityTensors} problem with $\Ba=\omega \BGve(\cdot, \omega)$ [resp. $\Ba=\overline{\omega\,\BGve(\cdot, \omega)}$]  and $f=V_0$ (resp. $f=V_0$). Using the Green's formula \eqref{eq:GreensFormula}, one has
\begin{equation*}
h_{V_0}(-\overline{\omega})=-\int_{\Omega}\overline{\omega \,\BGve(\cdot, \omega)} \nabla u \cdot \overline{\nabla w}  \,\mathrm{d} \Bx \, .
\end{equation*}
{\bf Case 1 :} If  the reciprocity assumption H7 is satisfied, it implies with \eqref{eq.permitivtyob} that $\BGve(\Bx,\Go)=\BGve(\Bx,\Go)^{\top}$ for a.e.\ $\Bx\in \Omega$. Thus, it leads to
\begin{eqnarray}\label{eq.symmetrieHV0}
 h_{V_0}(-\overline{\omega})&=&-\int_{\Omega}  \overline{\omega \,\BGve(\cdot, \omega)}\, \overline{\nabla w} \cdot \nabla u  \,\mathrm{d} \Bx=-\overline{\int_{\Omega}  \omega \,\BGve(\cdot, \omega)\, \nabla w \cdot \overline{\nabla u}  \,\mathrm{d} \Bx} \\
 &=&- \overline{\langle \Lambda_{\omega\,\BGve(\cdot, \omega)} V_0, \overline{V_0} \rangle}_{H^{-\frac{1}{2}}(\partial \Omega), H^{\frac{1}{2}}(\partial \Omega)}=-\overline{h_{V_0}(\omega)}.\nonumber
\end{eqnarray}
{\bf Case 2 :} If  the reciprocity assumption H7 does not hold but $V_0$ is a real-valued function, then as $V_0=\overline{V_0}$ by uniqueness of $H^1$-solutions of the problem \eqref{def.DirBVPForCoercivityTensors}, one gets that $u=\overline{w}$ so that the first equality in \eqref{eq.symmetrieHV0} still holds without using that $\BGve(\Bx,\Go)=\BGve(\Bx,\Go)^{\top}$ for a.e.\ $\Bx\in \Omega$. From this, we conclude that $h_{V_0}(-\overline{\omega})=-\overline{h_{V_0}(\omega)}$. 

Combining \eqref{eq.difHergsym} and \eqref{eq.sym-vaccum} with the conclusion in these two cases yields the desired equality \eqref{eq.Hergsym}. \\[6pt]
{\bf{Step 3}: Proof of assertion 3}\\[4pt] 
Let $\Go\in [\Go_-,\Go_+] \cup \bbC^+$.  We point out that $0\notin [\Go_-,\Go_+ ]$ by assumption. Thus, using Remark \ref{rem:CoercivityAssumpForVarEpsilon} and Proposition \ref{pro.DtNDirBVPResults}, one can define the DtN map $\Lambda_{\BGve(\cdot, \omega)}$. Furthermore, using Corollary \ref{cor:BasicPropertiesEffOpDiricZProb} both (i) and (ii), one gets that $[\omega \BGve(\cdot, \omega)]_*=\omega \, [\BGve(\cdot, \omega)]_* $ and $(\omega \varepsilon_0)_*=\omega \,\varepsilon_0$. Thus, by virtue of the representation \eqref{eq.factDtNmap} of DtN operator, it yields  to the relations  $\Lambda_{\omega \BGve(\cdot, \Go)}=\omega \Lambda_{\BGve(\cdot, \Go)}$ and $\Lambda_{\omega \varepsilon_0}=\omega \Lambda_{\varepsilon_0}$ which implies that 
\begin{equation}\label{eq.omegahomogen}
\forall \Go\in [\Go_-,\Go_+]\cup \bbC^+, \quad \frac{h_{\operatorname{dif},V_0} (\Go) }{\Go}=\langle [\Lambda_{\BGve(\cdot,\Go)}-\Lambda_{\varepsilon_0}]V_0,\overline{V_0} \rangle_{H^{-\frac{1}{2}}(\partial \Omega), H^{\frac{1}{2}}(\partial \Omega)}=F_{V_0}(\Go).
\end{equation}
{\bf{Step 4}: Proof of assertion 4}\\[4pt] 
Let $\Go=\mathrm{i}\, y$ with $y>0$. First, one has the identity \eqref{eq.omegahomogen} for  $\Go=\mathrm{i}\, y$. Second, as \eqref{eq.passive-cloak} and \eqref{eq.permitivtyob} are assumed to hold on $\bbC^+$ then together with assumption \textnormal{H3} and the definition of $\BGve_{\infty}(\cdot)$ in \eqref{eq.defGVeinf} we conclude that the following limit holds along the positive imaginary axis:
\begin{equation}\label{eq.pointbehav}
 \mbox{ for  a.e. } \Bx \in \Omega,  \  \BGve(\Bx, \mathrm{i}y)\to \BGve_{\infty}(\Bx) \ \text{ as }y\rightarrow \infty.
 \end{equation}
 Now, one can use a similar argument as in step 4 of Proposition \ref{Pro:DtNMapQuadFormIsHerglotz} except with $\varepsilon(\cdot,\omega)$ instead of $\omega \varepsilon(\cdot,\omega)$. More precisely, replacing assumption \textnormal{H4} with \eqref{eq.pointbehav}, the limit $\Go_n\to \Go$ by $\Go_n=\mathrm{i} y_n\to \infty$ and using in the neighborhood of $\infty$ assumption \textnormal{H5}, and  the coercivity of $\varepsilon(\cdot,\omega)$ with coercivity constant $\Gve_0$ [which is easily derived  from the representation \eqref{eq.Herglotzeps} on the positive imaginary axis], one shows that
$$
F_{V_0}(\mathrm{i} y)=\frac{h_{\operatorname{dif},V_0} (\mathrm{i} y)}{\mathrm{i} y} \to \langle[ \Lambda_{\BGve_{\infty}(\cdot)}-\Lambda_{\varepsilon_0}]V_0,\overline{V_0} \rangle_{H^{-\frac{1}{2}}(\partial \Omega), H^{\frac{1}{2}}(\partial \Omega)}=:F_{V_0,\infty}, \mbox{ as } \ y\to +\infty.
$$
It remains to prove that the above limit is non-negative. 
Introducing with $\BE_0=\Pi V_0$ defined  and using the representation of DtN operator  \eqref{eq.factDtNmap} yields:
$$
 \langle[ \Lambda_{\BGve_{\infty}(\cdot)}-\Lambda_{\varepsilon_0}]V_0,\overline{V_0} \rangle_{H^{-\frac{1}{2}}(\partial \Omega), H^{\frac{1}{2}}(\partial \Omega)}=\big([\BGve_{\infty}(\cdot)]_* \BE_0, \BE_0 \big)_{\mathcal{H}}-  \big((\varepsilon_0)_* \BE_0, \BE_0 \big)_{\mathcal{H}} \geq 0.
$$
The last inequality is a consequence of the monotonicity result of Corollary \ref{cor:MonotonocityOfEffOps} since by \eqref{eq.permitivtyob}, for a.e.\ $\Bx \in \Omega$,  $\BGve_{\infty}(\Bx)=\BGve_{\infty}^\dagger(\Bx)$ and $\BGve_{\infty}(\Bx) \geq \varepsilon_0\BI$ and thus in the sense of left multiplicative operators, one has $\BGve_{\infty}(\cdot)\geq \varepsilon_0\BI$.
\end{proof}

\begin{Cor}\label{Cor.boundleadingcoefficient}
If $V_0\in H^{1/2}(\partial\Omega)$ is affine, i.e., $V_0=-{\bf e}_0\cdot \Bx|_{\partial\Omega}$ for a constant vector ${\bf e}_0\in \bbC^d$, then
\begin{equation}\label{eq.boundinf}
F_{V_0, \infty}=\langle[ \Lambda_{\BGve_{\infty}(\cdot)}-\Lambda_{\varepsilon_0}]V_0,\overline{V_0} \rangle_{H^{-\frac{1}{2}}(\partial \Omega), H^{\frac{1}{2}}(\partial \Omega)}\geq  \Big( \big[\big\langle \BGve^{-1}_{\infty}(\cdot)\big\rangle^{-1}- \varepsilon_0\big] {\bf e}_0,{\bf e}_0 \Big)_{\mathcal{H}},
\end{equation}
where
\begin{eqnarray}
 \Big( \big[\big\langle \BGve^{-1}_{\infty}(\cdot)\big\rangle^{-1}- \varepsilon_0\big] {\bf e}_0,{\bf e}_0 \Big)_{\mathcal{H}} &=&    \Big(|\Omega| \Big[  \int_{\Omega} \BGve_{\infty}(\Bx)^{-1} \mathrm{d} \Bx \Big]^{-1}- \varepsilon_0\Big)  |\Omega| \, \, \|{\bf e}_0\|_{\bbC^d}^2 \label{eq.eqboundtensoreq}\\
 &\geq &  |\mathcal{O}| \  \Big(1- \frac{\varepsilon_0} {\varepsilon} \Big) \ \    \frac{  \varepsilon_0 \, |\Omega| \, \, \|{\bf e}_0\|_{\bbC^d}^2}{|\mathcal{O}|\,  \frac{\varepsilon_0}{\varepsilon}+ |\Omega \setminus \mathcal{O}|  }. \label{eq.eqboundtensorineq}
\end{eqnarray}
\end{Cor}
\begin{proof}
The  representation
\eqref{eq.factDtNmap}  of the DtN operator  $\Lambda_{\BGve_0}= \Pi^\dagger ( \varepsilon_0)_* \Pi=\Pi^\dagger \varepsilon_0 \Pi$ (where $( \varepsilon_0)_* =\varepsilon_0$ by Corollary \ref{cor:BasicPropertiesEffOpDiricZProb}), and the relation 
\eqref{LiftOpOnAffBCs} imply that
$$
\langle[\Lambda_{\BGve_{\infty}(\cdot)}-\Lambda_{\varepsilon_0}]V_0, \overline{V_0} \rangle_{H^{-\frac{1}{2}}(\partial \Omega), H^{\frac{1}{2}}(\partial \Omega)}= \langle \Lambda_{\BGve_{\infty}(\cdot)} (-{\bf e}_0\cdot \Bx), \overline{(-{\bf e}_0\cdot \Bx)}\rangle_{H^{-\frac{1}{2}}(\partial \Omega), H^{\frac{1}{2}}(\partial \Omega)} -( \varepsilon_0 \, {\bf e}_0, {\bf e}_0)_{\mathcal{H}}.
$$
We apply now the lower bound of inequality \eqref{eq.boundefftnesor}   of Theorem \ref{Thm:SigmaDBndsOfMilton} [since $\Ba=\BGve_{\infty}(\cdot)$ is a coercive self-adjoint left multiplicative operator] on the above expression which yields the inequality \eqref{eq.boundinf}.
The equality of \eqref{eq.eqboundtensoreq} is just derived from the definition \eqref{eq.average op} which implies that for constant fields ${\bf e}_0\in \langle\mathcal{U}\rangle$:
\begin{equation}\label{eq.equalityaverage}
\big\langle \BGve^{-1}_{\infty}(\cdot)\big\rangle^{-1} {\bf e}_0=|\Omega| \Big[  \int_{\Omega} \BGve_{\infty}(\Bx)^{-1} \mathrm{d} \Bx \Big]^{-1}{\bf e}_0.
\end{equation}
Then, using the definition \eqref{eq.defGVeinf} together with \eqref{eq.permitivtyob} yields that $$\BGve_{\infty}(\Bx)^{-1} \leq \varepsilon^{-1} \, \BI \ \mbox{ for a.e. }\Bx\in \mathcal{O} \mbox{ and } \BGve_{\infty}(\Bx)^{-1} =\varepsilon_0^{-1} \BI \ \mbox{ for a.e. }\Bx\in \Omega \setminus  \mathcal{O}  $$
implying
\begin{equation}\label{eq.inequalityaverage}
  \int_{\Omega} \BGve_{\infty}(\Bx)^{-1} \mathrm{d} \Bx \leq |\mathcal{O}| \varepsilon^{-1}\BI+   |\Omega \setminus \mathcal{O}|  \,\varepsilon_0^{-1}\BI.
\end{equation}
Finally, combining the equality \eqref{eq.eqboundtensoreq}, \eqref{eq.equalityaverage} and \eqref{eq.inequalityaverage}
and the fact that if  $\mathbb{A}$ and $\mathbb{B}$ are positive invertible operators with $\mathbb{A}\leq \mathbb{B}$ then $\mathbb{B}^{-1}\leq \mathbb{A}^{-1}$ gives
\begin{eqnarray}\label{eq.inequalityaverage2}
\Big( \big[\big\langle \BGve^{-1}_{\infty}(\cdot)\big\rangle^{-1}- \varepsilon_0\big] {\bf e}_0,{\bf e}_0 \Big)_{\mathcal{H}}&\geq& \Big(\frac{|\Omega|}{|\mathcal{O}| \varepsilon^{-1}+   |\Omega \setminus \mathcal{O}|  \,\varepsilon_0^{-1}}- \varepsilon_0\Big) |\Omega| \, \|{\bf e}_0\|_{\bbC^d}^2. 
\end{eqnarray}
Finally, using that $|\mathcal{O}|+ |\Omega\setminus\mathcal{O}|=|\Omega|$, one can rewrite the right hand side of \eqref{eq.inequalityaverage2} as the right hand side of \eqref{eq.eqboundtensorineq} which concludes the proof.
 
\end{proof}

\subsection{A Connection between Stieltjes and Herglotz function}\label{sec-Herg-Stielt}
In this section, we recall a strategy followed in \cite{Cassier:2017:BHF} to construct a Herglotz function from a Stieltjes function that will inherit more properties than a ``standard" Herglotz function. Such properties will be important to develop optimal bounds as in \cite{Cassier:2017:BHF}.

We  introduce  the complex square root by
\begin{equation}\label{eq.racine}
\sqrt{z}=|z|^{\frac{1}{2}}\,e^{\ii\arg z/2} \ \mbox{ if } \ \arg{z} \in (0,2\pi)
\end{equation}
and extend it on the branch cut $\bbR^{+}$ by its limit from the upper-half plane, in other words, the square root of non-negative real number $x$ is given by $\sqrt{x}=|x|^{\frac{1}{2}}$.
 The following proposition gives sufficient conditions to construct a Stieltjes function. It is inspired by  Lemma 8 and Theorem 9  of \cite{Cassier:2017:BHF}. We recall its proof for the reader; see also Corollary 10 in \cite{Cassier:2017:BHF}.

\begin{Pro}\label{cor.HergStielt}
If $S$ a Stieltjes function such that $S(x)\to S_{\infty}$ as $x\to+\infty$ on $\bbR^{+,*}$, then $H$ defined by
\begin{equation}\label{eq.defv}
H(z):=z\, S(-z), \quad  \forall z \in \mathbb{C}\setminus \mathbb{R}^+
\end{equation}
is a Herglotz function which is analytic on $\mathbb{C}\setminus \mathbb{R}^+$ and non-positive (and even negative if $S_{\infty}>0$) on $\mathbb{R}^{-,*}:=(-\infty,0)$. Moreover, $H$ satisfies the following Schwarz reflection principle:
\begin{equation}\label{eq.SymHerg}
\overline{H(\overline{z})}= H(z), \quad \forall z \in \bbC \setminus \bbR^{+},
\end{equation}
and in the associated representation \eqref{eq.defhergl} of $H$ given by Theorem \ref{thm.Herg}, the positive regular Borel measure $\mathrm{m}$ has support included in $\mathbb{R}^+$ and  the coefficient $\alpha=S_{\infty}$.
\end{Pro}
\begin{proof}
We give the main ideas here and refer to  the proof of Corollary 10 of \cite{Cassier:2017:BHF} for more details.
The fact that $H$ is well-defined, analytic  on $\mathbb{C}\setminus \mathbb{R}^+$ and non-positive on $\mathbb{R}^{-,*}$ is an immediate consequence of the property  of the Stieltjes  function $S$ via the Definition \ref{Def.Stielt}. The fact that $H(z)=z\, S(-z)$  is negative on $\mathbb{R}^{-,*}$ if $S_{\infty}>0$ is an immediate consequence of the representation \eqref{eq.representationStieltjes}  which implies that $S$ is positive on $\mathbb{R}^{+,*}$ if its coefficient $\alpha_{S}=S_{\infty}>0$. Moreover, using the Stieltjes  representation \eqref{eq.representationStieltjes} of $S$, one  checks easily that $\operatorname{Im}H(z)\geq 0$ for $z\in \mathbb{C}^+$. Thus, $H$ is Herglotz function.

The Stieltjes function $S$ satisfies the Schwarz reflection principle \eqref{eq.StieltSchwprinciple}, namely $\overline{S(\overline{z})}=S(z)$ for all $z\in \mathbb{C}\setminus \bbR^-$. Thus, one has 
\begin{equation*}
\forall z\in \bbC \setminus \bbR^+, \quad \overline{H(\overline{z})}=z \, \overline{S(\overline{-z})}=z\, S(-z)=H(z).
\end{equation*}
From the representation \eqref{eq.representationStieltjes} of $S$, one proves also with the dominated convergence Theorem that   $\alpha_{S}=S_{\infty}=\lim_{y\to \infty} S(-\mathrm{i} \, y) $ and thus   \eqref{eq.coeffdomherg} and \eqref{eq.defv} gives $\alpha=S_{\infty}$. Finally, using \eqref{eq.coeffdomherg} and  \eqref{eq.mesHerg} and the fact  $H$ is  analytic  on $\mathbb{C}\setminus \mathbb{R}^+$ and real-valued on $\mathbb{R}^{-,*}$ implies (via the dominated convergence Theorem)  that $\mathrm{m}([a,b])=0$ for any $a<b<0$ and thus $\mathrm{m}$ has support included in $\mathbb{R}^+$.
\end{proof}

\subsection{Construction of a Stieltjes function related to the DtN map}\label{sec-Herg-phys}
To derive analytical bounds on the DtN map as a function of the frequency for the passive cloaking device, we  follow the approach developed in \cite{Cassier:2017:BHF}. 
In that perspective, one needs first to construct a Stieltjes function $S_{V_0}$
using the function $F_{V_0}$  which measures the quality of the cloaking effect of our device and then we will obtain a Herglotz function $H_{V_0}$ with useful properties for deriving bounds in the next section. 
\begin{Lem}\label{Const.Stielt}
Suppose that the cloaking device satisfies assumptions \textnormal{H1--H6} and $\mathrm{H}7$. Then, for any $V_0\in H^{1/2}(\partial \Omega)$,  the function $S_{V_0}$ defined $\forall z \in  (\bbC \setminus \bbR^-)\cup [-\omega_+^2, -\omega_-^2] $ by
\begin{equation}\label{eq.defSvo}
S_{V_0}(z):=F_{V_0}(\sqrt{-z})=\langle[\Lambda_{\BGve(\cdot,\sqrt{-z} )} -\Lambda_{\BGve_0}]V_0,\overline{V_0} \rangle_{H^{-\frac{1}{2}}(\partial \Omega), H^{\frac{1}{2}}(\partial \Omega)}
\end{equation}
is a Stieltjes function which satisfies $\displaystyle \lim_{x\to +\infty} S_{V_0}(x)=F_{V_0,\infty}\geq 0$ [with $F_{V_0,\infty}$ defined in \eqref{eq.limitinfherg}]. Moreover, if $ F_{V_0,\infty}>0$ then $S_{V_0}$ is positive on  $\bbR^{+,*}$. Furthermore, $S_{V_0}$ is continuous on $ (\bbC \setminus \bbR^-) \cup [-\omega_+^2, -\omega_-^2] $.
Finally, if the reciprocity assumption $\mathrm{H}7$ is not satisfied, all these properties remain valid if $V_0$ is a real-valued function in $H^{1/2}(\partial \Omega)$.
\end{Lem}
\begin{proof}
This proof is based on the relation 
\eqref{eq.homogDtnrelation} between  the function $F_{V_0}$ 
and the Herglotz function $h_{\operatorname{dif}, V_0}$ introduced in Theorem \ref{thm.differenceHerg}.\\[4pt]
Let $\Go\in \bbC^+\cup [\Go_-, \Go_+] $ then by the identity \eqref{eq.homogDtnrelation}, one has 
$F_{V_0}(\Go)=h_{\operatorname{dif}, V_0}(\Go)/\Go$. Thus,  by the assertion 1 of Theorem \eqref{thm.differenceHerg}, $F_{V_0}$ is analytic on $\bbC^+$ and continuous on $[\Go_-,\Go_+]$.\\[4pt]
\noindent To prove that $S_{V_0}$ is a Stieltjes function, we prove an integral representation for $F_{V_0}$.
This part of the proof is similar to the one done in Theorem 4.5 of \cite{Cassier:2017:MMD} to obtained the representation of the  permittivity as a function of the frequency in a passive electromagnetic linear system. Using the fact that $h_{\operatorname{dif}, V_0}(\Go)$ is Herglotz function satisfying the ``symmetry relation"  \eqref{eq.Hergsym} and  the limit behavior \eqref{eq.limitinfherg}, we get by Theorem \ref{thm.Herg} and  Corollary \ref{cor.Herg} that it admits the following representation:
\begin{equation}\label{eq.hergdtn}
h_{\operatorname{dif}, V_0}(\Go)=F_{V_0, \infty}\, \Go+ \displaystyle \int_{\bbR} \left(  \frac{1}{\xi-\Go}- \frac{\xi}{1+\xi^2}\right)\md \mm_{V_0}( \xi), \quad \forall \Go \in \bbC^+,
\end{equation}
where we use that the ``symmetry relation" $h_{\operatorname{dif},V_0}(-\overline{\omega})=-\overline{h_{\operatorname{dif}, V_0}(\omega)}$ implies that $h_{\operatorname{dif}, V_0}$ is purely imaginary on $\mathrm{i} \bbR^{+,*}$ which leads to $\operatorname{Re}h_{\operatorname{dif}, V_0}(\mathrm{i})=0$. Moreover, we point out that this ``symmetry relation"   implies that $\operatorname{Im}h_{\operatorname{dif}, V_0}(-\overline{\omega})=\operatorname{Im}h_{\operatorname{dif}, V_0}(\Go)$. Thus, with \eqref{eq.coeffdomherg} and \eqref{eq.mesHerg}, one shows that  the positive regular Borel measure associated to $h_{\operatorname{dif}, V_0}$ by Theorem \ref{thm.Herg} is ``even" in the sense that $\mathrm{m}_{V_0}(B)=\mathrm{m}_{V_0}( -B)$ for any Borel set $B$. 
Hence \eqref{eq.hergdtn} gives that  
\begin{equation}\label{eq.stieldtn1}
F_{V_0}(\Go)=\frac{h_{\operatorname{dif}, V_0}(\Go)}{\Go}=F_{V_0, \infty}+\frac{1}{\Go}  \int_{\bbR} \left(  \frac{1}{\xi-\Go}- \frac{\xi}{1+\xi^2}\right)\md \mm_{V_0}( \xi)  \quad \forall \Go \in \bbC^+.
\end{equation}
As the ``symmetry relation" \eqref{eq.Hergsym} implies that $F_{V_0}(\Go)=\overline{F_{V_0}(-\overline{\Go})}$ on $\bbC^+$, one can use that $F_{V_0}(\Go)=1/2[F_{V_0}(\Go)+\overline{F_{V_0}(-\overline{\Go})}]$ alongside with formula \eqref{eq.stieldtn1} at $\omega$ and $-\overline{\omega}\in \bbC^+$ to get
\begin{equation}\label{eq.stieldtn2}
F_{V_0}(\Go)=F_{V_0, \infty}- \int_{\bbR}  \frac{\md \mm_{V_0}( \xi)}{\Go^2-\xi^2},  \quad \forall \Go \in \bbC^+ .
\end{equation}
Thus, one gets from the above formula that  for all $\Go\in \bbC^+ \cap \{ u \in \bbC\mid \pm \operatorname{Re}(u)>0\}$:
\begin{equation}\label{eq.Stieltprop}
 \pm \operatorname{Im}F_{V_0}(\Go)\geq 0 \
\mbox{ and } \ \forall y>0, \ F_{V_0}(\mathrm{i} y)=  F_{V_0, \infty}+\int_{\xi \in \mathbb{R}}\frac{\md \mm_{V_0}( \xi) }{y^2+\xi^2}\geq F_{V_0,\infty}\geq 0
\end{equation}
since  $\mm_{V_0}$ is a positive measure, and $F_{V_0,\infty}\geq 0$ by \eqref{eq.defGVeinf}.\\[4pt]
\noindent By virtue of the definition of the complex square root \eqref{eq.racine} with a branch cut on $\mathbb{R}^+$, one gets that $z \mapsto \omega=\sqrt{-z}$ is analytic on $\bbC\setminus (-\infty,0]$ and maps:  $\bbC\setminus (-\infty,0]$  to $\bbC^+$ and $\bbC^{\pm}$ to $\bbC^{+} \cap \{ \mp \operatorname{Re}(z)\geq 0\}$. Thus, one deduces by composition that $S_{V_0}$ is well-defined and analytic on $\bbC\setminus (-\infty,0]$. Furthermore,  by \eqref{eq.Stieltprop},  $\operatorname{Im}F_{V_0}(\omega)\leq 0$ for $\Go\in \bbC^{+} \cap \{  \operatorname{Re}(u)\leq 0\}$, thus  by composition, $\operatorname{Im}S_{V_0}(z)\leq 0$ for $z\in \bbC^+$. In addition, if $y>0$, then by \eqref{eq.Stieltprop},  $S_{V_0}(y)=F_{V_0}(\mathrm{i} \, y^{1/2})\geq 0$.
Thus, by Definition \ref{Def.Stielt},  $S_{V_0}$ is Stieltjes function. Moreover, if $ F_{V_0,\infty}>0$ then (by \eqref{eq.Stieltprop}) $S_{V_0}>0$ on  $\bbR^{+,*}$ . \\[4pt]
Furthermore,  by  \eqref{eq.limitinfherg}, one has on $\bbR^{+,*}$: $S_{V_0}(y)=F_{V_0}(\mathrm{i} \, y^{1/2})\to F_{V_0,\infty}\geq 0$ when $y\to +\infty$. Then, using the definition of the complex square root  $z \mapsto \omega=\sqrt{-z}$  maps continuously $(\bbC\setminus \bbR^-) \cup [-\omega_+^{2}, -\omega_-^2]$ to  $\bbC^+ \cup [\Go_-, \Go_+]$, one gets using assertion 1 and 3 of Theorem \ref{thm.differenceHerg} that by composition $S_{V_0}$ is continuous on  $ (\bbC \setminus \bbR^-) \cup [-\omega_+^2, -\omega_-^2] $.\\[4pt]
Finally, using Theorem \ref{thm.differenceHerg},  if H7 does not hold, our proof is unchanged if $V_0\in H^{1/2}(\partial \Omega)$ is a real-valued function since the ``symmetry relation" \eqref{eq.Hergsym} holds in this case as well.
\end{proof}

Combining the preceding Lemma \ref{Const.Stielt} and Corollary \ref{cor.HergStielt}  yields directly the following theorem. This result establishes the existence of a Herglotz function related to DtN map. This function incorporates on one hand all the properties of our physical passive system. On the other hand, it is  directly related to the function $F_{V_0}$, introduced in \eqref{eq.defFV}, whose modulus measures quantitatively the quality of the cloaking.
Moreover, the properties of these functions allows us to apply the bounds developed in \cite{Cassier:2017:BHF} in the general context of  passive linear systems.
\begin{Thm}\label{Thm.Hergphys}
Suppose that the cloaking device satisfies assumptions \textnormal{H1--H6} and 
$\mathrm{H}7$. Then, for any $V_0\in H^{1/2}(\partial \Omega)$, the function $H_{V_0}$ defined for all $z \in (\mathbb{C}\setminus \mathbb{R}^+ )\cup [\omega_-^2, \omega_+^2]$ by
\begin{gather}\label{eq.defHVo}
H_{V_0}(z)= z S_{V_0}(-z) = z F_{V_0}(\sqrt{z})= z\langle[\Lambda_{\BGve(\cdot,\sqrt{z} )} -\Lambda_{\BGve_0}]V_0,\overline{V_0} \rangle_{H^{-\frac{1}{2}}(\partial \Omega), H^{\frac{1}{2}}(\partial \Omega)} 
\end{gather}
is a Herglotz function which is analytic on $\mathbb{C}\setminus \mathbb{R}^+$, continuous on $[\omega_-^2, \omega_+^2]$, and non-positive (and even negative if the non-negative coefficient $F_{V_0,\infty}$ defined in \eqref{eq.limitinfherg} is positive) on $\mathbb{R}^{-,*}$. Moreover,  
$H_{V_0}$ satisfies the following Schwarz reflection principle:
\begin{equation}\label{eq.SymHerg2}
\overline{H_{V_0}(\overline{z})}= H_{V_0}(z), \quad \forall z \in \bbC \setminus \bbR^{+},
\end{equation}
and in the associated representation \eqref{eq.defhergl} of  $H_{V_0}$ given by Theorem \ref{thm.Herg}, the positive regular Borel measure $\mathrm{m}$ has support included in $\mathbb{R}^+$ and the coefficient $\alpha=F_{V_0,\infty}$. 
Furthermore, if the reciprocity assumption $\mathrm{H}7$ is not satisfied, all these properties remain valid if $V_0$ is a real-valued function in $H^{1/2}(\partial \Omega)$.
\end{Thm}

\section{Quantitative bounds on passive cloaking over a frequency interval}\label{sec:QuantiativeBndsPassiveCloakingFreqInterval}

\subsection{Perfect cloaking cannot occur}\label{sec:PerfectCloakingCantOccur}
As in \cite{Cassier:2017:BHF} for the far-field cloaking problem,  we show in this section that perfect cloaking, in the sense of Definition \ref{Def-perfectcloaking}, cannot occur on any positive finite frequency interval. 
\begin{Thm}\label{Pro.perfectcloaking}
Suppose that the cloaking device satisfies assumptions \textnormal{H1--H6} then the obstacle
 is not perfectly cloaked on the frequency interval.
\end{Thm}
\begin{proof}
Assume by contradiction that perfect cloaking occurs and let $V_0$ be any affine boundary condition in $H^{1/2}(\partial \Omega)$ defined by $V_0(\Bx)=- {\bf e}_0\cdot \Bx$ on $\partial \Omega$ with a constant vector ${\bf e}_0 \in \bbC^d\setminus \{0\}$. We point out if H7 not satisfied (i.e., if the cloak contains a non-reciprocal material), we further assume  that $V_0$ is real-valued and hence ${\bf e}_0\in \bbR^d$. Then perfect cloaking implies that $F_{V_0}(\Go)=0$ for all $\Go$ on the bandwith $[\Go_-, \Go_+]$. Thus, by Theorem \ref{Thm.Hergphys}, $H_{V_0}(z)=0$ for all $z\in [\Go_-^2, \Go_+^2]$ and using the Schwarz reflection  principle [see \eqref{eq.SymHerg2}], $H_{V_0}$ defined by \eqref{eq.defHVo} is analytic on $(\mathbb{C}\setminus \mathbb{R}^+ )\cup (\omega_-^2, \omega_+^2)$ and vanishes on $(\omega_-^2, \omega_+^2)$ implying it is identically equal to zero on $(\mathbb{C}\setminus \mathbb{R}^+ )\cup (\omega_-^2, \omega_+^2)$. Thus, as $H_{V_0}$ is the zero Herglotz function, its leading coefficient $\alpha$ is zero. 
However, the leading coefficient of $H_{V_0}$ is  $F_{V_0,\infty}$ (see Theorem \ref{Thm.Hergphys}) which is  positive by Corollary \ref{Cor.boundleadingcoefficient}. This is a contradiction and hence perfect cloaking does not occur.
\end{proof}
As mentioned, this result prevents perfect cloaking, but it does not prevent approximate cloaking since it is based 
on an analytic continuation argument and analytic continuation is (exponentially) unstable. 
Roughly speaking, it means that if one takes two disjoint open sets  in the domain of analyticity of an analytic function, this function could be arbitrary small in one of these open sets  and enormous in the other \cite{Vasquez:2009:AEC, Grabovsky:2021:OEE}. Therefore, one really needs quantitative bounds to impose fundamental limits on  passive cloaking over the positive frequency interval $[\Go_-, \Go_+]$. This is the aim of the following subsections. 
\subsection{The sum rules approach and general bounds}\label{sec:SumRulesApprGenBounds}
In this subsection, we recall  general bounds on Herglotz functions derived in \cite{Cassier:2017:BHF} for passive linear systems based on the following sum rules theorem proved in \cite{Bernland:2011:SRC}. 
\begin{Thm}\label{Thm.sumrules}
Assume $h:\mathbb{C}^+\rightarrow \mathbb{C}$ be a Herglotz function which admits the following asymptotic expansions along the positive imaginary axis (i.e., $z\in \mathrm{i}\mathbb{R}^{+,*}$):
\begin{eqnarray*}
&&h(z)=a_{-1} \,z^{-1}+o(z^{-1}) \ \mbox{ as } |z| \to  0, \\[10pt] 
\mbox{ and } && h(z)=  b_{-1}\, z^{-1}+o(z^{-1})\ \mbox{ as }  |z|  \to +\infty,
\end{eqnarray*}
with $a_{-1}$ and $b_{-1} \in \bbR$. Then the following identity holds 
\begin{equation}\label{eq.sumrules}
\lim_{\eta \to 0^{+}}\lim_{y\to 0^{+}}\frac{1}{\pi} \int_{\eta<|x|<\eta^{-1}} \Imag h (x+\ii y) \, \md x=a_{-1}-b_{-1}.
\end{equation}
\end{Thm}

We observe that the above sum-rule theorem cannot be applied directly to the Herglotz function $H_{V_0}$ (defined in Theorem~\ref{Thm.Hergphys}) associated with our physical system. Indeed, along the positive imaginary axis $\mathrm{i}\mathbb{R}^{+,*}$, this function does not decay to $0$ but instead blows up linearly at high frequencies. Consequently, we will apply the theorem to other Herglotz functions obtained by composing $H_{V_0}$ [see~\eqref{eq.defHmuVO}] with  suitable  Herglotz functions $h_{\mu}$ associated to Borel probability measures $\mu$ [see \eqref{eq.Hergprob}].

The bounds derived in \cite{Cassier:2017:BHF} generalize the ones developed in \cite{Gustafsson:2010:SRP,Bernland:2011:SRC} since they apply to a class of Herglotz functions associated to Borel probability measures [see \eqref{eq.Hergprob}], whereas in  \cite{Gustafsson:2010:SRP,Bernland:2011:SRC}, the authors  consider only the case of an Herglotz function associated with a uniform probability measure on a fintie interval to derive their bounds. Moreover, they are optimal in the sense of Theorem \ref{thm.mesopt} as they maximize the left-hand side of the inequality \eqref{eq.sumruleineq} derived from the sum rules \eqref{eq.sumrules}  over the finite interval $[x_-,x_+]$. Furthermore, we will see in Section \ref{sec-losslesscase} that if the cloak is lossless then the bounds are sharp in the sense that one can exhibit an Herglotz function $H_{V_0}$ for which the inequality becomes an equality. 

To derive these bounds, we being by introducing the Herglotz function $h_{\mu}$ defined by:
\begin{equation}\label{eq.Hergprob}
h_{\mu}(z)=\int_{\mathbb{R}} \frac{\md  \mu( \xi) }{\xi-z} , \ \forall z\in \bbC^{+},
\end{equation}
for $\mu\in\CM$\footnote{In \cite{Cassier:2017:BHF}, the authors consider probability Borel measures $\mu$ supported in a compact interval $[-\Delta, \Delta]$. However, all the proofs would still follow by replacing $[-\Delta, \Delta]$ by $\bbR$ even if one does not make any assumption on the support of the measures $\mu$. For this reason, we present here these results in this slightly more general setting.}. Here $\CM$ stands for the set of finite regular positive Borel measure $\mu$ whose total mass is normalized to $1$. In other words: $ \mu(\bbR)=1$, for all $\mu\in\CM$ and thus, $\CM$ is the set of Borel probability measures on $\mathbb{R}$.

Next, let $\mu \in \CM$ and $V_0\in H^{1/2}(\partial \Omega)$,
with $V_0$ being taken to be real-valued if the reciprocity assumption H7 does not hold for the cloak.
Let $H_{V_0}$ be the Herglotz function, given in Theorem \ref{Thm.Hergphys} under the assumptions that the cloaking device satisfies \textnormal{H1--H6}. In addition, we assume that $V_0$ is chosen such that the non-negative coefficient $F_{V_0, \infty}$ [defined in \eqref{eq.limitinfherg}] is positive [for instance, we know from Corollary  \ref{Cor.boundleadingcoefficient} that these conditions are satisfied if $V_0$ is affine and nonzero, i.e., $V_0=-\mathbf{ e}_0 \cdot \Bx$ on $\partial \Omega$ for a constant vector $\mathbf{ e}_0\in \bbC^d\setminus\{0\}$ (or resp. $\mathbf{e}_0\in \bbR^d\setminus\{0\}$ when H7 does not hold)]. 
In particular, as $F_{V_0, \infty}>0$ and since 
it equals the coefficient $\alpha$ in the representation \eqref{eq.defhergl} of $H_{V_0}$, it means that $H_{V_0}$ is a non-constant Herglotz function and hence $H_{V_0}(\bbC^+)\subset \bbC^+$ by the open mapping theorem for analytic functions.

As $H_{V_0}$ is not constant, one can define $H_{\mu, V_0}$ as the following composition of two Herglotz functions
\begin{equation}\label{eq.defHmuVO}
H_{\mu, V_0}=h_{\mu} \circ H_{V_0} \ \mbox{ on } \bbC^+,
\end{equation}
and thus $H_{\mu, V_0}$ is itself a Herglotz function.
The approach, developed in \cite{Cassier:2017:BHF}, consists of deriving bounds on the Herglotz function
$
H_{\mu, V_0}
$
for any $\mu \in \CM$ using sum rules theorem, i.e., Theorem \ref{Thm.sumrules}. These bounds contain the physical properties and constrains of the passive linear system (here the cloak) which are encoded in $H_{V_0}$. 

We will apply Theorem \ref{Thm.sumrules} to the Herglotz function  $H_{\mu, V_0}$. To do so, we will need the asymptotic behavior of this function on $\mathrm{i}\mathbb{R}^{+,*}$ in both the low and high frequency regimes.
The low-frequency behavior on  $\mathrm{i}\mathbb{R}^{+,*}$ follows from  \eqref{eq.asymherglotzfunction} which states that
\begin{equation}\label{eq.LF}
H_{\mu, V_0}(z)=-\mu_{V_0}(\{0\}) z^{-1}+o(z^{-1} ) \ \mbox{ when } \ |z|\to 0 \mbox{ on } \mathrm{i}\bbR^{+,*},
\end{equation}
where $\mu_{V_0}$ is the positive regular Borel measure associated to the Herglotz function $H_{\mu, V_0}$ from its representation in Theorem \ref{thm.Herg} [thus satisfying $\mu_{V_0}(\{0\})\geq 0$]. 
Its high-frequency behavior on  $\mathrm{i}\mathbb{R}^{+,*}$ is given by 
\begin{equation}\label{eq.HF}
H_{\mu, V_0}(z)=-\frac{1}{F_{V_0,\infty}} z^{-1}+o(z^{-1}) \ \mbox{ when } \ |z|\to \infty \mbox{ on } \mathrm{i}\bbR^{+,*}.
\end{equation}
For the proof of this latter asymptotic expansion, see \cite[Lemma 13]{Cassier:2017:BHF}. 

It now follows from the asymptotics \eqref{eq.LF} and \eqref{eq.HF} together with Theorem  \ref{Thm.sumrules} and Corollary \ref{cor.Herg} that, for any $[x_-,x_+]\subset \bbR$ (with $x_-<x_+$), we have the inequalities
\begin{align*}
\hspace{-3em}\lim_{y\to 0^{+}}\frac{1}{\pi} \int_{x_-}^{x_+} \operatorname{Im} H_{\mu, V_0} (x+\ii y) \, \md x  &\leq \lim_{\eta \to 0^{+}}\lim_{y\to 0^{+}}\frac{1}{\pi} \int_{\eta<|x|<\eta^{-1}} \operatorname{Im} H_{\mu, V_0} (x+\ii y) \, \md x\\
&\leq -\mu_{V_0}(\{0\})+  \frac{1}{F_{V_0,\infty}}\leq \frac{1}{F_{V_0,\infty}}
\end{align*}
and, in particular, we get the useful bounds:
\begin{gather}\label{eq.sumruleineq}
    \lim_{y\to 0^{+}}\frac{1}{\pi} \int_{x_-}^{x_+} \operatorname{Im} H_{\mu, V_0} (x+\ii y) \, \md x\leq \frac{1}{F_{V_0,\infty}}, \ \forall \mu\in \mathcal{M}.
\end{gather}
We point out that the physical properties of the cloaking device (the volume of the obstacle $|\mathcal{O}|$, the volume of the cloak $|\Omega\setminus \mathcal{O}|$, the relative permittivity  of the obstacle with respect to the vacuum, etc.) are encoded in the general bounds \eqref{eq.sumruleineq} via the coefficient  $F_{V_0,\infty}$.  

 The goal of the next subsections \ref{sec-losslesscase} and \ref{sec-lossycase} will be to derive from the general bounds  \eqref{eq.sumruleineq} more explicit inequalities in terms of the physical parameters of the cloaking device by choosing the measure $\mu\in \CM$ of the Herglotz function $h_{\mu}$.
 This will allow to get rid of the limit in \eqref{eq.sumruleineq}.
 In Section  \ref{sec-losslesscase}, Dirac measures $\delta_{\xi}$ play a key role among probability Borel measures in obtaining sharp bounds when the cloak is assumed to be lossless in the frequency interval $[\Go_-, \Go_+]$. To motivate the use of Dirac measures, we conclude this section with the following theorem which states that if one wants to maximize the left hand side of the sum rules type inequality \eqref{eq.sumruleineq} on the set of Borel probability  measures $\CM$, it is sufficient to use Dirac measures: $\mu=\delta_{\xi}$ for points $\xi \in \mathbb{R}$. Furthermore, if the Borel probability  measures $\mu$ are supported in $[-\Delta, \Delta]$ then
one only needs to consider Dirac measures: $\mu=\delta_{\xi}$ for points $\xi \in [-\Delta, \Delta]$. It has been stated in \cite[Theorem 14]{Cassier:2017:BHF} for Borel probability  measures $\CM$
which are compactly supported in $[-\Delta, \Delta]$ which corresponds to  the relation \eqref{eq.maxsumrule2}. But, the proof of relation \eqref{eq.maxsumrule1} remains the same  as the one in \cite[Theorem 14]{Cassier:2017:BHF} by replacing $[-\Delta, \Delta]$ by $\bbR$. Therefore, it is stated without proof.

\begin{Thm}\label{thm.mesopt}
Assume that the cloaking device satisfies assumptions \textnormal{H1--H6} and that $V_0\in H^{1/2}(\partial \Omega)$ is chosen such that $F_{V_0,\infty}>0$. In addition, if H7 is not satisfied, i.e., if the cloak contains a non-reciprocal material, assume that $V_0$ is real-valued. Let $[x_-,x_+] \subset \bbR$ (with $x_-< x_+$), then one has
\begin{equation}\label{eq.maxsumrule1}
\sup_{\mu \in\CM}  \frac{1}{\pi} \lim_{y\to 0^{+}}  \int_{x_-}^{x_+} \Imag H_{\mu, V_0}(x+\ii y) \, \md x=\sup_{\xi \in \mathbb{R}} \frac{1}{\pi} \lim_{y \to 0^{+}}  \int_{x_-}^{x_+} \Imag H_{ \delta_{\xi}, V_0}(x+\ii y) \, \md x
\end{equation}
and if the measures belongs to the subset $\CM_{\Delta}$ of measures of $\CM$ that are supported in $[-\Delta,\Delta]$:
\begin{equation}\label{eq.maxsumrule2} \sup_{\mu \in\CM_{\Delta}}  \frac{1}{\pi} \lim_{y\to 0^{+}}  \int_{x_-}^{x_+} \Imag H_{\mu, V_0}(x+\ii y) \, \md x=\sup_{\xi \in [-\Delta, \Delta]} \frac{1}{\pi} \lim_{y \to 0^{+}}  \int_{x_-}^{x_+} \Imag H_{ \delta_{\xi}, V_0}(x+\ii y) \, \md x.
\end{equation}
\end{Thm}

\subsection{The lossless case}\label{sec-losslesscase}
In this section, one assumes that the cloak is lossless on the frequency interval $[\Go_-, \Go_+]$ (i.e., a transparency window of the cloak; see \cite[Sec.\ 80 and 84]{Landau:1984:ECM}, \cite{Welters:2014:SLL, Cassier:2017:MMD, Cassier:2017:BHF} for more details) which implies (assuming \textnormal{H1} and \textnormal{H4}) that permittivity of the cloak satisfies the following hypothesis:
\begin{itemize}
\item[H8] {{\bf Lossless cloak} in $[\Go_-, \Go_+]$:} the cloak is assumed lossless on $[\Go_-, \Go_+]$, i.e.,
\begin{equation}\label{eq.lossless}
\forall \Go \in [\Go_-, \Go_+] \ \mbox{and for a.e. }  \Bx\in \Omega\setminus \mathcal{O}, \ \mathfrak{I}[\BGve(\Bx, \Go)]=0.
\end{equation}
\end{itemize}
The following Lemma shows that the functions $F_{V_0}$
and the Herlgotz function $H_{V_0}$ inherits a similar property from the permittivity if one assumes that the cloak is lossless on   $[\Go_-, \Go_+]$. 
\begin{Lem}\label{lem.DiffQuadFormDtNMapsLosslessCase}
Assume that the cloaking device satisfies assumptions 
$\textnormal{H1--H6}$ and $\textnormal{H8}$ (i.e., the cloak is lossless on the frequency interval $[\Go_-, \Go_+]$). Let  $V_0\in H^{1/2}(\partial \Omega)$ and, if $\textnormal{H7}$ is not satisfied (i.e., if the cloak contains a non-reciprocal material), let $V_0$ be a real-valued function. Then 
\begin{equation}\label{eq.transparencywindowF}
\operatorname{Im} F_{V_0}(\Go)=0 , \quad \forall \Go \in [\Go_-,\Go_+] \quad  \mbox{ and }  \quad \operatorname{Im} H_{V_0}(x)=0 , \quad \forall x \in [\Go_-^2,\Go_+^2],
\end{equation}
where $H_{V_0}$ is the Herglotz function defined by \eqref{eq.defHVo} in Theorem \ref{Thm.Hergphys}.
\end{Lem}
\begin{proof}
This property is again a consequence of the theory of composites via the connection between the DtN map and the effective tensor with the factorization \eqref{eq.factDtNmap} (one could also prove this via Green's formula, but we choose this approach to emphasize the connection to the theory of composites). 
Let $\BE_0=\Pi V_0$. Then, since $\BE_0\in \mathcal{U}$, for a fixed  $\omega \in [\Go_-,\Go_+] $, there exists a solution $(\BJ_0,\BE,\BJ)\in \mathcal{U}\times\mathcal{E}\times\mathcal{J}$ to the Dirichlet $Z$-problem $(\mathcal{H},\mathcal{U},\mathcal{E},\mathcal{J},\bf a)$, where $\bf a = \BGve(\cdot,\omega)$, at $\BE_0$ which implies that
    \begin{eqnarray*}
F_{V_0}(\Go)&=&     \langle [\Lambda_{ \BGve(\cdot,\omega)}- \Lambda_{ \varepsilon_0}] V_0,\overline{V_0} \rangle_{H^{-\frac{1}{2}}(\partial \Omega), H^{\frac{1}{2}}(\partial \Omega)}\\
 &=& \big( \big[ [ \BGve(\cdot,\omega)]_*- [\varepsilon_0]_* \big] \Pi V_0,\Pi V_0 \big)_{\mathcal{H}} \\ &=&\big(  [\BGve(\cdot,\omega) ]_* \BE_0, \BE_0\big)_{\mathcal{H}} -(\varepsilon_0\BE_0, \BE_0)_{\mathcal{H}}  \\
        &=&( [\BGve(\cdot,\omega)](\BE_0+\BE), \BE_0+\BE )_{\mathcal{H}} - \varepsilon_0 \|\BE_0\|^2_{\mathcal{H}},
    \end{eqnarray*}
where in the above equalities, we use that $\BJ_0=[\BGve(\cdot,\omega)]_*\BE_0$, $\BJ_0+\BJ= \BGve(\cdot,\omega) (\BE_0+\BE)$, $[\varepsilon_0]_*=\varepsilon_0$ [by Corollary \ref{cor:BasicPropertiesEffOpDiricZProb}.(ii)], and the fact that $\mathcal{U},\mathcal{E},\mathcal{J}$ are orthogonal subspaces in $\mathcal{H}$. Since \eqref{eq.permitivtyob} together with H8 ensures that \eqref{eq.lossless} remains valid when $\Omega \setminus \CO$ is replaced by $\Omega$, it follows that for all $\Go \in [\Go_-, \Go_+]$:
$$
\operatorname{Im}F_{V_0}(\Go)=\operatorname{Im}( \BGve(\cdot,\omega)(\BE_0+\BE), \BE_0+\BE )_{\mathcal{H}}- \operatorname{Im}\varepsilon_0 \|\BE_0\|^2_{\mathcal{H}} =( \mathfrak{I}[\BGve(\cdot,\omega)](\BE_0+\BE), \BE_0+\BE )_{\mathcal{H}}=0
$$
and, thus by \eqref{eq.defHVo}, we also have $\operatorname{Im} H_{V_0}(x)=x \operatorname{Im}  F_{V_0}(\sqrt{x})=0$ for all $x \in [\Go_-^2,\Go_+^2]$.
\end{proof}

The following theorem provides precise bounds when the cloak is lossless in the frequency interval $[\Go_-, \Go_+]$, i.e., when  H8 is satisfied.  Its proof has been  established in \cite[Proposition 15]{Cassier:2017:BHF} (where $H_{V_0}$ and $F_{V_0}$ play the role of the functions $v$ and $f$ in \cite{Cassier:2017:BHF}).
This proof relies  on the fact that $[\Go_-, \Go_+]$ is a transparency window for $F_{V_0}$, i.e., $F_{V_0}$ is real-valued on $[\Go_-, \Go_+]$ by \eqref{eq.transparencywindowF}. One uses this information to compute the limit in left hand-side of the inequality \eqref{eq.sumruleineq} when $\mu$  are well-chosen Dirac measures $\mu=\delta_{\xi}\in \CM$, for $\xi \in \mathbb{R}$, to get the bound \eqref{eq.ineqHV0}. Then, one deduces immediately from  \eqref{eq.ineqHV0} that the inequality \eqref{eq.boundtranspFV0} holds using the definition \eqref{eq.defHVo} of the Herglotz function $H_{V_0}$. We point out that bounds similar to \eqref{eq.boundtranspFV0} were derived without a sum rule approach for the permittivity of a passive linear electromagnetic media in \cite{Milton:1997:FFR} on a frequency interval where the material is assumed lossless based on interpolation  techniques related to Stieltjes function. In \cite{Cassier:2017:BHF}, the authors also give a second proof of \eqref{eq.boundtranspFV0}  based on  Kramers-Kronig relations.
\begin{Thm}\label{Th.boundtransp}
Let the cloaking device satisfies assumptions $\textnormal{H1--H6}$ and  $V_0\in H^{1/2}(\partial \Omega)$ being chosen such that $F_{V_0,\infty}>0$. Moreover, assume that $V_0$ is a real-valued function if $\textnormal{H7}$ is not satisfied (i.e., if the cloak contains a non-reciprocal material).
If $\textnormal{H8}$ holds, in other words the cloak is lossless in the positive frequency interval $[\Go_-, \Go_+]$ then the function $H_{V_0}$ satisfies
\begin{equation}\label{eq.ineqHV0}
F_{V_0, \infty} \, (x-x_0) \leq H_{V_0}(x)-H_{V_0}(x_0), \  \forall x, x_0 \in [\Go_-^2, \Go_+^2] \mbox{ such that} \ x_0\leq x,
\end{equation}
which yields the following bound on $F_{V_0}$:
\begin{equation}\label{eq.boundtranspFV0}
\Go_0^2 [F_{V_0}(\Go_0)-F_{V_0, \infty}] \leq \Go^2  [F_{V_0}(\Go)-F_{V_0 ,\infty}], \  \forall \Go, \Go_0 \in [\Go_-,\Go_+] \mbox{ such that} \ \Go_0\leq \Go.
\end{equation}
\end{Thm}

\begin{Rem}Note the bounds \eqref{eq.ineqHV0} and \eqref{eq.boundtranspFV0} are sharp in the following sense. For any fixed real number $\Go_0\in [\Go_-,\Go_+]$, suppose $F_{V_0}$ is given by the Drude dispersion law
$$F_{V_0}(\Go)=F_{V_{0,\infty}}-\frac{\Go_0^2}{\Go^2} [F_{V_{0,\infty}}- F_{V_0}(\Go_0)] \quad \mbox{ where } F_{V_0}(\Go_0)\in [0,F_{V_{0,\infty}} ].$$
Then, on one hand, $S_{V_0}$ defined by \eqref{eq.defSvo} is a Stieltjes function that satisfies all the properties of Lemma \ref{Const.Stielt} and thus $H_{V_0}$ defined by \eqref{eq.defHVo}, is given here by \begin{equation*} 
H_{V_0}(z)= F_{V_{0,\infty}} (z-x_0)+ H_{V_0}(x_0), \quad \mbox{ where for } \  x_0=\Go_0^2, \ \  H_{V_0}(x_0)=\Go_0^2F_{V{0}}(\Go_0). 
\end{equation*}
This is an affine Herglotz function which shares all the properties of Theorem \ref{Thm.Hergphys}. On the other hand, $H_{V_0}$ and  $F_{V_0}$ satisfy
\eqref{eq.transparencywindowF} and for these functions, the inequalities \eqref{eq.ineqHV0} (for all $x\in [x_-,x_+]=[\Go_-^2, \Go_+^2]$ and $x_0=\Go_0^2$ fixed) and \eqref{eq.boundtranspFV0}  (for all $\Go\in [\Go_-,\Go_+]$ and $\Go_0$ fixed)
become equalities.
\end{Rem}

\noindent The following corollary is an immediate consequence of the inequality \eqref{eq.boundtranspFV0} of Theorem \ref{Th.boundtransp} if one assumes that the obstacle is approximately cloaked at frequency $\Go_0$ in the sense of \eqref{eq.aprroximatecloaking}. 
\begin{Cor}\label{Cor.lossless}
Assume  the hypothesis of Theorem \ref{Th.boundtransp} for the cloaking device and  for the fixed input  signal $V_0$. Moreover, assume there exists a frequency $\Go_0\in [\Go_-, \Go_+]$ such that the cloak achieves approximate  cloaking at $\Go_0$, i.e., there exists $\eta>0$ such that 
\begin{equation}\label{eq.aprroximatecloaking}
|F_{\tilde{V}_0}(\Go_0)|\leq \eta \, G^{\operatorname{vac}}_{\tilde{V}_0} \ \mbox{ with } \ G^{\operatorname{vac}}_{\tilde{V}_0}=\langle \Lambda_{\varepsilon_0} \, \tilde{V}_0,\overline{\tilde{V}_0} \rangle_{H^{-\frac{1}{2}}(\partial \Omega), H^{\frac{1}{2}}(\partial \Omega)} ,  \quad \, \forall  \tilde{V}_0\in H^{1/2}(\partial \Omega)
\end{equation}
then
\begin{eqnarray}
 \displaystyle F_{V_0}(\Go) & \leq &  \big(-F_{V_0,\infty}+\eta \, G^{\operatorname{vac}}_{V_0} \big) \ \frac{\Go^2_0-\Go^2}{\Go^2} + \eta \,  G^{\operatorname{vac}}_{V_0}  \ \mbox{if}  \ \Go_-\leq  \Go\leq \Go_0, \label{eq.boundtransp1}\\[10pt]
 \displaystyle  F_{V_0}(\Go) & \geq & \big(F_{V_0,\infty}+\eta  \, G^{\operatorname{vac}}_{V_0} \big) \ \frac{\Go^2-\Go^2_0}{\Go^2}- \eta \,  G^{\operatorname{vac}}_{V_0} \ \  \mbox{if} \   \Go_0 \leq \Go \leq \Go_+ \label{eq.boundtransp2}.
\end{eqnarray}
In particular, if there exists a frequency $\Go_0\in [\Go_-, \Go_+]$ such that the cloak achieves  perfect cloaking at $\Go_0$, i.e., $F_{\tilde{V}_0}(\Go_0)=0$ for all $\tilde{V}_0\in H^{1/2}(\partial \Omega)$, then the inequalities \eqref{eq.boundtransp1} and \eqref{eq.boundtransp2}   hold with $\eta=0$. \\[4pt]
Furthermore, if $V_0\in H^{1/2}(\partial \Omega)$  is affine, i.e., $V_0=-{\bf e}_0\cdot \Bx|_{\partial\Omega}$ for  ${\bf e}_0\in \bbC^d\setminus \{0 \}$ when $\textnormal{H7}$ is satisfied (otherwise assume ${\bf e}_0\in \bbR^d$, i.e., if $\textnormal{H7}$ does not hold), then (by Corollary  \ref{Cor.boundleadingcoefficient}) we have $F_{V_0,\infty}>0$ and one can replace $F_{V_0, \infty}$ [in the inequalities \eqref{eq.boundtransp1}  and \eqref{eq.boundtransp2}] by its lower bound
\begin{equation}\label{eq.lowerboundtransp3}
|\mathcal{O}| \  \Big(1- \frac{\varepsilon_0} {\varepsilon} \Big) \ \    \frac{\varepsilon_0 \, |\Omega| \, \, \|{\bf e}_0\|_{\bbC^d}^2}{|\mathcal{O}|\,  \frac{\varepsilon_0}{\varepsilon}  + |\Omega \setminus \mathcal{O}|  }. 
\end{equation} 
\end{Cor}
Assuming perfect cloaking at $\Go_0$ [i.e., when $F_{\tilde{V}_0}(\Go_0)=0$, for all $\tilde{V}_0\in H^{1/2}(\partial \Omega)$]  and choosing  an affine boundary
input condition $V_0$, leads  to the bounds \eqref{eq.boundtransp1} and
\eqref{eq.boundtransp2}  with $\eta=0$  and  \eqref{eq.lowerboundtransp3} instead of  $F_{V_{0}, \infty}.$  Then these two precise inequalities gives fundamental limits to the cloaking effect on the frequency interval $[\Go_-, \Go_+]$ since in this case one has 

\begin{eqnarray}
 \displaystyle F_{V_0}(\Go)    & \leq  -\displaystyle  |\mathcal{O}| \  \Big(1- \frac{\varepsilon_0} {\varepsilon} \Big) \ \     \frac{\varepsilon_0 \, |\Omega| \, \, \|{\bf e}_0\|_{\bbC^d}^2}{|\mathcal{O}|\,  \frac{\varepsilon_0}{\varepsilon}  + |\Omega \setminus \mathcal{O}|  } \ \frac{\Go^2_0-\Go^2}{\Go^2} <0 & \mbox{if}  \ \Go\in [\omega_-, \omega_0), \label{eq.boundtransp1bis}\\[10pt]
  \displaystyle F_{V_0}(\Go)    & \geq  \displaystyle  |\mathcal{O}| \  \Big(1- \frac{\varepsilon_0} {\varepsilon} \Big) \ \    \frac{\varepsilon_0 \, |\Omega| \, \, \|{\bf e}_0\|_{\bbC^d}^2}{|\mathcal{O}|\,  \frac{\varepsilon_0}{\varepsilon}  + |\Omega \setminus \mathcal{O}|  } \ \frac{\Go^2-\Go_0^2}{\Go^2} >0 & \mbox{if}  \ \Go\in (\omega_0, \omega_+]\label{eq.boundtransp2bis}.
\end{eqnarray}
and hence perfect cloaking cannot occur for any other frequency in $[\Go_-,\Go_+]$ than $\Go_0$. But they give more precise information, in the sense that the lower bound \eqref{eq.boundtransp1bis} and the upper bound \eqref{eq.boundtransp2bis}  on $F_{V_0}(\Go)$ depend on physical parameters of the passive system and the input signal: the bandwidth (via the terms $ \pm |\Go^2-\Go_0^2|$),  and also [from \eqref{eq.lowerboundtransp3}] the volume of the obstacle $|\CO|$, the volume of the cloak $|\Omega \setminus \mathcal{O}|$, a lower bound on  the relative permittivity  of the obstacle via the term $\varepsilon/\varepsilon_0$, and (up to a $1/2$ factor) the energy of the electric field generated by the input signal in $\Omega$ (when $\Omega$ is filled with vacuum), i.e., $ \varepsilon _0\, |\Omega| \|{\bf e}_0\|_{\bbC^d}^2$. We point out that in the quasistatic approximation in the near-field setting, $ \varepsilon _0 \,|\Omega| \|{\bf e}_0\|_{\bbC^d}^2$ is a physical quantity comparable to the energy of the incident field in electromagnetic scattering theory. In comparison, the bounds derived in \cite{Cassier:2017:BHF} in the context of far-field cloaking does not depend on  the volume of the cloak via the term $|\Omega \setminus \mathcal{O}|$ since this information  is  not measurable in the far-field regime.

Now if the cloak does not perfectly cloak the object at $\omega_0$, but only approximate cloaking \eqref{eq.aprroximatecloaking} holds at $\omega_0$ with a sufficiently small $\eta > 0$, then our bounds still impose some constraints that prevent cloaking from occurring on $[\omega_-, \omega_+]$. Let us consider this next.

 For an affine potential $V_0$ given by $V_0 = -\mathbf{e}_0 \cdot \mathbf{x}\big|_{\partial\Omega}$ with $\mathbf{e}_0 \in \mathbb{C}^d \setminus \{0\}$, the unique solution $u\in H^1(\Omega)$ of the boundary-value problem \eqref{def.DirBVPForCoercivityTensors} for $\Ba=\varepsilon_0 \,\mathbf{I}$ is  affine as well:  $u(\mathbf{x}) =- \mathbf{e}_0 \cdot \mathbf{x}$ for all $\mathbf{x} \in \Omega$. Thus, one has by \eqref{eq.defGvac}:
 $$
G^{\operatorname{vac}}_{V_0}=  \varepsilon_0  \|\nabla u\|^2_{\mathcal{H}} = \varepsilon_0 \,  |\Omega| \,  \|{\bf e}_0\|_{\bbC^d}^2>0.
$$
Thus, as $V_0$ is affine,  using \eqref{eq.boundtransp1} with $F_{V_0,\infty}$ replaced by the expression \eqref{eq.lowerboundtransp3} gives
\begin{equation}\label{eq.approximatecloacking1}
F_{V_0}(\omega) \leq \varepsilon_0 \,  |\Omega| \,  \|{\bf e}_0\|_{\bbC^d}^2 \Big[ (-\eta_{\operatorname{lim}} +\eta)   \frac{\Go^2_0-\Go^2}{\Go^2} +\eta \Big], \quad \forall \Go \in [\Go_-,\Go_0],
\end{equation}
where
$$
 \eta_{\operatorname{lim}}=\Big(1- \frac{\varepsilon_0} {\varepsilon} \Big) \ \    \frac{ |\mathcal{O}| }{|\mathcal{O}|\, \displaystyle \frac{\varepsilon_0} {\varepsilon}+ |\Omega \setminus \mathcal{O}| }>0.
 $$
Hence if $\eta\in (0, \eta_{\lim})$, one gets that 
\begin{equation}\label{ineq.approximate1}
F_{V_0}(\omega)<0 \quad \mbox{ if } \ \Go \in \Big[\omega_-, \sqrt{1-\frac{\eta}{\eta_{\lim}}} \,\omega_0\Big).
\end{equation}
Thus if $\omega_-< (1-\eta/\eta_{\lim})^{1/2}\, \omega_0$ then the  quantitative bounds  \eqref{eq.approximatecloacking1}  shows that perfect cloaking can not occur at any frequency in the interval $[\omega_-,(1-\eta/\eta_{\lim})^{1/2} \omega_0) $. Similarly, using the bound \eqref{eq.boundtransp2}, one obtains that 
\begin{equation}\label{eq.approximatecloacking2}
F_{V_0}(\omega) \geq \varepsilon_0 \,  |\Omega| \,  \|{\bf e}_0\|_{\bbC^d}^2 \Big[ (\eta_{\operatorname{lim}} +\eta)   \frac{\Go^2-\Go_0^2}{\Go^2} -\eta \Big], \quad \forall \Go \in [\Go_0,\Go_+].
\end{equation}
Hence, one gets that 
\begin{equation}\label{ineq.approximate2}
F_{V_0}(\omega)>0 \quad \mbox{ if } \ \Go \in \Big(\sqrt{1+\frac{\eta}{\eta_{\lim}}} \,\omega_0, \omega_+\Big].
\end{equation}
Therefore, if $\omega_+>(1+\eta/\eta_{\lim})^{1/2}\, \omega_0$ then the  quantitative bound \eqref{eq.approximatecloacking2}  shows that perfect cloaking can not occur at any frequency in the interval $((1+\eta/\eta_{\lim})^{1/2} \omega_0, \omega_+] $.
Hence, the bounds \eqref{eq.approximatecloacking1} and \eqref{eq.approximatecloacking2}
prove that perfect cloaking and even approximate cloaking cannot occur if the bandwidth is sufficiently large. Again, the physical parameters: the volumes of the obstacle $|\mathcal{O}|$ and the cloak $|\Omega\setminus \mathcal{O}|$ and a lower bound $\varepsilon/\varepsilon_0$ on the relative permittivity of the obstacle  appear in these bounds via the explicit constant $\eta_{\lim}$. 

\subsection{The lossy case}\label{sec-lossycase}

We follow here a method introduced  in \cite{Bernland:2011:SRC},  in the context of passive linear electromagnetic systems, to derive bounds via sum rules (Theorem \ref{Thm.sumrules}) by choosing the uniform probability measure on the interval $[-\Delta, \Delta]$ for $\Delta>0$, namely,
\begin{equation}\label{eq.uniformmeasure}
\md \mu(\xi)= \frac{1_{[-\Delta,\Delta]}(\xi)}{2 \Delta}  \md \xi
\end{equation}
for the Herglotz function $h_{\mu}$ defined in \eqref{eq.Hergprob}. Here $1_{[-\Delta,\Delta]}(\xi)$ is the indicator function that takes the value $1$ on the interval $[-\Delta, \Delta]$ and is zero outside that interval. This method was applied the first time  in the context of cloaking  \cite{Cassier:2017:BHF} for the far-field quasistatic problem to derive inequalities  on the polarizability tensor of the passive cloaking device which prevent approximate cloaking on a frequency interval. In this present work, we use it in an original framework related to the DtN operator to derive fundamental limits  for the near-field  quasistatic cloaking  problem.
We recall here the main arguments of this  method. For additional details, we refer to \cite[Sec.\ II.F]{Cassier:2017:BHF}. 

Using the expression  \eqref{eq.uniformmeasure} of $\mu$ in  the definition \eqref{eq.defHmuVO} of the Herglotz function $H_{\mu,V_0}$ yields the following formula: 
\begin{equation}\label{eq.Hergunif}
H_{\mu, V_0}(z)=\frac{1}{2\GD} \int_{-\GD}^{\GD}\frac{1}{\xi-H_{V_0}(z)}\md \xi=\frac{1}{2\GD}\log \left( \frac{H_{V_0}(z)-\GD}{H_{V_0}(z)+\GD}\right), \quad \forall z \in \bbC^{+},
\end{equation}
where the definition of the complex logarithm function $\log$, we use  the same branch cut $\bbR^{+}$ as for the square root function \eqref{eq.racine}.
Next, one has
\begin{equation*}
 \frac{H_{V_0}(z)-\GD}{H_{V_0}(z)+\GD}=\frac{|H_{V_0}(z)|^2-\GD^2+2\ii \GD \Imag[H_{V_0}(z)]}{|H_{V_0}(z)+\GD|^2}, \quad \forall z \in \bbC^{+}. 
\end{equation*}
From these formulas and since $\Imag[H_{V_0}(z)]\in \bbC^+$ for $z\in \bbC^+$ (because $H_{V_0}$ is a non-constant Herglotz function), one deduces that
\begin{equation}\label{eq.boundimag}
 \Imag  H_{\mu,V_0}(z)=\frac{1}{2\GD}\operatorname{arg}\left( \frac{H_{V_0}(z)-\GD}{H_{V_0}(z)+\GD}\right) \geq \frac{\pi}{4 \GD} H(\GD-|H_{V_0}(z)|), \quad  \forall z \in \bbC^{+},
\end{equation}
where $H=\mathbf{1}_{\bbR^{+}}$ is the Heaviside function (i.e., the indicator function of $\mathbb{R}^+$).
We use now the inequality \eqref{eq.boundimag} to get a lower bound on the left hand side of \eqref{eq.sumruleineq}.\\[4pt]
One introduces the set $I_{\Delta,H_{V_0}}=\{ x\in [x_-,x_+] \mid |H_{V_0}(x)|=\Delta\}$ with $[x_-,x_+] =[\Go_-^2, \Go_+^2]$. The complement of this set  in $[x_-,x_+]$  corresponds to the set of points $x$ in $[x_-,x_+]$  where the function $t\mapsto H(\Delta-|H_{V_0}(t)|)$ is continuous in a neighborhood of $x$ in $[x_-,x_+]\cup \bbC^+$. Thus, assuming that the Lebesgue measure of $I_{\Delta,H_{V_0}}$ in $\mathbb{R}$, which we denote by $|I_{\Delta,H_{V_0}}|$, is zero and using the continuity  of $H_{V_0}$ on $[x_-,x_+]$ [see Theorem  \eqref{Thm.Hergphys}], one obtains  via the Lebesgue's dominated convergence theorem and the inequality \eqref{eq.boundimag} that 
$$
\lim_{y\to 0^{+}}  \int_{x_-}^{x_+}  H(\GD-|H_{V_0}(x+\ii y )|)\, \md x  \!=\!  \int_{x_-}^{x_+}  H(\GD-|H_{V_0}(x)|)\, \md x
\leq \frac{4 \Delta}{\pi} \lim_{y\to 0^{+}} \int_{x_-}^{x_+} \operatorname{Im} H_{\mu, V_0} (x+\ii y) \, \md x$$
[where the existence of the limit of the right hand side is guaranteed by formula \eqref{eq.mesHerg} in Corollary \ref{cor.Herg}].
Hence, under the assumption $|I_{\Delta,H_{V_0}}|=0$ and by virtue of the general bound \eqref{eq.sumruleineq},  we get the following   less stringent
but more transparent inequality than \eqref{eq.sumruleineq}:
\begin{equation}\label{eq.bound1losscase}
 |\{ x\in [x_-,x_+] \mid |H_{V_0}(x)|\leq\Delta\}|=\int_{x_-}^{x_+}H(\GD-|H_{V_0}(x)|)\, \md x
\leq\frac{4\GD}{F_{V_0,\infty}} .
 \end{equation}
In fact, the inequality \eqref{eq.bound1losscase} holds without the assumption $|I_{\Delta,H_{V_0}}|=0$. More precisely, $\{\Delta\in \bbR^{+,*} \mid |I_{\Delta,H_{V_0}}|>0\}$ is at most a countable set (since  $|H_{V_0}|$, as a continuous function on $[x_-,x_+]$,  is measurable on this closed interval). Thus, if $|I_{\Delta,H_{V_0}}|>0$ for $\Delta>0$, one can  find a sequence of positive real numbers $(\Delta_n)$  satisfying $\Delta_n>\Delta$ and $|I_{\Delta_n,H_{V_0}}|=0$ for all $n\in \mathbb{N}$ which converge to $\Delta$ as $n\to +\infty$. Hence, it follows with \eqref{eq.bound1losscase} that
 \begin{equation}\label{eq.bound1losscasebis}
 \int_{x_-}^{x_+}H(\GD-|H_{V_0}(x)|)\, \md x \leq  \int_{x_-}^{x_+}H(\GD_n-|H_{V_0}(x)|)\, \md x
\leq\frac{4\GD_n}{F_{V_0,\infty}} \to \frac{4\GD}{F_{V_0,\infty}} \mbox{ as } n\to+\infty.
  \end{equation}
  The quantity on the left of \eqref{eq.bound1losscasebis} represents the total length of the set, between $x_-$ and $x_+$, where the function $|H_{V_0}|$ is less  or equal than $\GD$. 
  Clearly, the bound implies that this total length has to tend  to zero as $\GD\to 0$. In particular, if we take
\begin{equation}\label{Delta-max} 
\GD=\max_{x\in[x_-,x_+]}|H_{V_0}(x)|,
\end{equation}
 (where $0<\Delta<\infty$ since $H_{V_0}$ is continuous and does not vanish on $[x_-,x_+]=[\Go_-^2,\Go_+^2]$ by Proposition  \ref{Pro.perfectcloaking}) then  the left hand side of \eqref{eq.bound1losscasebis}  is $x_+ - x_-$ and thus one obtains
\begin{equation}\label{eq.BoundHV0lossy}
\frac{1}{4}(x_+-x_-) F_{V_0,\infty}\leq \max_{x\in[x_-,x_+]}|H_{V_0}(x)|. 
\end{equation}
Finally, from  the inequality \eqref{eq.BoundHV0lossy} and the definition \eqref{eq.defHVo} of $H_{V_0}$, one gets immediately the following theorem. Also, the next corollary follows immediately from this theorem and Corollary \ref{Cor.boundleadingcoefficient}. One should note that inequality \eqref{eq.bound2losscase} below is a direct analogy of the bound in \cite[p.\ 24, (3.23)]{Cassier:2017:BHF} on the polarizability tensor (which is a consequence of \cite[Proposition 16]{Cassier:2017:BHF}).

\begin{Thm}\label{Th.lossy}
Let the cloaking device satisfy assumptions $\textnormal{H1--H6}$ and  $V_0\in H^{1/2}(\partial \Omega)$ being chosen such that $F_{V_0,\infty}>0$. Moreover, assume that $V_0$ is a real-valued function if $\textnormal{H7}$ is not satisfied (i.e., if the cloak contains a non-reciprocal material). Then the function $F_{V_0}$ satisfies the following inequality:
\begin{equation}\label{eq.bound2losscase}
\frac{1}{4}(\Go_+^2-\Go_-^2) F_{V_0,\infty} \leq \max_{\Go\in[\Go_-,\Go_+]}|\omega^2 F_{V_0}(\Go)|.
\end{equation}
\end{Thm}
\begin{Cor}\label{Cor.boundlossy}
\hspace{-0.91678pt}Assume the cloaking device satisfies assumptions $\textnormal{H1--H6}$ and that $V_0\in H^{1/2}(\partial \Omega)$  is affine, i.e., $V_0=-{\bf e}_0\cdot \Bx|_{\partial\Omega}$ for a fixed ${\bf e}_0\in \bbC^d \setminus \{0 \}$ whenever $\textnormal{H7}$ is satisfied (otherwise assume ${\bf e}_0\in \bbR^d$, i.e., if $\textnormal{H7}$ does not hold),
then
\begin{equation}\label{eq.boundlossy}
\frac{1}{4}(\Go_+^2-\Go_-^2)  |\mathcal{O}| \Big(1- \frac{\varepsilon_0} {\varepsilon} \Big) \ \       \frac{1}{|\mathcal{O}|\,  \frac{\varepsilon_0}{\varepsilon}  + |\Omega \setminus \mathcal{O}|  }  \  \varepsilon_0 \ |\Omega| \, \, \|{\bf e}_0\|_{\bbC^d}^2 \leq \max_{\Go\in[\Go_-,\Go_+]}|\omega^2 F_{V_0}(\Go)|.
\end{equation}
\end{Cor}

The bound \eqref{eq.boundlossy} provides a fundamental limit to the cloaking effect by giving a lower bound on the maximum of the function  $\Go\mapsto|\omega^2 F_{V_0}(\Go)|$ on the frequency interval. Again, as with the bounds  \eqref{eq.boundtransp1} and
\eqref{eq.boundtransp2}, this bound 
depends on physical parameters of the passive  system and the input signal: the bandwidth, the volume of the obstacle $|\CO|$, the volume of the cloak $|\Omega \setminus \mathcal{O}|$, a lower bound $\varepsilon/\varepsilon_0$ on the the relative permittivity  of the obstacle, and the energy $ \varepsilon_0 \ |\Omega| \, \, \|{\bf e}_0\|_{\bbC^d}^2 $
of the electric  field generated by the input signal in $\Omega$ when $\Omega$ is filled with vacuum.
They are less precise than the bounds of Corollary \ref{Cor.lossless} since one controls the maximum by a positive quantity. Indeed, these bounds do not tell if $F_{V_0}$ can vanish at some particular frequency on the frequency interval.
However, we want to note that it can be applied by replacing $[\Go_-,\Go_+]$ by any subinterval $[\tilde{\Go}_-,\tilde{\Go}_+]\subseteq [\Go_-,\Go_+]$.
Thus, it provides a positive lower bound on the maximum  of $\Go\mapsto|\omega^2 F_{V_0}(\Go)|$ on any subinterval of $[\Go_-,\Go_+]$. Finally, we point out that choosing other value for $\GD>0$ than the maximum   \eqref{Delta-max} can provide   fundamental limits to cloaking via computing (numerically or experimentally) the left-hand side of the bound \eqref{eq.bound1losscasebis}.

Let us consider now the case of perfect cloaking at a single frequency $\omega_0$, i.e., $F_{V_0}(\omega_0)=0$. Let us assume that $F_{V_0}(\cdot)\in \mathcal{C}^1([\omega_-,\omega_+])$.  Then using the above results we can derive a bound on the derivative $F_{V_0}^{\prime}(\omega_0)$, where $()^{\prime}=\frac{d}{d\omega}$. Indeed, if $[a,b]\subseteq [\omega_-,\omega_+]$ then by the Taylor-Lagrange inequality 
\begin{gather*}
    |b^2F_{V_0}(b)-a^2F_{V_0}(a)|\leq \max_{\omega\in [a,b]}\left\vert\frac{d}{d\omega}[\omega^2F_{V_0}(\omega)]\right\vert(b-a).
\end{gather*}
By virtue of \eqref{eq.bound2losscase}, this implies that
\begin{gather*}
    \frac{1}{4}(\Go_+^2-\Go_-^2) F_{V_0,\infty} \leq \max_{\Go\in[\Go_-,\Go_+]}|\omega^2 F_{V_0}(\Go)|=\max_{\Go\in[\Go_-,\Go_+]}|\omega^2 F_{V_0}(\Go)-\omega_0^2 F_{V_0}(\Go_0)| \\
   \leq \max_{\Go\in[\Go_-,\Go_+]}\left\vert\frac{d}{d\omega}[\omega^2F_{V_0}(\omega)]\right\vert(\omega_+-\omega_-)
\end{gather*}
implying
\begin{gather}
    \frac{1}{4}(\Go_++\Go_-) F_{V_0,\infty}\leq \max_{\Go\in[\Go_-,\Go_+]}\Big\vert\frac{d}{d\omega}[\omega^2F_{V_0}(\omega)]\Big\vert.\label{IneqLossyCloakAtSingleFreqOnMaxDeriv}
\end{gather}
In particular, if we shrink the interval $[\omega_-,\omega_+]$ to a single point $\{\omega_0\}$ then by continuity
\begin{gather*}
    \frac{1}{2}\Go_0 F_{V_0,\infty}\leq \left\vert\frac{d}{d\omega}[\omega^2F_{V_0}(\omega)]|_{\omega=\omega_0}\right\vert=\omega_0^2\left\vert F_{V_0}^{\prime}(\omega_0)\right\vert,
\end{gather*}
which proves the inequality
\begin{gather}
    F_{V_0,\infty}\leq 2\omega_0\left\vert F_{V_0}^{\prime}(\omega_0)\right\vert.\label{IneqLossyCloakAtSingleFreqOnDeriv}
\end{gather}

Now if we assume that $V_0\in H^{1/2}(\partial \Omega)$ is affine $V_0=-{\bf e}_0\cdot \Bx|_{\partial\Omega}$ as in Corollary \ref{Cor.boundlossy}, then it follows from this corollary and the proof just given that
\begin{eqnarray}
  \displaystyle  \frac{1}{4}(\Go_++\Go_-)  |\mathcal{O}| \Big(1- \frac{\varepsilon_0} {\varepsilon} \Big) \ \  \frac{\varepsilon_0 \, |\Omega| \, \, \|{\bf e}_0\|_{\bbC^d}^2}{|\mathcal{O}|\,  \frac{\varepsilon_0}{\varepsilon}  + |\Omega \setminus \mathcal{O}|  } & \displaystyle\leq &\max_{\Go\in[\Go_-,\Go_+]}\Big\vert\frac{d}{d\omega}[\omega^2F_{V_0}(\omega)]\Big\vert, \label{IneqLossyCloakAtSingleFreqOnMaxDerivAffineBC}\\[6pt]
     \displaystyle  |\mathcal{O}| \Big(1- \frac{\varepsilon_0} {\varepsilon} \Big) \ \    \frac{\varepsilon_0 \, |\Omega| \, \, \|{\bf e}_0\|_{\bbC^d}^2}{|\mathcal{O}|\,  \frac{\varepsilon_0}{\varepsilon}  + |\Omega \setminus \mathcal{O}|  } &\leq & 2\omega_0\left\vert F_{V_0}^{\prime}(\omega_0)\right\vert .\label{IneqLossyCloakAtSingleFreqOnDerivAffineBC}
\end{eqnarray}
Thus we have proven the following corollary of Theorem \ref{Th.lossy} and Corollary \ref{Cor.boundlossy}. 
\begin{Cor}\label{cor-bound-dissipative}
    If the assumptions in Theorem \ref{Th.lossy} are satisfied, $F_{V_0}(\cdot)\in \mathcal{C}^1([\omega_-,\omega_+])$ and $\omega_0\in[\omega_-,\omega_+]$ with $F_{V_0}(\omega_0)=0$ then inequalities \eqref{IneqLossyCloakAtSingleFreqOnMaxDeriv} and \eqref{IneqLossyCloakAtSingleFreqOnDeriv} hold. Moreover, if $V_0$ is affine satisfying the hypotheses of Corollary \ref{Cor.boundlossy} then inequalities \eqref{IneqLossyCloakAtSingleFreqOnMaxDerivAffineBC} and \eqref{IneqLossyCloakAtSingleFreqOnDerivAffineBC} also hold.
\end{Cor}

\begin{Rem}
Corollary \ref{cor-bound-dissipative} shows that if $F_{V_0} \in \mathcal{C}^1([\omega_-, \omega_+])$ for a given affine boundary condition $V_0 \in H^{1/2}(\Omega)$, then the set of frequencies at which perfect cloaking can occur within the interval $[\omega_-, \omega_+]$ is finite. Indeed, this set is contained in  
$
S_{V_0} := \{ \omega \in [\omega_-, \omega_+] \mid F_{V_0}(\omega) = 0 \}.
$ 
From inequality~\eqref{IneqLossyCloakAtSingleFreqOnDerivAffineBC}, at any frequency $\omega_0 \in S_{V_0}$, one has $F_{V_0}(\omega_0) = 0$ and $F'_{V_0}(\omega_0) \neq 0$. Consequently, by the compactness of $[\omega_-, \omega_+]$ and the regularity of 
$F_{V_0}$, we deduce that 
$S_{V_0}$ is a finite set.\end{Rem}

\begin{Rem}\label{rem.EpsilMoreRegularity}
Compared with the other results, Corollary~\ref{cor-bound-dissipative} requires the additional assumption that 
$F_{V_0}(\cdot)\in \mathcal{C}^{1}([\omega_-,\omega_+])$.
A sufficient condition for this regularity is to replace Hypothesis~H4 with the following assumption: the map
$\omega \mapsto \omega\,\varepsilon(\mathbf{x},\omega)$
admits an analytic extension from $\mathbb{C}^+$ to a neighborhood of the real interval $[\omega_-,\omega_+]$.  
In other words, we assume that the function $\omega \mapsto \omega\,\varepsilon(\mathbf{x},\omega)$ is analytic in an open set containing the frequency range $[\omega_-,\omega_+]$. For instance, this is the case if the permittivity  $\varepsilon(\mathbf{x},\omega)$ is given for a.e.\ $\Bx\in\Omega\setminus \mathcal{O}$ by a generalized dissipative Drude-Lorentz model (see, e.g., \cite{Cassier:2025:LBS,Cassier-Joy:2025operator1}). Under this analyticity assumption, one can use the same reasoning as in Step 2 of Proposition \ref{Pro:DtNMapQuadFormIsHerglotz}, but now applied for $\omega \in [\omega_-,\omega_+]$ (instead of $\omega \in \mathbb{C}^+$), to conclude that $F_{V_0}$ is analytic (and a fortiori $\mathcal{C}^1$) in a neighborhood of  $[\omega_-,\omega_+]$.
\end{Rem}

\subsection{Generalization of the lossless case to include dispersive obstacles}\label{sec:GenOurCloakingBnds}
We show here  that if the cloak is lossless in the  frequency interval $[\Go_-,\Go_+]$ (i.e., if H8 holds) then one can extend  the bounds  of Corollary \ref{Cor.lossless}  to the case  of a (possibly anisotropic)  passive dispersive lossless and reciprocal material whose   real-valued  permittivity tensor  $\BGve_{\mathrm{ob}}(\Bx, \Go)$  is ``larger" than the vacuum on  $[\Go_-,\Go_+]$. More precisely, in this section we assume that the permittivity $\BGve_{\mathrm{ob}}$ of the obstacle satisfies:
\begin{itemize}
\item the assumptions H1, H2, and H4 (with $\Omega\setminus \mathcal{O}$  replaced by  $\mathcal{O}$);
\item  $\forall \Go \in [\Go_-,\Go_+]$, $\BGve_{\mathrm{ob}}(\cdot,\Go) \in M_d(L^{\infty}(\mathcal{O}))$ and [instead of \eqref{eq.permitivtyob}] it satisfies for a.e.\ $\Bx\in \mathcal{O}$ and $\forall \Go \in [\Go_-,\Go_+]$:  
\begin{equation}\label{eq.Displossless}
 \BGve_{\mathrm{ob}}(\Bx, \Go)=\BGve_{\mathrm{ob}}(\Bx, \Go)^{\top}, \ \mathfrak{I}[\BGve_{\mathrm{ob}}(\Bx,\Go)]=0, \mbox{ and }  \BGve_{\mathrm{ob}}(\Bx,\Go)\geq\Gve \BI \mbox{ with } \Gve\in \mathbb{R}, \, \Gve>\Gve_0.
\end{equation} 
\end{itemize}
Finally, we also assume that H8 holds, from which it now follows that 
\begin{equation}\label{eq.Dispob}
 \mbox{for a.e.\ } \Bx \in \Omega \mbox{ and }\forall \Go \in [\Go_-,\Go_+],   \, \mathfrak{I}[\BGve(\Bx,\omega)]=0. \
\end{equation}

\begin{Lem}\label{Lem.monotinicity}  Assume the  permittivity of the obstacle satisfies  $\textnormal{H1, H2, H4}$ (with $\mathcal{O}$ instead of $\Omega\setminus \mathcal{O}$), and  \eqref{eq.Displossless} then for a.e. $\Bx\in \mathcal{O}$, one has
\begin{equation}\label{eq.monotonicity1}
 \BGve_{\mathrm{ob}}(\Bx,\Go_0) \leq \BGve_{\mathrm{ob}}(\Bx,\Go),  \quad  \forall \omega_0, \omega \in  [\Go_-, \Go_+] \mid \omega_0 \leq \omega.
\end{equation}
\end{Lem}
\begin{proof}
Let $\bU\in \mathbb{C}^d$ be fixed. Using the  assumption H1 and H2, one shows in the same way as in the proof of Lemma \ref{Const.Stielt}  a formula of the type \eqref{eq.stieldtn2} [by replacing the Herglotz function $h_{\operatorname{dif}, V_0}$ by the Herglotz functions $\omega \mapsto \omega \BGve(\Bx,\Go)  \bU \cdot \overline{\bU}$], namely, for a.e.\ $x\in \mathcal{O}$, one has
\begin{equation}\label{eq.permittivity2}
\BGve_{\mathrm{ob}}(\Bx,\Go) \bU \cdot \overline{\bU}= a_{\Bx,\Bu}- \int_{\bbR}  \frac{\md \mu_{\Bx,\bU}( \xi)}{\Go^2-\xi^2},  \quad \forall \Go \in \bbC^+,
\end{equation}
where  $a_{ \Bx,\Bu}$ and $\mu_{\Bx,\bU}$ are the positive coefficient and (even) measure  of the Herglotz function $\omega \mapsto \omega \BGve_{\mathrm{ob}}(\Bx,\Go)  \bU \cdot \overline{\bU}$ in the representation formula \eqref{eq.defhergl}. Moreover, using H4 and \eqref{eq.Dispob}, one gets from the formula \eqref{eq.mesHerg} that the support of the measure $\mu_{\Bx,\bU}$ is contained in $Y=\bbR\setminus [ (-\omega_+, -\omega_-)\cup  (\omega_-, \omega_+) ]$. Thus, the assumption H1 to H4 and the Schwartz reflection principle show that one  extends analytically  the formula \eqref{eq.permittivity2} to $\bbC\setminus Y $. Thus, in particular, one can differentiate  this formula under the integral (see, e.g., \cite{Mattner:2001:CDI})  on $(\Go_-, \Go_-)$ to obtain 
\begin{equation}\label{eq.permittivity3}
\frac{\mathrm{d} }{\mathrm{d} \Go}[ \BGve_{\mathrm{ob}}(\Bx,\cdot)\bU \cdot \overline{\bU}](\Go)=   2\Go  \int_{\bbR}  \frac{\md \mu_{\Bx, \BU}( \xi)}{(\Go^2-\xi^2)^2}\geq 0,  \quad \forall \Go \in (\Go_-,\Go_+) .
\end{equation}
Thus, \eqref{eq.permittivity3}  and the continuity of $\BGve_{\mathrm{ob}}(\Bx,\cdot)$ at $\Go_{\pm}$ (by hypothesis H4) clearly implies \eqref{eq.monotonicity1}. 
\end{proof}
We extend now the  bounds obtained  in Corollary \ref{Cor.lossless} to the case  of a dispersive obstacle. The key ingredients are the above monotonicity property, the Dirichlet variational principle \eqref{eq.Dirichletpriciple} (which holds when the cloak is lossless), and the analytical bounds of Corollary \ref{Cor.lossless} obtained for a non-dispersive obstacle.

To begin,  using the fact  $\BGve_{\mathrm{ob}}(\cdot,\Go) \in M_d(L^{\infty}(\mathcal{O}))$, \eqref{eq.Displossless},  and the assumptions H5 and  H6 on the cloaking device, one can still  define (via Proposition \ref{pro.DtNDirBVPResults}) on $[\Go_-,\Go_+]$,  the function $F_{V_0}$ defined for all $V_0\in  H^{\frac{1}{2}}(\partial \Omega)$ by
\begin{equation}\label{eq.FVObis}
F_{V_0}(\Go)=\langle [\Lambda_{\BGve(\cdot, \omega)} -\Lambda_{\BGve_0}]V_0,\overline{V_0} \rangle_{H^{-\frac{1}{2}}(\partial \Omega), H^{\frac{1}{2}}(\partial \Omega)},
\end{equation}
when  the obstacle is dispersive (corresponding to the right situation of Figure \ref{fig.comp}).

Let $\Go_0$ be a fixed reference frequency in $[\Go_-,\Go_+]$ where we will assume in the following that approximate cloaking holds. Then, one can define [as previously, see \eqref{eq.defFV}] a  function $F_{V_0, \mathrm{ref}}$ on $\bbC^+\cup [\Go_-,\Go_+]$ via the  assumptions H1, H3, H4, H5 and H6 on the cloaking device
\begin{equation}\label{eq.FVOref}
F_{V_0, \mathrm{ref}}(\Go)=\langle [\Lambda_{\BGve_{\mathrm{ref}}(\cdot, \omega)} -\Lambda_{\BGve_0}]V_0,\overline{V_0} \rangle_{H^{-\frac{1}{2}}(\partial \Omega), H^{\frac{1}{2}}(\partial \Omega)}
\end{equation}
for a reference permittivity  $\BGve_{\mathrm{ref}}(\cdot, \omega)$  defined   for all $\omega\in  \bbC^+\cup [\Go_-,\Go_+]$ by
$$
\BGve_{\mathrm{ref}}(\cdot, \omega)=\BGve(\cdot, \omega) \ \mbox{ on }   \ \Omega \setminus \mathcal{O}  \quad \mbox{ and } \quad \BGve_{\mathrm{ref}}(\cdot, \omega):=\BGve_{\mathrm{ob}}(\cdot,\Go_0)  \mbox { on  } \CO.$$
In other words, $\BGve_{\mathrm{ref}}$  coincides with $\BGve$ within the cloak $\Omega \setminus \mathcal{O}$
and  $\BGve_{\mathrm{ref}}$ has frequency independent  permittivity  $\BGve_{\mathrm{ob}}(\cdot,\Go_0)$ within the obstacle  $\CO$ (this corresponds to the left situation of Figure \ref{fig.comp}). Thus,  for this reference situation, we are back to the same assumptions as previously since the reference  permittivity is non-dispersive in the obstacle and [by \eqref{eq.Dispob} evaluated at $\Go=\Go_0$] it satisfies \eqref{eq.passive-cloak} with $\BGve_{\mathrm{ob}}(\Bx) =\BGve_{\mathrm{ob}}(\Bx,\Go_0)  $.  We denote by $F_{V_0,\mathrm{ref},\infty}$ its limit when $|\Go|\to\infty$ (along the ``positive" imaginary axis)  defined in Theorem  \ref{thm.differenceHerg} by \eqref{eq.limitinfherg}. Finally, we point out that by construction $F_{V_0,\mathrm{ref}}(\Go_0)=F_{V_0}(\Go_0)$.

 \begin{figure}[!ht]
\centering
 \includegraphics[width=0.95\textwidth]{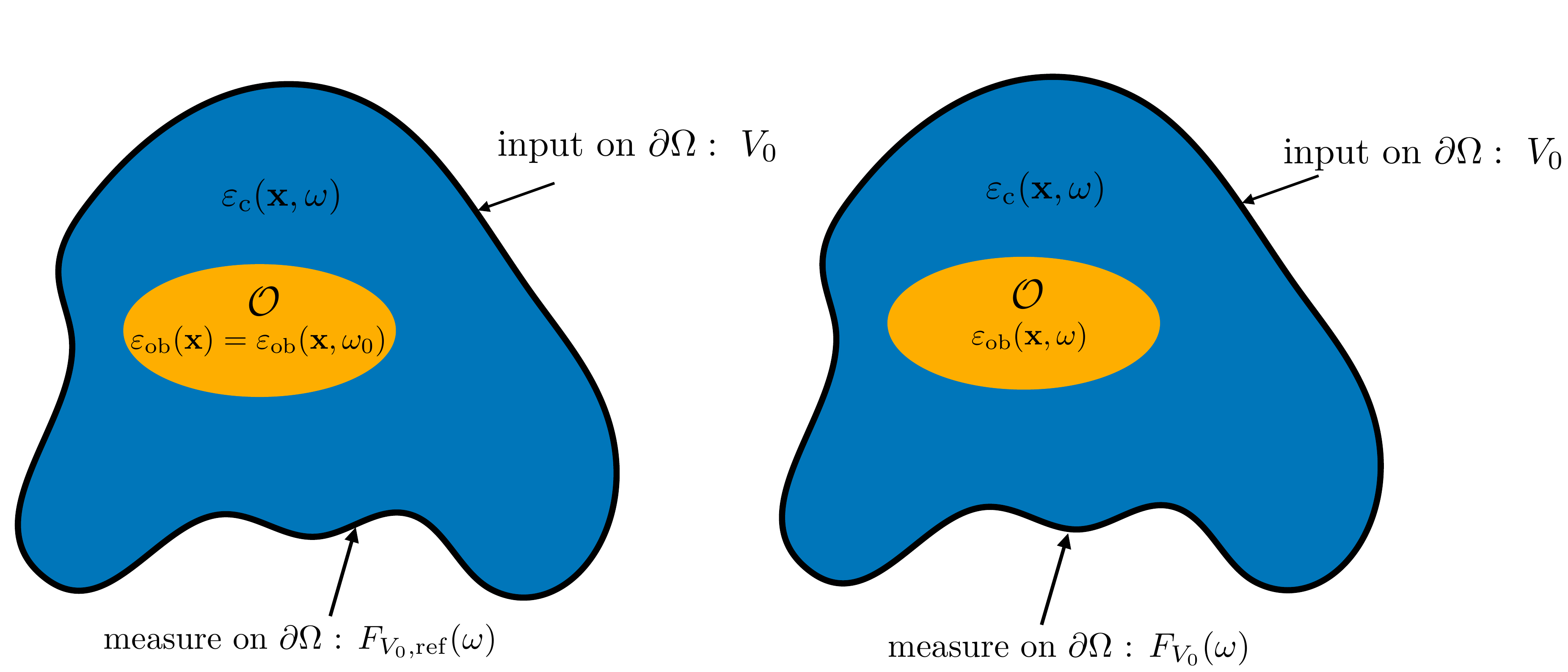}
 \caption{Comparison  of two situations. The difference between the two is that for the left figure, the obstacle is not dispersive on $[\Go_-, \Go_+]$ whereas it is dispersive in the right figure. In the left,  $\BGve_{\mathrm{ob}}(\Bx,\Go)$ is  chosen such that  $\BGve_{\mathrm{ob}}(\Bx) :=\BGve_{\mathrm{ob}}(\Bx,\Go_0)  $ on  $\CO$ and $[\Go_-, \Go_+]$, where $\BGve_{\mathrm{ob}}(\cdot,\Go_0)$ is the permittivity  of the dispersive obstacle on the right at a  reference frequency $\Go_0\in [\Go_-, \Go_+]$, where approximate cloaking occurs. Thus, the two physical setups coincide at $\Go_0$.}
 \label{fig.comp}
\end{figure}

\begin{Cor}\label{cor:ApproxCloakingWDisperObj}
Let  the obstacle satisfies  $\textnormal{H1, H2 and H4}$ (with $\mathcal{O}$ instead $\Omega\setminus \CO$)  and  the relation \eqref{eq.Displossless}. Assume that the cloaking device satisfy assumptions $\textnormal{H1--H6}$ and is lossless on the frequency range $[\Go_-,\Go_+]$ (i.e., that $\textnormal{H8}$ holds). Moreover, assume that $V_0$ is a real-valued function if $\textnormal{H7}$ is not satisfied (i.e., if the cloak contains a non-reciprocal material). Then, if  there exists a frequency $\Go_0\in [\Go_-, \Go_+]$ such that the cloak achieves approximate  cloaking at $\Go_0$, i.e., there exists $\eta> 0$ such that 
\begin{equation}\label{eq.aprroximatecloakingbis}
|F_{\tilde{V}_0}(\Go_0)|\leq \eta \,  G^{\operatorname{vac}}_{\tilde{V}_0},  \quad \, \forall  \tilde{V}_0\in H^{1/2}(\partial \Omega)
\end{equation}
then, for every $V_0\in H^{1/2}(\partial\Omega)$ such that $F_{V_0,\mathrm{ref},\infty}>0$,
\begin{eqnarray}
  F_{V_0}(\Go) \leq F_{V_0,\mathrm{ref},\infty}(\Go) & \leq &  (-F_{V_0,\mathrm{ref},\infty}+\eta\,  G^{\operatorname{vac}}_{V_0}) \, \frac{\Go^2_0-\Go^2}{\Go^2} + \eta \,  G^{\operatorname{vac}}_{V_0} \ \mbox{ if }  \ \Go\in [\Go_-,\Go_0], \label{eq.boundtransp3}\\[5pt]
F_{V_0}(\Go) \geq  F_{V_0,\mathrm{ref},\infty}(\Go) & \geq & ( F_{V_0,\mathrm{ref},\infty}+\eta \, G^{\operatorname{vac}}_{V_0}) \, \frac{\Go^2-\Go^2_0}{\Go^2}- \eta \, G^{\operatorname{vac}}_{V_0}\ \mbox{ if } \ \omega\in  [\Go_0,\Go_+]\label{eq.boundtransp4}.
\end{eqnarray}
Furthermore, if $V_0\in H^{1/2}(\partial \Omega)$  is affine, i.e., $V_0=-{\bf e}_0\cdot \Bx|_{\partial\Omega}$ for a fixed ${\bf e}_0\in \bbC^d\setminus\{0\}$ when $\textnormal{H7}$ is satisfied (otherwise assume ${\bf e}_0\in \bbR^d$, i.e., if $\textnormal{H7}$ does not hold), then one can replace $F_{V_0,\mathrm{ref}, \infty}$ [in inequalities \eqref{eq.boundtransp3}  and \eqref{eq.boundtransp4}] by 
\eqref{eq.lowerboundtransp3}.
\end{Cor}
\begin{proof}
Use the monotonicity property on  the permittivity on the obstacle proved in Lemma \ref{Lem.monotinicity} and the fact that $\BGve(\cdot, \omega)$ and  $\BGve_{\mathrm{ref}}(\cdot, \omega)$ coincides on the cloak, one obtains 
\begin{eqnarray}
 \BGve_{\mathrm{ref}}(\Bx, \omega) \leq& \hspace{-1em}\BGve(\Bx,\Go),& \quad  \forall \omega_0, \omega \in  [\Go_-, \Go_+] \mid \omega_0 \leq \omega, \ \mbox{and  a.e } \Bx\in \Omega, \label{eq.monotonicity2}\\[4pt]
 \BGve(\Bx, \omega)  \leq& \BGve_{\mathrm{ref}}(\Bx,\Go),&  \quad  \forall \omega_0, \omega \in  [\Go_-, \Go_+] \mid \omega\leq \omega_0,  \ \mbox{and  a.e $\Bx\in \Omega$. } \label{eq.monotonicity3}
\end{eqnarray}
Furthermore, we let the reader check that by virtue of \eqref{eq.Displossless}, \eqref{eq.Dispob} and H6, one has 
for all $\omega \in  [\Go_-, \Go_+]$ that $\BGve(\cdot, \omega)$ and $\BGve_{\mathrm{ref}}(\cdot, \omega)$  are uniformly positive in the sense of the definition \eqref{def:UniformlyPosDefinite}.
Thus combining  the  Dirichlet minimization principle \eqref{eq.Dirichletpriciple} and the monotonicity properties  \eqref{eq.monotonicity2}  and \eqref{eq.monotonicity3} yields that for all $V_0\in  H^{\frac{1}{2}}(\partial \Omega)$:
\begin{eqnarray*}\label{eq.monotonicity4}
 \langle \Lambda_{\BGve_{\mathrm{ref}}(\cdot, \omega)}V_0, \overline{V_0}\rangle_{H^{-\frac{1}{2}}(\partial \Omega), H^{\frac{1}{2}}(\partial \Omega)}& \leq  \langle \Lambda_{\BGve(\cdot, \omega)}V_0, \overline{V_0}\rangle_{H^{-\frac{1}{2}}(\partial \Omega), H^{\frac{1}{2}}(\partial \Omega)}, & \forall \omega_0, \omega \in  [\Go_-, \Go_+] \mid \omega_0 \leq \omega,\\
\langle \Lambda_{\BGve(\cdot, \omega)}V_0, \overline{V_0}\rangle_{H^{-\frac{1}{2}}(\partial \Omega), H^{\frac{1}{2}}(\partial \Omega)}& \leq  \langle \Lambda_{\BGve_{\mathrm{ref}}(\cdot, \omega)}V_0, \overline{V_0}\rangle_{H^{-\frac{1}{2}}(\partial \Omega), H^{\frac{1}{2}}(\partial \Omega)},& \forall \omega_0, \omega \in  [\Go_-, \Go_+] \mid \omega\leq \omega_0.
\end{eqnarray*}
Hence, it immediately implies by the  definitions \eqref{eq.FVObis} and \eqref{eq.FVOref}  of $F_{V_0}$ and $F_{V_0, \mathrm{ref}}$ that
\begin{equation*}
F_{V_0}(\Go) \leq F_{V_0,\mathrm{ref},\infty}(\Go) \mbox{ if }  \ \Go_-\leq  \Go\leq \Go_0 \ \mbox{ and } \ F_{V_0}(\Go) \geq  F_{V_0,\mathrm{ref},\infty}(\Go)  \mbox{ if } \   \Go_0 \leq \Go \leq \Go_+.
\end{equation*}
Assuming that $F_{V_0,\mathrm{ref},\infty}>0$, one obtains the other inequality in  \eqref{eq.boundtransp3} and  \eqref{eq.boundtransp4}  by applying  the Corollary \ref{Cor.lossless} to the reference function $F_{V_0, \mathrm{ref}}$ with constant permittivity on the obstacle. 
\end{proof}
In conclusion, all the results  of Section \ref{sec-losslesscase}  about approximate cloaking or cloaking at $\Go_0$ still holds  (with the coefficient $F_{V_0,\mathrm{ref},\infty}$ instead of $F_{V_0,\infty}$).
In particular, if  cloaking holds at $\Go_0$, one can pass to the limit $\eta \searrow 0$ and take $\eta=0$ in the bound \eqref{eq.boundtransp3} and  \eqref{eq.boundtransp4}. When the input field $V_0$ is affine, it  yields the quantitative bounds \eqref{eq.boundtransp1bis} and \eqref{eq.boundtransp2bis} which involved physical quantities showing that cloaking can occurs at frequencies other than $\Go_0$ on the interval $[\omega_-,\omega_+]$. Furthermore, for affine input $V_0$, the results  on approximate cloaking still apply. Namely,  the bounds \eqref{eq.approximatecloacking1} and \eqref{eq.approximatecloacking2} and the inequalities  \eqref{ineq.approximate1} and \eqref{ineq.approximate2} holds also
in this setting. 
\appendix
\section{Appendix: Auxiliary Results}\label{sec:Appendix}

\subsection{Proof of Proposition \ref{pro.DtNDirBVPResults}}\label{sec-annexe}
\begin{proof}
For a given $V_0\in H^{\frac{1}{2}}(\partial \Omega)$, there exists a unique $u_{V_0} \in H^1(\Omega)$ that solves the Dirichlet problem for the Laplacian equation, i.e., $\Delta u_{V_0}=0$ in $\Omega$ and $u_{V_0}=V_0$ on $\partial \Omega$. This defines a bounded linear operator
$P_{\Omega} \in \mathcal{L}(H^{\frac{1}{2}}(\partial \Omega), H^1(\Omega))$ by $P_{\Omega}(V_0)=u_{V_0}$ (see, e.g., Theorem 3.14 on p.\ 46 of \cite{Monk:2003:FEM}). Then, as $\Ba\in M_d( L^{\infty}(\Omega))$, one defines two operators $\mathbb{A}\in \mathcal{L}(H^1_0(\Omega), H^{-1}(\Omega))$ and $\mathbb{B} \in \mathcal{L}(H^{\frac{1}{2}}(\partial \Omega), H^{-1}(\Omega))$  by:
\begin{equation*}
\langle \mathbb{A} v,\overline{w}   \rangle_{H^{-1}(\Omega),H^1_0(\Omega)}=\int_{\Omega} \Ba \nabla v\cdot \overline{\nabla w} \, \mathrm{d}\Bx \ \mbox{ and } \langle  \mathbb{B} V_0 ,\overline{w}   \rangle_{H^{-1}(\Omega),H^1_0(\Omega)}=-\int_{\Omega} \Ba \nabla \mathcal{P}_{\Omega}(V_0) \cdot \overline{\nabla w}\, \mathrm{d} \Bx
\end{equation*}
for all $v, \, w \in H^1_0(\Omega)$ and all $V_0\in H^{\frac{1}{2}}(\partial \Omega)$. By the Cauchy-Schwarz inequality, one shows easily that $\|\mathbb{A}\|\leq \|\Ba\|_{\infty}$ and that $\|B\|\leq \|\Ba\|_{\infty} \|P_{\Omega}\|$.
Moreover, one easily proves that $u=\mathcal{P}_{\Omega}(V_0) +\tilde{u}_{V_0}\in H^1(\Omega)$ solves Dirichlet problem \eqref{def.DirBVPForCoercivityTensors} if and only if $\tilde{u}_{V_0}\in H^1_0(\Omega)$ solves the the linear system:
$
\mathbb{A} \tilde{u}_{V_0}=\mathbb{B} V_0 .
$ The next step is to prove that the operator $\mathbb{A}$ is invertible (as a consequence of the Lax-Milgram lemma, we recall here the proof).

Using the uniformly coercive assumption  \eqref{def.coercivitytensor}  and the Poincar\'e inequality yields $\forall v\in H^1_0(\Omega)$:
\begin{eqnarray*}
&& \|\mathbb{A} v\|_{H^{-1}(\Omega)} \|v\|_{H^1(\Omega)} \geq |\langle \mathbb{A} v,\overline{v}   \rangle_{H^{-1}(\Omega),H^1_0(\Omega)}|  \\
&& \geq \left|\operatorname{Im}\left[e^{\mathrm{i} \gamma}\langle \mathbb{A} v,\overline{v}   \rangle_{H^{-1}(\Omega),H^1_0(\Omega)}\right]\right|=\left|\operatorname{Im}\left(e^{\mathrm{i} \gamma}\int_{\Omega} \Ba \nabla v \cdot \overline{\nabla v } \mathrm{d} \Bx\right)\right|=\int_{\Omega}\mathfrak{I}(e^{\mathrm{i} \gamma} \Ba) \nabla v \cdot \overline{\nabla v}  \,  \mathrm{d} \Bx \\
&& \geq c \|\nabla v\|_{\mathcal{H}}^2\geq C \|v\|_{H^1(\Omega)}^2 \quad \mbox{ with $C=c/2\, \min(1, C(\Omega))$,} 
\end{eqnarray*}
where $C(\Omega)$ is the Poincar\'e constant associated to $\Omega$. Thus,  $\|\mathbb{A}v \|_{H^{-1}(\Omega)} \geq C \,\|v\|_{H^1(\Omega)}, \, \forall v \in H^1_0(\Omega) $. Hence, $\mathbb{A}$ is injective and has closed range.

To prove the surjectivity of $\mathbb{A}$, one introduces the adjoint of $\mathbb{A}$: $\mathbb{A}^\dagger\in \mathcal{L}(H^1_0(\Omega), H^{-1}(\Omega))$ by identifying $H^1_0(\Omega)$ with its bidual (see \cite{McLean:2000:SES}, pp.\ 42-43 for more details) so that $\mathbb{A}^\dagger$ is defined by the following relation
$$
\langle \mathbb{A}^\dagger v, \overline{w}\rangle_{H^{-1}(\Omega),H^1_0(\Omega)} =\langle v,\overline{ \mathbb{A} w}\rangle_{H^1_0(\Omega),H^{-1}(\Omega)} =\overline{\langle \mathbb{A} w ,\overline{v}\rangle}_{H^{-1}(\Omega),H^1_0(\Omega)} , \quad  \forall \, v,\, w\in H^1_0(\Omega).
$$
 Using the same argument as above, one shows also that $\|\mathbb{A}^\dagger v \|_{H^{-1}(\Omega)}\geq C \|v\|_{H^1(\Omega)}, \, \forall v \in H^1_0(\Omega) $.  Thus $\mathbb{A}^{\dagger}$ is injective and therefore $\mathbb{A}$ is surjective. One concludes that $\mathbb{A}$ is invertible and $\mathbb{A}^{-1}\in \mathcal{L}(H^{-1}(\Omega), H^{1}_0(\Omega))$. 
Therefore, the solution $u\in H^1(\Omega)$ of the Dirichlet problem \eqref{def.DirBVPForCoercivityTensors} is unique and given by $u=u_{V_0}+\tilde{u}_{V_0}=(\mathbb{A}^{-1}B+P_{\Omega}) V_0$.

Next, one easily shows that the ``displacement operator" given by $\mathbb{D}:   u \mapsto  \Ba \nabla u$ belongs to $ \mathcal{L}(\mathcal{V}, H_{\mathrm{div}}(\Omega))$, where $\mathcal{V}$ is the  closed subspace  of $H^1(\Omega)$ defined by $\mathcal{V}:=\{u\in H^1(\Omega) \mid \,\nabla\cdot \Ba\, \nabla u=0 \}$.  Furthermore, it satisfies $\|\mathbb{D}\|\leq \|\Ba\|_{\infty}$.  Finally, by defining the normal trace operator  $\gamma_{\Bn}: \Bv \mapsto \Bv \cdot \Bn \in \mathcal{L}(H_{div}(\Omega),H^{-\frac{1}{2}}(\partial \Omega))$ (see, e.g., Theorem 2.5 on p.\ 27 of \cite{Girault:1986:FNS}),  one gets that $\Lambda_{a}=\gamma_{\Bn} \, \mathbb{D}  (A^{-1}B+P_{\Omega}) \in \mathcal{L}(H^{\frac{1}{2}}(\partial \Omega), H^{-\frac{1}{2}}(\partial \Omega))$ is well-defined by the relation \eqref{def.DtNOpDirBVPForCoercivityTensors}. 
\end{proof}

\subsection{On the strong convergence of inverse operators}\label{sec:StrConvInvOps}
We recall the following definition from \cite[Chap.\ VII, pp.\ 182--183]{Reed:1980:FA}.
\begin{Def}
Let $E, \, F$ be two Banach spaces, $(\mathbb{T}_n)_{n\in \mathbb{N}}$ be a sequence of operators in $\mathcal{L}(E,F)$, and $\mathbb{T}\in\mathcal{L}(E,F)$. One says that $(\mathbb{T}_n)_{n\in \mathbb{N}}$ converges strongly to $\mathbb{T}\in\mathcal{L}(E,F)$ if only if
$$
\forall u \in E, \quad \|\mathbb{T}_n u- \mathbb{T} u\|_{F}\to 0 \ \mbox{ as } \ n\to +\infty.$$
\end{Def}
\begin{Lem}\label{Lem.strongconv}
Let $E, \, F$ be two Banach spaces,  $\mathbb{T}$ be an invertible operator in $\mathcal{L}(E,F)$, and $(\mathbb{T}_n)_{n\in \mathbb{N}}$ be a sequence of  invertible operators in $\mathcal{L}(E,F)$  such that 
\begin{equation}\label{eq.uniformbound}
\displaystyle\sup_{n\in \mathbb{N}}\|\mathbb{T}_n^{-1}\|_{\mathcal{L}(F,E)}<+\infty.
\end{equation}
Then, if $(\mathbb{T}_n)_{n\in \mathbb{N}}$ converges strongly to $\mathbb{T}$,  one has $(\mathbb{T}_n^{-1})_{n\in \mathbb{N}}$ converges strongly to $\mathbb{T}^{-1}$.
\end{Lem}
\begin{proof}
We first point out that since the operators $\mathbb{T}_n \in \mathcal{L}(E,F)$ are  invertible then by the inverse mapping theorem (see Theorem III.11 page 83 of \cite{Reed:1980:FA}) their inverses $\mathbb{T}_n^{-1}$ are bounded and thus belong to $\mathcal{L}(F,E)$. In addition, we assume here with \eqref{eq.uniformbound} that these inverses are uniformly bounded with respect to $n$. Then, one uses the following  identity: 
\begin{equation}\label{eq.indenditycont}
\mathbb{T}_{n}^{-1}v  -\mathbb{T}^{-1}v  =\mathbb{T}_{n}^{-1} \, \big( (\mathbb{T} -\mathbb{T}_{n}) \ \mathbb{T}^{-1}\, v\big), \quad \forall v\in F.
\end{equation}
Thus, it follows  with \eqref{eq.indenditycont} and the strong convergence of $(\mathbb{T}_n)_{n\in \mathbb{N}}$ to $\mathbb{T}$ that there exists $C>0$ such that
$$
\|\mathbb{T}_{n}^{-1}v  -\mathbb{T}^{-1}v\|_E \leq C\, \|(\mathbb{T} -\mathbb{T}_{n}) \ \mathbb{T}^{-1}\, v\|_E\to 0 \ \mbox{ as } n\to+\infty.
$$

\end{proof}

\subsection{Herglotz functions in Banach spaces}\label{sec:HergFuncBanachSpaces}
In this section, we generalize the discussion from Sec.\ \ref{subsec:ReviewHerglotzStieltjesFuncs} on scalar- and matrix-valued Herglotz functions to operator-valued Herglotz functions on Banach spaces.

First, recall the Hilbert space definition for such functions along with the well-known lemma (see, e.g., \cite{Gesztesy:2001:SAO}) that generalizes Def.\ \ref{def.matHerg} and Lemma \ref{lem:matHergViaQuadForms}, stated  for matrix-valued Herglotz functions, to the setting of operator-valued Herglotz functions.
\begin{Def}\label{def:OperatorValuedHerglotzFuncHilbertSpace}
Let $\mathcal{H}$ be a complex Hilbert space. An analytic function $h:\mathbb{C}^+\rightarrow \mathcal{L}(\mathcal{H})$ is a operator-valued Herglotz function if
\begin{gather*}
    \mathfrak{I}[h(z)]\geq 0, \ \forall z\in \mathbb{C}^+.
\end{gather*}
\end{Def}
\begin{Lem}\label{lem:WeakOperHerglotzImpliesStrongOpHerglotz}
    Let $\mathcal{H}$ be a complex Hilbert space with inner product $(\cdot, \cdot)$. An operator-valued function $h:\mathbb{C}^+\rightarrow \mathcal{L}(\mathcal{H})$ is a Herglotz function if and only if $z\mapsto (h(z)u, u)$ is a Herglotz function for every $u\in \mathcal{H}$.
\end{Lem}

 Now, we give a generalization, from Hilbert spaces to Banach spaces, of operator-valued Herglotz functions in Definition \ref{def:OperatorValuedHerglotzFuncHilbertSpace} and Lemma \ref{lem:WeakOperHerglotzImpliesStrongOpHerglotz}, which is inspired by the work of D.\ Alpay, O.\ Timoshenko, and D.\ Volok \cite{Alpay:2007:CFB, Alpay:2009:CFI} on a similar generalization for Carathéodory functions. Our motivation is due to the fact that the DtN map as a function of frequency is such an operator-valued Herglotz function  with the Banach space and its (conjugate) dual space
\begin{gather}\label{eq.dualspaces}
    \mathcal{B}=H^{\frac{1}{2}}(\partial \Omega), \ \mathcal{B}^{\dagger}=H^{-\frac{1}{2}}(\partial \Omega)_c=\{\ell_c:\ell \in H^{-\frac{1}{2}}(\partial \Omega)\},
\end{gather}
where we define the inverse conjugation maps by
\begin{gather*}
    H^{-\frac{1}{2}}(\partial \Omega)\ni \ell\mapsto \ell_c\in H^{-\frac{1}{2}}(\partial \Omega)_c,\ \ell_c(v)=\ell(\overline {v}), \ \forall \ell\in H^{-\frac{1}{2}}(\partial \Omega), v\in H^{\frac{1}{2}}(\partial \Omega),\\
  H^{-\frac{1}{2}}(\partial \Omega)_c\ni f\mapsto f_c\in H^{-\frac{1}{2}}(\partial \Omega),\ f_c(v)=f(\overline{v}), \ \forall f\in H^{-\frac{1}{2}}(\partial \Omega)_c, v\in H^{\frac{1}{2}}(\partial \Omega)
\end{gather*}
so that the relationship between the dual pairing of $\mathcal{B}$ and $\mathcal{B}^{\dagger}$ and the dual pairing of $H^{\frac{1}{2}}(\partial \Omega)$ and $H^{-\frac{1}{2}}(\partial \Omega)$ in this case becomes
\begin{gather*}
[f,v]_{\mathcal{B}}=f(v)=f_c(\overline{v})=\langle f_c,\overline{v} \rangle_{H^{-\frac{1}{2}}(\partial \Omega), H^{\frac{1}{2}}(\partial \Omega)}, \ \forall f\in \mathcal{B}^{\dagger}, v\in \mathcal{B},\\
[\ell_c,v]_{\mathcal{B}}=\ell_c(v)=\ell(\overline{v})=\langle \ell,\overline{v} \rangle_{H^{-\frac{1}{2}}(\partial \Omega), H^{\frac{1}{2}}(\partial \Omega)}, \ \forall \ell\in H^{-\frac{1}{2}}(\partial \Omega), v\in H^{\frac{1}{2}}(\partial \Omega).
\end{gather*}
\begin{Def}\label{def:PrelimOperatorValuedHerglotzFuncBanachSpace}
    Let $\mathcal{B}$ be a complex Banach space, $\mathcal{B}^\dagger$ the space of antilinear bounded functions (i.e., conjugate dual space), and denote the duality pairing of $\mathcal{B}$ and $\mathcal{B}^\dagger$ by 
    \begin{gather*}
        [f,b ]_{\mathcal{B}}:=f(b), \ \forall b\in \mathcal{B}, f\in \mathcal{B}^{\dagger}.
    \end{gather*}
    For an element $A\in \mathcal{L}(\mathcal{B},\mathcal{B}^\dagger)$, define its adjoint $A^{\times}\in \mathcal{L}(\mathcal{B}^{\dagger\dagger},\mathcal{B}^\dagger)$ by
    \begin{gather*}
        A^{\times}(g)(b)=g(Ab), \ \forall g\in \mathcal{B}^{\dagger\dagger}, b\in B;
    \end{gather*}
    let $C\in \mathcal{L}(\mathcal{B}, \mathcal{B}^{\dagger\dagger})$ denote the canonical embedding, i.e., 
    \begin{gather*}
        C(b)(g)=\overline{g(b)}, \ \forall b\in \mathcal{B}, g\in \mathcal{B}^{\dagger};
    \end{gather*}
    its conjugate-adjoint $A^\dagger\in  \mathcal{L}(\mathcal{B},\mathcal{B}^\dagger)$ by
    \begin{gather*}
        A^\dagger:=A^{\times} \circ C,
    \end{gather*}
    which satisfies
    \begin{gather*}
        [ A^\dagger b,c ]_{\mathcal{B}}=(A^\dagger b)(c)=(A^\times C(b))(c)=C(b)(Ac)=\overline{(Ac)(b)}=\overline{[ Ac,b ]_{\mathcal{B}}};
    \end{gather*}
    the real $\mathfrak{R}(A)$ and imaginary $\mathfrak{I}(A)$ parts of $A$ are defined by
    \begin{gather*}
        \mathfrak{R}(A):=\frac{1}{2}(A+A^\dagger), \ \mathfrak{I}(A):=\frac{1}{2\mathrm{i}}(A-A^\dagger);
    \end{gather*}
    the operator $A$ is said to be positive semidefinite, which we denote by $A\geq 0$, if
    \begin{gather*}
      [ A b,b ]_{\mathcal{B}}\geq 0, \ \forall b\in \mathcal{B}.
    \end{gather*}
\end{Def}

\begin{Def}\label{def:OperatorValuedHerglotzFuncBanachSpace}
    An analytic function $h:\mathbb{C}^+\rightarrow \mathcal{L}(\mathcal{B},\mathcal{B}^\dagger)$ is an operator-valued Herglotz function if 
    \begin{gather*}
    \mathfrak{I}[h(z)]\geq 0, \ \forall z\in \mathbb{C}^+.
    \end{gather*}
\end{Def}

\begin{Lem}\label{lem:WeakOperHerglotzBanachSpaceImpliesStrongOpHerglotz}
    Let $\mathcal{B}$ be a complex Banach space. An operator-valued function $h:\mathbb{C}^+\rightarrow \mathcal{L}(\mathcal{B}, \mathcal{B}^\dagger)$ is a Herglotz function if and only if $z\mapsto [h(z)u, u]_{\mathcal{B}}$ is a Herglotz function for every $u\in \mathcal{B}$.
\end{Lem}
\begin{proof}
    ($\Rightarrow$): Obvious. ($\Leftarrow$): By hypothesis and the polarization identity of sesquilinear forms on complex vector spaces \cite[p.\ 25]{Teschl:2014:MMQ}, it follows that $z\mapsto [h(z)u, v]_{\mathcal{B}}$ is an analytic function on $\mathbb{C}^+$ for each $u,v\in \mathcal{B}$, hence $h:\mathbb{C}^+\rightarrow \mathcal{L}(\mathcal{B}, \mathcal{B}^\dagger)$ is analytic by \cite[Theorem 3.12, pp.\ 152-153]{Kato:1995:PTO}. Next, for any $z\in \mathbb{C}^+, u\in \mathcal{B},$ we have
    \begin{gather*}
        0\leq \operatorname{Im}[h(z)u, u]_{\mathcal{B}}=\frac{1}{2\mathrm{i}}([h(z)u, u]_{\mathcal{B}}-\overline{[h(z)u, u]_{\mathcal{B}}})=\frac{1}{2\mathrm{i}}\mathrm{i}([h(z)u, u]_{\mathcal{B}}-[h(z)^{\dagger}u, u]_{\mathcal{B}})\\
        =\left[\frac{1}{2\mathrm{i}}(h(z)-h(z)^{\dagger})u, u\right]_{\mathcal{B}}
        =[\mathfrak{I}[h(z)]u, u]_{\mathcal{B}}
    \end{gather*}
    which implies $\mathfrak{I}[h(z)]\geq 0$. This completes the proof of the lemma.
\end{proof}

Combining Lemma \ref{lem:WeakOperHerglotzBanachSpaceImpliesStrongOpHerglotz}, Proposition \ref{Pro:DtNMapQuadFormIsHerglotz} and \eqref{eq.dualspaces} immediately yields a proof of the following main result of this appendix. 
\begin{Pro}\label{pro.DtNMapIsHerglotzFunc}
Under the assumptions $\textnormal{H1}, \, \textnormal{H5},$ and $\textnormal{H6}$, the function
$ \Go \in \bbC^+\mapsto \Lambda_{\omega \varepsilon(\cdot, \Go) }\in \mathcal{L}(H^{1/2}(\partial \Omega), H^{-1/2}(\partial \Omega))$ is an operator-valued Herglotz function.
\end{Pro}

Note that in our setting for the DtN operator-valued Herglotz function,  $\mathcal{B}=H^{1/2}(\partial \Omega)$ is a reflexive Banach space (since it is a Hilbert space) and so $\mathcal{B}$ can be identified with its double dual $\mathcal{B}^{\dagger, \dagger}$ via $C$ which is an antilinear isometric isomorphism of Banach spaces. In fact, more can be said since $(H^{1/2}(\partial \Omega), L^2(\partial \Omega), H^{- 1/2}(\partial \Omega))$ forms  a Gelfand triple also referred to as a rigged Hilbert space; see \cite{Coj:2015:TCE,Skrepek2025:QUA} for the definition and properties of a Gelfand triple. For examples of its applications to PDEs, we refer to \cite[pp.~136–137]{Brezis:2011:FAS}, and for its use in the context of boundary Sobolev spaces, see \cite[p.~8]{Girault:1986:FNS}. We  note also  that one could go further by considering the more general framework of operator-valued Herglotz functions associated with quasi-Gelfand triples as defined in \cite{Skrepek2025:QUA}; however, a detailed investigation of this direction lies beyond the scope of our paper. 

Finally, its worth pointing out that there is also an operator theory approach to the DtN map, in which Gelfand triples appear, namely, boundary triples \cite{Behrndt:2014:EDO, Behrndt:2020:BVP} that has been investigated in connection to inverse problems for elliptic PDEs \cite{Brown:2008:BTM, Behrndt:2012:ANP}. It would be of interest to see if the results of our paper could be applied in that context and given the importance the operator-valued Herglotz functions play in that context (see \cite{Behrndt:2020:BVP}), we speculate that this appendix will be of use.

\section*{Acknowledgement}
The three authors are grateful to the Institute for Mathematics and its Applications in Minneapolis for hosting their visit during the program on Mathematics and Optics in Fall 2016, where this project was initiated. They are also grateful to the Banff International Research Station and the Centre International de Rencontres Mathématiques – Marseille Luminy for hosting their visits in 2019 and 2022, respectively, during the conferences “Herglotz–Nevanlinna Theory Applied to Passive, Causal and Active Systems” and “Herglotz–Nevanlinna Functions and their Applications to Dispersive Systems and Composite Materials,” where this work was further developed. GWM and AW are also respectively very thankful to Aix-Marseille University  and the Fresnel Institute for their one-month positions as Invited Professor and Invited Maître de conférences in 2022 and 2024 where significant progress was made on this work. GWM is grateful to the National Science Foundation for support through grants DMS-1814854 and DMS-2107926. AW is grateful to the Air Force Office of Scientific Research through grant FA9550-15-1-0086, to the National Science Foundation for support through grant DMS-2410678, and to the Simons Foundation for the support through grant MPS-TSM-00002799.

\bibliographystyle{abbrv}
\bibliography{CloakDtNMapEMQS}
\end{document}